%% file: main.tex
\documentclass[twoside,leqno]{article}
\usepackage[letterpaper]{geometry}
\usepackage{siamproceedings}

\usepackage[T1]{fontenc}
\usepackage{amsfonts}
\usepackage{graphicx}
\usepackage{epstopdf}
\usepackage[shortlabels]{enumitem}
\usepackage{algorithmic}
\ifpdf
  \DeclareGraphicsExtensions{.eps,.pdf,.png,.jpg}
\else
  \DeclareGraphicsExtensions{.eps}
\fi

\newsiamremark{remark}{Remark}
\newsiamremark{hypothesis}{Hypothesis}
\crefname{hypothesis}{Hypothesis}{Hypotheses}
\newsiamthm{claim}{Claim}
\usepackage{amsopn}

\usepackage{amssymb}
\usepackage{xcolor}
\usepackage{xifthen}
\allowdisplaybreaks
\def\^#1{\mathbb{#1}}
\def\*#1{\mathbf{#1}}
\def\+#1{\mathcal{#1}}
\def\-#1{\mathrm{#1}}
\def\=#1{\boldsymbol{#1}}
\def\!#1{\mathtt{#1}}
\newcommand{\abs}[1]{\left|#1\right|}

\newcommand{\id}[1]{\ensuremath{\mathbf{1}}\left[#1\right]}

\renewcommand{\Pr}[2][]{\ifthenelse{\isempty{#1}}{\mathop{\mathbf{Pr}}\left[#2\right]}{\mathop{\mathbf{Pr}}_{#1}\left[#2\right]}}
\newcommand{\E}[2][]{\ifthenelse{\isempty{#1}}{\mathbf{E}\left[#2\right]}{\mathbf{E}_{#1}\left[#2\right]}}

\newcommand{\diag}{\mathsf{diag}}
\newcommand{\perman}{\mathsf{per}}
\newcommand{\norm}[1]{\left\Vert#1\right\Vert}
\newcommand{\poly}{\textnormal{\textsf{poly}}}

\begin{document}
\newcommand\relatedversion{}
\renewcommand\relatedversion{\thanks{The full version of the paper can be accessed at \protect\url{https://arxiv.org/abs/2604.03787}}} % Replace URL with link to full paper or comment out this line

\title{\Large Dimension-Free Sinkhorn for Optimal Transport:\\ Bypassing Extremal Costs\relatedversion{}}
\author{Kun He\thanks{DEKE lab, Renmin University of China, China (\email{hekun.threebody@foxmail.com}).}}
\date{}
\maketitle
\fancyfoot[R]{\scriptsize{Copyright \textcopyright\ 2027 by SIAM\\
Unauthorized reproduction of this article is prohibited}}

\begin{abstract}
The Sinkhorn--Knopp (SK) algorithm is a cornerstone method for matrix scaling and entropically regularized optimal transport (EOT). Despite its empirical efficiency, existing theoretical guarantees to achieve a target marginal accuracy $\varepsilon$ deteriorate severely in the presence of extreme entries, bottlenecked either by the global maximum regularized cost $\eta\|C\|_\infty$ (where $\eta$ is the regularization parameter and $C$ the cost matrix) or the matrix's minimum-to-maximum entry ratio $\nu$. This reveals a limitation of existing complexity bounds: they may be overly sensitive to rare extreme entries that do not dominate the overall matrix structure. In this paper, we identify a broad structural regime in which the dependence on the magnitude of global extreme entries can be circumvented. For EOT, we introduce the novel concept of \textit{well-boundedness}, a local mass property that captures the row- and column-wise bulk structure of the underlying matrix while maintaining robustness against scattered entrywise outliers. We prove that, governed by this structural condition, SK reaches marginal accuracy $\varepsilon$ for an EOT problem of dimension $n$ in $O(\log n + \log(1/\varepsilon))$ iterations. Crucially, this complexity bound is completely independent of the maximum regularized cost $\eta\|C\|_\infty$. Furthermore, we show that a computationally negligible pre-scaling step eliminates this dimensional dependence, reducing the iteration complexity to a strictly dimension-free $O(\log(1/\varepsilon))$. Beyond EOT, we establish a sharp phase transition for general $(\boldsymbol{u},\boldsymbol{v})$-scaling governed by a critical matrix density threshold. We prove that when a matrix's density exceeds this threshold, the iteration complexity is strictly independent of $\nu$. Conversely, when the density falls below this threshold, the dependence on $\nu$ becomes unavoidable; in this sub-critical regime, we construct instances where SK requires $\Omega(n/\varepsilon)$ iterations. Technically, our analysis relies on two synergistic techniques. First, a novel discretization framework reduces general $(\boldsymbol{u},\boldsymbol{v})$-scaling for rectangular matrices to uniform $(\boldsymbol{1},\boldsymbol{1})$-scaling, unlocking square-matrix combinatorial tools such as the permanent. Second, we establish a strengthened \textit{structural stability} property, demonstrating that matrix entries and their permanent bounds are robustly controlled as the scaling orbit approaches the doubly stochastic limit. Working in tandem, these two techniques allow us to tightly track the underlying matrix dynamics and rigorously prove the fast convergence of the SK algorithm.

\end{abstract}

\input{Introduction.tex}
\input{technique.tex}
\input{technicaltools.tex}
\input{Reduction.tex}
\input{Permanent.tex}
\input{target.tex}
\input{Expansity.tex}
\input{Tightness.tex}

\appendix
\input{MissingProofs.tex}

\section*{Acknowledgments.}
This work was supported in part by the National Natural Science Foundation of China Grants No. 62472430.

\bibliographystyle{siamplain}
\input{references.tex}

\end{document}

%% file: Introduction.tex
\section{Introduction.}\label{sec-intro}

The problem of matrix scaling seeks strictly positive diagonal matrices $X$ and $Y$ such that, given a nonnegative matrix $A$, the scaled matrix $XAY$ attains prescribed row and column sums. 
This fundamental primitive arises frequently across theory and practice. 
For instance, it serves as a preconditioner to improve the numerical stability of linear systems \cite{osborne1960pre}; acts as the core mechanism for enforcing marginal constraints in optimal transport \cite{altschuler2017}; and remains a cornerstone technique in statistical normalization \cite{deming1940least}, image processing \cite{rubner2000earth}, and numerous other domains \cite{idel2016review}.

A venerable approach for matrix scaling
is the Sinkhorn--Knopp (SK) algorithm \cite{sinkhorn1967diagonal,sinkhorn1967concerning}, also known as RAS \cite{bacharach1965estimating} or iterative proportional fitting \cite{ruschendorf1995convergence}. 
It alternates row and column normalizations, producing a sequence of scaled matrices whose marginals progressively approach the prescribed targets.
A central question is how fast SK converges and how this rate depends on the structure of the input matrix.
Broadly speaking, prior analyses of SK can be organized into three
complementary directions. 
The first direction studies nonasymptotic, finite-\(\varepsilon\)
iteration complexity: given an error metric and a target accuracy
\(\varepsilon\), how many iterations are sufficient to drive the error below
\(\varepsilon\)? These results provide explicit stopping guarantees under
concrete norms such as \(\ell_1\) and \(\ell_2\), and are therefore directly
relevant for understanding the computational cost of SK in practice. 
The second direction studies global geometric convergence through
Perron--Frobenius theory and Hilbert's projective metric. 
These results provide
nonasymptotic contraction guarantees from arbitrary initialization, with rates
controlled by quantities such as the projective diameter of the kernel.
However, the resulting bounds are typically expressed in projective metrics rather than directly in terms of standard marginal-error measures.
The third direction studies local asymptotic convergence near the
limiting doubly stochastic scaling. By linearizing the SK iteration around the
fixed point, these works characterize the local contraction factor through
spectral properties of the associated Jacobian or limiting matrix.
Despite substantial progress, sharp finite-\(\varepsilon\) bounds under
standard marginal-error norms remain poorly understood for general
nonnegative inputs.

On the finite-\(\varepsilon\) side, we call \(A \in \mathbb{R}_{\ge 0}^{n \times m}\) \((\boldsymbol{u},\boldsymbol{v})\)-scalable if there exist positive diagonal matrices \(X,Y\) such that \(XAY\) has row sums \(\boldsymbol{u}\) and column sums \(\boldsymbol{v}\). For such scalable nonnegative inputs, SK computes, for any \(\varepsilon>0\), an \(\varepsilon\)-approximate scaling with \(\ell_{1}\)-error at most \(\varepsilon\) in
$t = O\Bigl(h^2 \varepsilon^{-2}\log\bigl(\Delta\,\mu/\nu\bigr)\Bigr)$
iterations, and an \(\varepsilon\)-approximate scaling with \(\ell_{2}\)-error at most \(\varepsilon\) in $t = O\Bigl(\mu\,h\,\log\bigl(\Delta\,\mu/\nu\bigr)\bigl(\varepsilon^{-1}+\varepsilon^{-2}\bigr)\Bigr)$
iterations \cite{chakrabarty2021better}. Here \(h\) is the sum of the target row entries (equivalently column entries), \(\mu\) is the largest target entry, \(\Delta\) is the maximum number of nonzeros in any column of \(A\), and \(\nu\) is the ratio of the smallest positive entry of \(A\) to its largest. In the special case of \((\boldsymbol{1},\boldsymbol{1})\)-scaling (i.e., scaling to a doubly stochastic matrix), these bounds specialize to \(O\bigl(n^2\varepsilon^{-2}\log(\Delta/\nu)\bigr)\) 
for \(\ell_{1}\)-error and \(O\bigl(n\log(\Delta/\nu)(\varepsilon^{-1}+\varepsilon^{-2})\bigr)\) for \(\ell_{2}\)-error. 
Under the stronger assumption that the input is strictly positive, earlier work obtained faster finite-\(\varepsilon\) \(\ell_2\)-bounds of order \(O\bigl(\sqrt{n}\log(1/\nu)/\varepsilon\bigr)\) for the doubly stochastic case, later extended to general \((\boldsymbol{u},\boldsymbol{v})\)-scaling \cite{kalantari1993rate,kalantari2008complexity}.

Also within the finite-\(\varepsilon\) line, a recent result establishes a density-based phase transition for the SK algorithm applied to the $(\={1},\={1})$-scaling problem \cite{he2025phase}. For an \(n\times n\) matrix, its normalized version is obtained by dividing each entry by the largest entry in the matrix. We say that a normalized matrix has density \(\gamma\) if there exists a constant \(\rho>0\) such that one row or column has exactly \(\lceil \gamma n\rceil\) entries of value at least \(\rho\), while every other row and column has at least \(\lceil \gamma n\rceil\) such entries. When \(\gamma>1/2\), SK reaches \(\ell_1\)-error at most \(\varepsilon\) in \(O(\log n-\log\varepsilon)\) iterations. In contrast, for every \(\gamma<1/2\), there exist matrices of density \(\gamma\) for which SK requires \(\Omega(n/\varepsilon)\) iterations to achieve \(\ell_1\)-error at most \(\varepsilon\), and \(\Omega(\sqrt{n}/\varepsilon)\) iterations to achieve \(\ell_2\)-error at most \(\varepsilon\) for sufficiently small \(\varepsilon\).

A complementary line of work studies geometric contraction properties of SK
rather than finite-\(\varepsilon\) iteration complexity under standard
marginal-error metrics. On the global side, Birkhoff's contraction theory for
Hilbert's projective metric \cite{birkhoff1957extensions} provides the
underlying framework, and Franklin and Lorenz used it to prove geometric
convergence of SK's alternating normalization for strictly positive matrices,
with contraction factors controlled by quantities derived from the input
kernel \cite{franklin1989scaling}. On the local side, Soules proved linear
convergence in the doubly stochastic setting under total support by analyzing
the Jacobian at the limiting point \cite{soules1991rate}. Knight later made the
asymptotic factor explicit: if \(A\) is fully indecomposable and the SK
iterates converge to a doubly stochastic limit \(P\), then the scaling vectors
converge locally at an asymptotic rate \((\sigma_2(P))^2\), where
\(\sigma_2(P)\) is the second singular value of \(P\)
\cite{knight2008sinkhorn}. These results clarify important contraction
mechanisms of SK, but they are conceptually distinct from finite-\(\varepsilon\)
bounds stated directly in terms of \(\ell_1\) or \(\ell_2\) marginal error after
a prescribed number of iterations.

% A different line of work studies linear convergence itself rather than explicit iteration counts to reach a prescribed accuracy. On the global side, Birkhoff's contraction theory for Hilbert's projective metric \cite{birkhoff1957extensions} provides the underlying framework, and Franklin and Lorenz used it to prove geometric convergence of SK's alternating normalization for strictly positive matrices, with a contraction factor estimated explicitly from the input data \cite{franklin1989scaling}. On the local side, Soules proved linear convergence in the doubly stochastic setting under total support by analyzing the Jacobian at the limiting point \cite{soules1991rate}. Knight later made the asymptotic factor explicit: if \(A\) is fully indecomposable and the SK iterates converge to a doubly stochastic limit \(P\), then the scaling vectors contract asymptotically by a factor \((\sigma_2(P))^2\), where \(\sigma_2(P)\) is the second singular value of \(P\) \cite{knight2008sinkhorn}. These results clarify the asymptotic contraction law of SK, but they are conceptually distinct from finite-\(\varepsilon\) bounds stated directly in terms of \(\ell_1\) or \(\ell_2\) marginal error after a prescribed number of iterations.

On the algorithmic side, the SK algorithm is only one representative method for matrix scaling. A variety of alternative routes have been developed, often trading the simplicity of SK iterations for stronger global complexity guarantees under various condition measures. A prominent theoretical milestone is the work of \cite{Linial2000Deterministic}, which provides the first deterministic, strongly polynomial-time algorithmic framework for matrix scaling and leverages it to approximate the permanent.
More recently, fast optimization-based approaches have yielded improved running-time guarantees in terms of the number of nonzeros, $\ell$. 
For instance, \cite{cohen2017matrix} develop box-constrained Newton-type algorithms for matrix scaling and balancing with running time $\widetilde{O}(\ell\log\kappa\log^2(1/\varepsilon))$, where $\kappa\geq 2$ bounds the ratio between the largest and smallest diagonal entries within each optimal scaling matrix. 
Concurrently, \cite{allen2017much} obtain an $\widetilde{O}((\ell+n^{4/3})\log\kappa)$-time algorithm for doubly stochastic scaling. Thus, both bounds depend logarithmically on the conditioning of the optimal scaling. 
When the optimal scaling factors are polynomially bounded, this dependence is absorbed into the polylogarithmic factors suppressed by $\widetilde{O}$, yielding the commonly stated $\widetilde{O}(\ell+n^{4/3})$ bound for \cite{allen2017much}. In addition, \cite{cohen2017matrix} provide a separate interior-point algorithm with running time $\widetilde{O}(\ell^{3/2}\log(1/\varepsilon))$.
\cite{kwok2019spectral} demonstrates that if the input instance exhibits a spectral gap, a natural gradient flow and its gradient descent discretization enjoy linear convergence, leading to sharper guarantees for structured instances like expander graphs. 
Finally, recent breakthroughs in almost-linear time maximum flow, minimum-cost flow, and broader convex flow objectives provide a powerful reduction-based toolkit \cite{chen2025maximum}. This framework explicitly encompasses matrix scaling and entropy-regularized optimal transport among its applications, further expanding the landscape of fast scaling algorithms beyond SK-style alternating normalizations.

\paragraph{Entropically regularized optimal transport.}
A canonical application of the SK algorithm is entropically regularized optimal transport (EOT). In EOT one solves
$$
\min_{P\in U(\=u,\=v)}\ \langle C,P\rangle-\eta^{-1}\,H(P),
$$
where $U(\=u,\=v)$ is the transportation polytope of nonnegative matrices with row sums $\=u$ and column sums $\=v$; $C$ is the cost matrix with entries $C_{ij}$ measuring the distance between source $i$ and target $j$; and $H(P)$ denotes the Shannon entropy of $P$ (e.g., $H(P)=-\sum_{i,j} P_{ij}(\log P_{ij}-1)$), with $\eta>0$ the regularization parameter.
The unique optimizer has the Gibbs–scaling form
\begin{align}\label{eq-solution-eot}
P^\ast=\operatorname{Diag}(\=X)\,K\,\operatorname{Diag}(\=Y),\qquad \text{with  }
K\triangleq \exp(-\eta C),
\end{align}
where the exponential function is applied element-wise, and $\=X,\=Y$ are some positive scaling vectors. Hence one can recover $P^\ast$ by scaling the kernel $K$ with the SK iterations. 

For a fixed regularization parameter $\eta$, one should first distinguish the complexity of computing the regularized optimizer $P^\ast$ itself from the end-to-end complexity of approximating unregularized OT. In the former sense, two rather different regimes are known. When the entries of $K$ admit an effective positive lower bound, the classical projective-metric analysis implies geometric convergence, so the number of iterations required to reach an $\ell_{1}$-projection accuracy $\varepsilon$ is logarithmic in $\varepsilon^{-1}$, with constants determined by the data through the lower bound on $K$ and the associated projective diameter \cite{franklin1989scaling}. A different line of work avoids writing the complexity directly in terms of $\min_{ij}K_{ij}$ and instead controls the decrease of KL-type potentials or of the EOT dual objective. 
In this view, for fixed $\eta$, SK reaches projection accuracy $\varepsilon$ in a number of iterations of order $O(R/\varepsilon)$, where $R$ denotes a bound on the dual amplitudes and typically scales like $\eta |C|_{\infty}$ up to logarithmic marginal terms \cite{dvurechensky2018computational,luo2023improved}. Although many results of this second type are embedded in analyses of additive-$\varepsilon$ approximation to unregularized OT, their inner-loop statements should first be read as genuine fixed-$\eta$ guarantees for solving the EOT projection problem.

A different question, and the one emphasized in most modern OT complexity papers, is how to choose $\eta$ and how accurately the regularized problem must be solved in order to obtain an additive $\varepsilon$-approximation to the unregularized OT cost; here $\varepsilon$ denotes cost accuracy rather than projection accuracy. The key issue is to balance the regularization bias, the inexact solution of the EOT subproblem, and the final rounding back to the transportation polytope. Starting from the near-linear-time regularize-then-round framework of \cite{altschuler2017}, SK-based bounds improved from $\widetilde O(n^{2}|C|_{\infty}^{3}\varepsilon^{-3})$ to $\widetilde O(n^{2}|C|_{\infty}^{2}\varepsilon^{-2})$ in \cite{dvurechensky2018computational}. For Greenkhorn, the same framework replaces full row/column sweeps by greedy single-coordinate updates, and the original $\widetilde O(n^{2}|C|_{\infty}^{3}\varepsilon^{-3})$ guarantee was sharpened to $\widetilde O(n^{2}|C|_{\infty}^{2}\varepsilon^{-2})$ in \cite{lin2019efficient,lin2022efficiency}; moreover, \cite{luo2023improved} showed that the same $\varepsilon^{-2}$ order already holds for vanilla SK and vanilla Greenkhorn, without modifying the marginals. Thus the literature on approximating unregularized OT by EOT is conceptually different from the fixed-$\eta$ theory above: the former optimizes the interplay among $\eta$, inner-loop accuracy, and rounding error, whereas the latter concerns only the cost of scaling a prescribed Gibbs kernel.

A complementary approach improves the convergence behavior of SK by
changing the regularization path rather than analyzing a fixed regularization
parameter. Chen et al.~\cite{chen2023exponential} propose a regularization
scheduling strategy that gradually increases the parameter \(\eta\) and warm
starts each SK computation from the solution of the previous
regularization level. They prove exponential convergence guarantees for this
scheduled procedure, which avoids the slow convergence associated with solving
EOT problems with large $\eta$ directly. This approach is fundamentally
different from fixed-\(\eta\) analyses: it accelerates SK by modifying the
sequence of kernels \(K_\eta=\exp(-\eta C)\), whereas fixed-\(\eta\) complexity
results characterize the convergence of scaling a prescribed kernel.

Beyond SK-type scaling, a prominent alternative line of research tackles EOT via smooth dual or saddle-point formulations using generic accelerated first-order methods. For instance, accelerated gradient approaches were specifically designed to overcome the $O(\varepsilon^{-2})$ accuracy dependence inherent in early SK-based OT approximations \cite{dvurechensky2018computational}. Subsequent works developed accelerated primal-dual mirror-descent variants, tightened the complexity bounds for accelerated gradient methods, and established $\widetilde{O}(\varepsilon^{-1})$ convergence guarantees across both deterministic and stochastic variance-reduced settings \cite{lin2019efficient,lin2022efficiency,luo2023firstorder}. 
These algorithms are genuinely different from SK: instead of exploiting the explicit matrix-scaling structure of $K$, they solve the smooth EOT dual by generic accelerated first-order machinery, trading the simplicity and robustness of alternating normalization for a better dependence on the target accuracy.

% Beyond SK-type scaling, another important route solves EOT through smooth dual or saddle-point formulations by general accelerated first-order methods. 
% An accelerated-gradient approach was designed precisely to improve on the $\varepsilon^{-2}$ accuracy dependence arising in SK-based OT approximation~\cite{dvurechensky2018computational}. Subsequent work developed accelerated primal-dual mirror-descent variants, clarified the correct complexity of the accelerated gradient route, and established essentially $\widetilde O(\varepsilon^{-1})$-type guarantees, including deterministic and stochastic variance-reduced methods \cite{lin2019efficient,lin2022efficiency,luo2023firstorder}. These algorithms are genuinely different from SK: instead of exploiting the explicit matrix-scaling structure of $K$, they solve the smooth EOT dual by generic accelerated first-order machinery, trading the simplicity and robustness of alternating normalization for a better dependence on the target accuracy.

While the SK algorithm often produces high-quality EOT approximations in only a few iterations~\cite{dufosse2022scaling}, bounds based on $\eta\|C\|_\infty$ can be overly pessimistic when this quantity is large. On the one hand, the classical projective-metric approach gives an $O(\log(1/\varepsilon))$ iteration bound, with constants depending on the kernel $K=\exp(-\eta C)$. After normalizing $\min_{i,j}C_{i,j}=0$, the ratio between the smallest and largest entries of $K$ is $\exp(-\eta\|C\|_\infty)$. Estimating the constants using only this ratio can introduce an exponential dependence on $\eta\|C\|_\infty$~\cite{franklin1989scaling}. On the other hand, in the KL- or dual-descent line, the bounds are controlled by dual-radius quantities such as $R$, which scale linearly with $\eta \|C\|_{\infty}$ up to logarithmic marginal terms; thus these guarantees also blow up \cite{dvurechensky2018computational,lin2019efficient,lin2022efficiency,luo2023improved}. This pessimism was explicitly observed in the experiments of \cite{dvurechensky2018computational}.

Crucially, the regime $\eta \|C\|_{\infty}\gg 1$ is not a pathological corner case, but the structural norm in EOT. In practice, the regularization parameter $\eta$ is typically tuned to the bulk scale of the empirical cost distribution \cite{Cuturi2013,Benamou2015,PeyreCuturi2019}. Meanwhile, $\|C\|_\infty$ is dictated by extreme outliers, corrupted measurements, or hard feasibility constraints where forbidden pairs are assigned arbitrarily large penalties \cite{Courty2017, PeyreCuturi2019}. 
Because $\eta \|C\|_{\infty}$ is a fragile, global $L_\infty$ quantity governed by a single extreme entry, bounds based solely on this quantity can become pessimistic in real-world EOT settings.

This limitation suggests that existing worst-case analyses fail to fully capture the structural features driving SK convergence in practice. 
Indeed, the empirical complexity of SK is dictated not by the global maximum cost, but rather by robust, local ``bulk'' properties. 
Even if a small fraction of pairs incur enormous costs, each source and target typically retains a nontrivial weight on moderate-cost assignments.
Fundamentally, this discrepancy arises because standard complexity bounds for SK are naturally expressed in terms of $\nu$, the ratio of the smallest positive entry to the largest entry of the nonnegative matrix $A$ after row-normalization. In the EOT setting, where the kernel matrix is $A = \exp(-\eta C)$, this parameter $\nu$ can shrink exponentially as $\eta\|C\|_\infty$ grows. This structural bottleneck causes existing worst-case guarantees to deteriorate rapidly, failing to reflect the algorithm's actual efficiency on the well-behaved ``bulk'' of the matrix. This observation motivates two closely related questions that we study in this paper:
\begin{itemize}
    \item Under mild assumptions, can one establish a complexity bound for SK applied to EOT that is strictly independent of $\eta\|C\|_\infty$?
    \item More broadly, under what conditions is the complexity of SK genuinely dictated by $\nu$, and when can this dependence be eliminated entirely?
\end{itemize}

% This limitation suggests that existing worst-case analyses may not fully capture the structural features that govern SK convergence in many practical instances: the practical complexity of SK is not governed by the global maximum cost, but rather by robust, local ``bulk" properties. 
% Even if a small fraction of pairs have enormous costs, each source and target point typically retains a nontrivial amount of probability mass on moderate-cost partners.
% Fundamentally, this issue arises because the standard complexity bounds for SK are naturally expressed in terms of $\nu$, the ratio of the smallest positive entry of the nonnegative matrix $A$ to its largest. In the EOT setting, where the kernel matrix is $A = \exp(-\eta C)$, the parameter $\nu$ shrinks exponentially as $\eta\|C\|_\infty$ grows. This structural bottleneck is exactly what causes existing worst-case guarantees to deteriorate rapidly, blinding them to the algorithm's actual efficiency on the well-behaved ``bulk" of the matrix. This observation motivates two closely related questions that we study in this paper:
% \begin{itemize}
% \item Under mild assumptions, can one obtain a complexity bound for SK applied to EOT that is completely independent of $\eta\|C\|_\infty$?
% \item More broadly, when is the complexity of SK genuinely governed by $\nu$, and when can the dependence on $\nu$ be removed altogether?
% \end{itemize}

\subsection{Main Results.}
In this paper, we resolve these open questions. First, we show that, under mild assumptions, the SK algorithm reaches marginal accuracy $\varepsilon$ for EOT in \(O(\log n-\log \varepsilon)\) iterations. Second, we establish a sharp phase transition that precisely characterizes when the complexity of SK is governed by $\nu$, and when it completely decouples from this parameter.

\paragraph{Iteration Complexity for EOT.}
To overcome the fragility of the global infinity norm $\eta\| C\|_\infty$ against extreme outliers, we introduce the concept of $(\tau,\kappa)$-well-boundedness. This notion formalizes the idea that a matrix is fundamentally well-behaved as long as the vast majority of its weighted mass is concentrated on moderate entries.
Let $\tau\geq 0,\kappa>0$, let $m,n\in\mathbb{Z}_{>0}$, and let $\={u}\in\mathbb{R}_{>0}^{m}$ and $\={v}\in\mathbb{R}_{>0}^{n}$ be positive weight vectors normalized so that $\|\={u}\|_{1}=\|\={v}\|_{1}=1$. For a matrix $A\in\mathbb{R}_{\ge 0}^{m\times n}$, we define the row and column bulk capacities as:
$$r_{\tau}(A;\={v}) \triangleq \min_{i\in[m]} \sum_{j\in[n]} v_j\,\id{A_{ij}\le \tau}, \qquad c_{\tau}(A;\={u}) \triangleq \min_{j\in[n]} \sum_{i\in[m]} u_i\,\id{A_{ij}\le \tau}.$$
We say that $A$ is $(\tau,\kappa)$-well-bounded with respect to $(\={u},\={v})$ if $$r_{\tau}(A;\={v})+c_{\tau}(A;\={u})\ge 1+\kappa.$$Here, $\tau$ and $\kappa$ are treated as constants independent of the problem dimensions. In words, this condition requires that the weighted fraction of moderate entries (bounded by $\tau$) in the worst-case row, combined with the corresponding fraction in the worst-case column, strictly exceeds 1 by a constant margin $\kappa$.
Crucially, rather than being bottlenecked by the global maximum $\eta\|C\|_{\infty}$, this condition depends strictly on the cumulative weight of bounded entries. This means that as long as the moderate entries carry sufficient mass to satisfy the $1+\kappa$ threshold, the remaining unbounded entries are permitted to be arbitrarily large without violating the definition. 
A particularly transparent sufficient condition for this property is:
$$r_{\tau}(A;\={v})\geq \frac{1+\kappa}{2}, \qquad c_{\tau}(A;\={u})\geq \frac{1+\kappa}{2}.$$
Under this condition, every row and column places strictly more than half of its weighted mass on entries bounded by $\tau$. 
This formulation is motivated by practical settings where the cost matrix may
contain a small number of extremely large entries, while a substantial fraction
of the entries in every row and column remain within a moderate range. 
Such behavior can arise, for example, when the cost matrix is induced by heavy-tailed data distributions, when pairwise costs are estimated from noisy measurements or corrupted correspondences, or when a subset of source-target pairs receives large penalties due to uncertainty or feasibility constraints.
In these settings, the scaled cost matrix $\eta C$ can be
$(\tau,\kappa)$-well-bounded for moderate $\tau$ and constant $\kappa$, even
when a few anomalously large entries cause the global supremum
$\|\eta C\|_{\infty}$ to become arbitrarily large. Thus, our condition replaces
the worst-case control of all entries by a row/column-wise bulk condition that
only requires sufficient mass on moderate-cost pairs. This allows our analysis
to capture structured cost matrices with distributed extreme entries, while
remaining agnostic to the magnitude of the exceptional entries outside the
bulk.
 
We specify the input to the SK algorithm as a matrix scaling instance, denoted by the tuple $(A, (\boldsymbol{u}, \boldsymbol{v}))$, where $A$ is the matrix to be scaled, and $\boldsymbol{u}$ and $\boldsymbol{v}$ specify the target row and column marginals, respectively.
For the presentation of our main results regarding $(\boldsymbol{u}, \boldsymbol{v})$-scaling, we assume without loss of generality that the target vectors are normalized such that $\|\boldsymbol{u}\|_1 = \|\boldsymbol{v}\|_1 = 1$.\footnote{However, for analytical convenience, we will frequently relax this normalization in our proofs to accommodate general balanced targets where $\|\boldsymbol{u}\|_1 = \|\boldsymbol{v}\|_1 \neq 1$.}
For each matrix $A$ of size $m\times n$ and each $i\in [m],j\in [n]$, let $A_{i,j}$ denote the element in row $i$ and column $j$ of $A$, $r_i(A)$ denote $\sum_{k\in [n]} A_{i,k}$,
and $c_j(A)$ denote $\sum_{k\in [m]} A_{k,j}$.
Let $\=r(A)$ denote the vector $(r_1(A),\dots,r_m(A))$
and $\=c(A)$ denote the vector $(c_1(A),\dots,c_n(A))$.
The following theorem explains the efficiency of SK for EOT.

\begin{theorem}\label{thm-upper-bound-general-ot-uv}
Let $m, n \in \mathbb{Z}_{>0}$, and let $\boldsymbol{u} \in \mathbb{R}_{>0}^m$ and $\boldsymbol{v} \in \mathbb{R}_{>0}^n$ be strictly positive vectors such that $\|\boldsymbol{u}\|_1 = \|\boldsymbol{v}\|_1 = 1$. 
Given parameters $\tau \geq 0$ and $\varepsilon, \kappa \in (0,1]$, let $C \in \mathbb{R}_{>0}^{m \times n}$ be a matrix and $\eta > 0$ be a scalar such that $\eta C$ is $(\tau,\kappa)$-well-bounded with respect to  $(\boldsymbol{u},\boldsymbol{v})$.
Then, given $\bigl(\exp(-\eta C),(\boldsymbol{u},\boldsymbol{v})\bigr)$ as input, the SK algorithm outputs a matrix $A$ satisfying
\[\norm{\=r\left(A\right)- \=u}_{1} + \norm{\=c\left(A\right)- \=v}_{1} \leq \varepsilon\]
in
$O\left(\exp(14\tau)\cdot \kappa^{-6} \cdot\left(\tau + \log n - \log \varepsilon - \log \kappa \right)\right)$ iterations.
\end{theorem}

This theorem provides a theoretical explanation for the efficiency of SK in EOT instances satisfying our structural condition.
Specifically, for any given target marginals $\boldsymbol{u}$ and $\boldsymbol{v}$, regularization parameter $\eta$, and cost matrix $C$, if the scaled cost $\eta C$ is $(\tau,\kappa)$-well-bounded with respect to  $(\={u},\={v})$ for some constants $\tau,\kappa$, our result establishes that SK reaches marginal accuracy $\varepsilon$ for the scaling problem in \eqref{eq-solution-eot} in $O(\log n-\log \varepsilon)$ iterations.
Here, the dependence on $n$ arises from our convention that SK first normalizes rows; consequently, the iteration bound is expressed in terms of the number of columns.

Note that each iteration of SK requires \(O(mn)\) time. Consequently, \cref{thm-upper-bound-general-ot-uv} demonstrates that the algorithm runs in \(\tilde{O}(mn)\) time, where the \(\tilde{O}\) notation suppresses logarithmic factors and constants that depend on \(n\) and \(\varepsilon\). This running time is optimal up to logarithmic factors, as merely reading the input matrix already takes \(\Omega(mn)\) time.

To contextualize the efficient iteration bound in \cref{thm-upper-bound-general-ot-uv}, it is instructive to contrast it with existing complexity guarantees. As discussed earlier, previous analyses essentially present a strict trade-off. The classical projective-metric approach~\cite{franklin1989scaling} gives an $O(\log(1/\varepsilon))$ iteration bound, but estimating its constants using only the global cost range can introduce an exponential dependence on $\eta\|C\|_\infty$. Conversely, modern dual-descent analyses (e.g., \cite{dvurechensky2018computational,lin2019efficient,lin2022efficiency,luo2023improved}) accommodate growing $\eta\|C\|_\infty$, but are bottlenecked by a slow $O(\eta\|C\|_\infty/\varepsilon)$ polynomial dependence on the target accuracy. Consequently, these global-cost-based bounds can become pessimistic when $\eta\|C\|_\infty$ is large. 
Our analysis resolves this bottleneck by shifting the structural requirement from a uniform bound on $\eta\|C\|_\infty$ to the $(\tau,\kappa)$-well-boundedness of $\eta C$. This row/column-wise bulk condition is substantially weaker; while a uniform bound on $\eta\|C\|_\infty$ trivially implies $(\tau,\kappa)$-well-boundedness, the converse fails in general. By operating under this relaxed framework, our theorem successfully allows $\eta\|C\|_\infty$ to be unbounded under the stated bulk condition without sacrificing the optimal logarithmic dependence on accuracy. This robustly explains the algorithm's empirical efficiency even in regimes where $\eta\|C\|_\infty$ is arbitrarily large.

A natural question is whether the $O(\log n - \log \varepsilon)$ iteration bound in \cref{thm-upper-bound-general-ot-uv} can be further improved to $O(\log(1/\varepsilon))$ by removing the dimensional dependence. As we demonstrate via the counterexample in \cref{thm-tightness-lognterm}, this iteration complexity is actually tight for the standard SK algorithm. The necessity of the $\log n$ term arises because SK inherently distorts the structure of the kernel matrix $\exp(-\eta C)$ during early iterations as it aggressively scales the rows and columns to match the target marginals $\={u}$ and $\={v}$. To circumvent this bottleneck, we propose pre-scaling the matrix $\exp(-\eta C)$ with $\diag(\={u})$ and $\diag(\={v})$. 
With this simple pre-scaling step, we prove that SK converges in an accelerated, strictly dimension-free $O(\log(1/\varepsilon))$ iterations.

\begin{theorem}\label{thm-upper-bound-general-ot}
Assume the same conditions and notation as in \cref{thm-upper-bound-general-ot-uv}.
With $(\diag(\=u) \cdot \exp(-\eta C) \cdot \diag(\=v),(\=u,\=v))$ as input, SK outputs a matrix $A$ satisfying 
\[\norm{\=r\left(A\right)- \=u}_{1} + \norm{\=c\left(A\right)- \=v}_{1} \leq \varepsilon\]
in $O\left(\exp(14\tau)\cdot \kappa^{-6}\cdot \left(\tau - \log \varepsilon - \log \kappa\right)\right)$ iterations.
\end{theorem}

We remark that the matrix $A$ in \cref{thm-upper-bound-general-ot} is precisely an approximate solution to the $(\={u}, \={v})$-scaling of $\exp(-\eta C)$.
After pre-scaling, the iteration bound has no explicit dependence on the dimensions $m$ and $n$, but it does depend on $\tau$, $\kappa$, and $\varepsilon$. This is a fixed-$\eta$ iteration guarantee for the regularized EOT scaling problem, not a dimension-independent arithmetic-time bound or an end-to-end guarantee for unregularized OT: each dense $m\times n$ SK iteration still requires $O(mn)$ arithmetic operations.

From a practical perspective, this dimension-independent iteration bound is particularly significant. Because each iteration of SK is highly parallelizable, the overall wall-clock time is primarily dominated by the total number of iterations. Our analysis demonstrates that pre-scaling ensures the iteration count remains stable even as the dimension scales, leading to substantial performance gains in high-dimensional EOT tasks. 
We therefore advocate for its adoption as a standard preprocessing routine in practical implementations.

\paragraph{Phase transition for $(\={u},\={v})$-scaling.}
We further investigate when the iteration complexity of the SK algorithm is strictly governed by $\nu$, and under what conditions this dependence can be eliminated. Our analysis reveals a sharp phase transition in the behavior of SK under $(\={u}, \={v})$-scaling.

We first extend the definition of density for $(\=1,\=1)$-scaling in \cite{he2025phase} to $(\=u,\=v)$-scaling.

\begin{definition}[Density]\label{def-gamma-rho-dense}
Let $\gamma, \gamma', \rho \in (0,1],m,n\in \mathbb{Z}_{>0}$, and let $\boldsymbol{u} \in \mathbb{R}_{>0}^m$ and $\boldsymbol{v} \in \mathbb{R}_{>0}^n$ be positive weight vectors such that $\|\boldsymbol{u}\|_1 = \|\boldsymbol{v}\|_1$. Let $A \in \mathbb{R}_{\ge 0}^{m \times n}$ be a nonzero matrix with maximum entry $t \triangleq \max_{i,j} A_{i,j}$. 
We say $A$ is $(\gamma, \gamma', \rho)$-dense with respect to  $(\boldsymbol{u}, \boldsymbol{v})$
if:
$$ \gamma = \min_{i \in [m]}  \sum_{k \in [n]} \frac{v_k\,\id{A_{ik}\geq\rho t}}{\|\boldsymbol{v}\|_1} , \quad \text{and} \quad \gamma' =  \min_{j \in [n]} \sum_{k \in [m]} \frac{u_k\,\id{A_{kj}\geq \rho t}}{\|\boldsymbol{u}\|_1}. $$

We say $A$ is at least $(\gamma, \gamma', \rho)$-dense with respect to  $(\boldsymbol{u}, \boldsymbol{v})$ if the minimums above are bounded below by $\gamma$ and $\gamma'$, respectively (i.e., replacing the equalities with $\ge$).
We say $A$ is $(\gamma, \gamma')$-dense with respect to  $(\boldsymbol{u}, \boldsymbol{v})$ if it is $(\gamma, \gamma', \rho)$-dense for some $\rho \in (0,1]$.
Finally, we say $A$ is dense with respect to  $(\boldsymbol{u}, \boldsymbol{v})$ if it is at least $(\gamma, \gamma')$-dense for some parameters satisfying $\gamma + \gamma' > 1$.
\end{definition}

The $(\gamma, \gamma', \rho)$-density captures the pervasive distribution of significant entries within a matrix. The threshold $\rho$ identifies ``active" elements relative to the maximum entry, while $\gamma$ and $\gamma'$ guarantee a strict minimum weighted proportion of these elements in every row and column. When a matrix is structurally ``dense" ($\gamma + \gamma' > 1$), these guaranteed minimums force a strong overlap of active entries, ensuring the matrix is highly interconnected and resistant to partitioning.

Fix two positive vectors $\boldsymbol{u}$ and $\boldsymbol{v}$. A pair $(\gamma, \gamma')$ with $\gamma, \gamma' \in (0,1]$ is said to be feasible if there exists a $(\gamma, \gamma')$-dense matrix with respect to  $(\boldsymbol{u}, \boldsymbol{v})$. Pairs that do not satisfy this condition are considered immaterial for the given $(\boldsymbol{u}, \boldsymbol{v})$ and are excluded from our analysis.

Given any nonnegative matrix $A$ with no zero rows, define 
\begin{align}
\nu(A) \triangleq \frac{\min_{A_{i,j} > 0} \left( A_{i,j} / r_i(A) \right)}{\max_{i,j} \left( A_{i,j} / r_i(A) \right)}.
\end{align}
Intuitively, $\nu(A)$ measures the effective ratio between the smallest positive element and the largest element of $A$, independent of arbitrary row scaling. We normalize each entry by its row sum, $r_i(A)$, because the SK algorithm begins with row normalization and is therefore invariant to the absolute scale of individual rows. By pre-normalizing, $\nu(A)$ avoids being artificially skewed by heavily scaled rows, accurately capturing the true dynamic range of the matrix exactly as the algorithm perceives it.

Our results about the phase transition of SK for $(\={u},\={v})$-scaling are as follows.
Fix any $\=u,\=v$. In the following informal bounds, we hold the density parameters $\gamma$, $\gamma'$, and $\rho$ fixed.
We will show that:
% \footnote{Fix any two positive vectors $\boldsymbol{u}$ and $\boldsymbol{v}$. Given any $\gamma,\gamma'\in(0,1]$, we assume throughout that $\gamma$ and $\gamma'$ are chosen so that there exist a matrix $A$ and some $\rho\in(0,1]$ for which $A$ is $(\gamma,\gamma',\rho)$-dense with respect to  $(\boldsymbol{u},\boldsymbol{v})$. Otherwise, such choices of $\gamma,\gamma'$ are immaterial for the given $(\boldsymbol{u},\boldsymbol{v})$ and will not be considered.}
\begin{enumerate}
\item Super-critical regime $\gamma+\gamma'>1$: If the matrix $A$ is $(\gamma,\gamma')$-dense with respect to $(\boldsymbol{u},\boldsymbol{v})$, then SK converges in $O(\log n - \log \varepsilon)$ iterations. In this regime, the complexity is fundamentally independent of $\nu(A)$. \label{phase-transition-item-first}

\item Sub-critical regime $\gamma+\gamma'<1$: For any feasible $\gamma,\gamma'$, there exists some matrix $A$ which is $(\gamma,\gamma')$-dense with respect to $(\boldsymbol{u},\boldsymbol{v})$ such that SK takes $\Omega(-\log \nu(A) - \log \varepsilon)$ iterations to converge. 
By constructing hard instances where $\nu(A) = \exp(-\Theta(n/\varepsilon))$, this lower bound translates to $\Omega(n/\varepsilon)$ iterations, heavily penalizing matrices with extreme dynamic ranges.\label{phase-transition-item-second}

\item Critical boundary $\gamma+\gamma'=1$: At the exact phase transition threshold, the dependence on $\nu(A)$ can still be circumvented for certain targets. Specifically, there exist target marginals $\boldsymbol{u}, \boldsymbol{v}$ such that for any matrix $A$ which is $(\gamma,\gamma')$-dense with respect to $(\boldsymbol{u},\boldsymbol{v})$, SK converges in $O(1/\varepsilon)$ iterations, remaining strictly independent of $\nu(A)$.
Furthermore, it has been proved that there exist some $\=u,\=v,A$ where $A$ is \((1/2,1/2)\)-dense with respect to  $(\boldsymbol{u},\boldsymbol{v})$ such that SK converges in $\Omega(1/\varepsilon)$ iterations~\cite{kalantari1993rate}.
Thus, the time complexity $O(\log n - \log \varepsilon)$ for the regime $\gamma+\gamma'>1$ cannot be extended to  $\gamma+\gamma'=1$.
\end{enumerate}

Our results establish a sharp phase transition in the iteration complexity of the SK algorithm, governed by the matrix density parameters $\gamma$ and $\gamma'$. The critical threshold $\gamma + \gamma' = 1$ separates universally efficient, structure-independent convergence from regimes susceptible to extreme computational degradation. Above this threshold, the algorithm exhibits rapid convergence entirely oblivious to the structural parameter $\nu(A)$. Conversely, in the sub-critical regime where $\gamma + \gamma' < 1$, this guarantee collapses. We prove the existence of matrices satisfying this looser density condition for which the SK algorithm slows down drastically, heavily depending on the structural parameter $\nu(A)$. 
Since $\nu(A)$ can take on arbitrarily small values, the number of iterations required by the algorithm suffers from an arbitrarily poor dependence on $n$ and $\varepsilon$.

We further show that our phase transition results are tight in the following aspects:
\begin{itemize}
\item There exist some $\=u,\=v$
$\gamma,\gamma',A$ where
$\gamma+\gamma'>1$ and $A$ is \((\gamma,\gamma')\)-dense with respect to  $(\boldsymbol{u},\boldsymbol{v})$ such that SK converges in $\Theta(\log n - \log \varepsilon)$ iterations. 
Thus, the time complexity $O(\log n - \log \varepsilon)$ in (\ref{phase-transition-item-first}) is tight.
\item There exist some $\=u,\=v$ such that for any feasible $\gamma,\gamma'$ with $\gamma+\gamma' <1$, any matrix which is \((\gamma,\gamma')\)-dense and $(\=u,\=v)$-scalable, SK converges in $O\left(-\log \nu (A) - \log \varepsilon\right)$ iterations. 
Thus, the time complexity $\Omega(-\log \nu(A) - \log \varepsilon)$ in (\ref{phase-transition-item-second}) is tight.
% \item It has been proved that there exist some $\=u,\=v,A$ where $A$ is \((1/2,1/2)\)-dense with respect to  $(\boldsymbol{u},\boldsymbol{v})$ such that SK converges in $\Omega(\varepsilon)$ iterations~\cite{kalantari1993rate}.
% Thus, the time complexity $O(\log n - \log \varepsilon)$ cannot be extended to the regime $\gamma+\gamma'=1$.
% Hence, the condition $\gamma+\gamma'>1$ in (\ref{phase-transition-item-first}) is tight.
\end{itemize}

Our phase transition results for $(\={u},\={v})$-scaling exhibit a notable difference from the $(\={1},\={1})$-scaling results established in \cite{he2025phase}. 
Fundamentally, the two transitions characterize different aspects of the algorithm's complexity. The phase transition identified in \cite{he2025phase} focuses on $(\={1},\={1})$-scaling under the assumption that the matrix is not overly extreme (i.e., $1/\nu(A) = O(\mathrm{poly}(n, 1/\varepsilon))$), determining when the iteration complexity shifts between a fast $O(\log n - \log \varepsilon)$ rate and a slow $\Omega(1/\varepsilon)$ rate based on matrix density. In contrast, our phase transition for general $(\={u},\={v})$-scaling explores a different dimension: the exact conditions under which the iteration count of SK is inherently dictated by the lower-bound parameter $\nu(A)$, and when it completely decouples from it.

Beyond this conceptual distinction, the algorithmic dynamics in these two settings contrast sharply. For $(\={1},\={1})$-scaling, even within the class of polynomially bounded matrices, a strong lower bound of $\Omega(1/\varepsilon)$ persists if the matrix density falls below a critical threshold. In stark contrast, our analysis reveals a unique structural phenomenon within general $(\boldsymbol{u},\boldsymbol{v})$-scaling: the density bottleneck can be fundamentally bypassed by certain target distributions. Specifically, there exist specific target marginals $\={u}$ and $\={v}$ that induce an extremely rapid mass flow across different regions of the matrix. Driven by this efficient mass transport, the SK algorithm can achieve a fast $O(\log n - \log \varepsilon)$ convergence rate, provided the matrix satisfies the same baseline condition $1/\nu(A) = O(\mathrm{poly}(n, 1/\varepsilon))$. Thus, our results highlight that the strict density limitations inherent to $(\={1},\={1})$-scaling are not an absolute barrier for the SK algorithm; as long as the matrix is not pathologically extreme, there exist specific $(\={u},\={v})$ configurations that unlock rapid mass flow and guarantee highly efficient convergence (see Theorems \ref{thm-lower-bound-main-rc-tight} and \ref{thm-tight-rc-scaling}).

The $\log n$ term in the time complexity of (\ref{phase-transition-item-first}) arises because SK distorts the dense structure of matrix $A$ as it scales the row and column sums to match $\boldsymbol{u}$ and $\boldsymbol{v}$. This distortion can be eliminated by pre-scaling $A$. Specifically, rather than using $(A, \boldsymbol{u}, \boldsymbol{v})$ as the input for SK, we can use the pre-scaling input $(\mathrm{diag}(\boldsymbol{u}) \cdot A \cdot \mathrm{diag}(\boldsymbol{v}), \boldsymbol{u}, \boldsymbol{v})$.
We further show that for any $\gamma,\gamma'$ where
$\gamma+\gamma'>1$, if the matrix $A$ is \((\gamma,\gamma')\)-dense with respect to  $(\boldsymbol{u},\boldsymbol{v})$, then SK converges in $O(-\log \varepsilon)$ iterations with input $(\diag(\=u)\cdot A \cdot \diag(\=v),\=u,\=v)$.
Our analysis reveals that pre-scaling prevents the target vectors $\boldsymbol{u}$ and $\boldsymbol{v}$ from severely distorting the structure of $A$, thereby accelerating the convergence of the SK algorithm by shaving off a $\log n$ term. 
Given that this preprocessing step incurs negligible computational overhead while the $\log n$ reduction yields substantial efficiency gains on massive datasets, it constitutes a highly practical addition to existing algorithmic pipelines. 
This simple modification is particularly beneficial when the target marginals are highly skewed, containing elements of drastically varying magnitudes.

In the following, we introduce the above results formally.
\begin{theorem}\label{thm-rc-scaling-log}
Let $\gamma,\gamma',\rho,\varepsilon \in (0,1],m, n \in \mathbb{Z}_{>0}$. Let $\boldsymbol{u} \in \mathbb{R}_{>0}^m$ and $\boldsymbol{v} \in \mathbb{R}_{>0}^n$ be vectors satisfying $\norm{\boldsymbol{u}}_1 = \norm{\boldsymbol{v}}_1 = 1$, and 
$B \in \mathbb{R}_{\ge 0}^{m \times n}$ be at least \((\gamma,\gamma',\rho)\)-dense with respect to  $(\boldsymbol{u},\boldsymbol{v})$. 
If $\gamma+\gamma'>1$, then with $(B,(\boldsymbol{u}, \boldsymbol{v}))$ as input, SK can output a matrix $A$ satisfying 
\[\norm{\=r\left(A\right)- \boldsymbol{u}}_{1} + \norm{\=c\left(A\right)- \boldsymbol{v}}_{1} \leq \varepsilon\]
in $O\left(\rho^{-14} \cdot\left(\gamma+\gamma'-1\right)^{-6}\left(\log n - \log \varepsilon - \log \rho - \log (\gamma+\gamma'-1)\right)\right)$ iterations.
\end{theorem}

The theorem above implies that the SK algorithm converges in \(O(\log n-\log \varepsilon)\) iterations and runs in \(\tilde{O}(mn)\) time for constant \(\rho,\gamma,\gamma'\).
This is optimal up to logarithmic factors, since merely reading the input matrix already requires \(\Omega(mn)\) time.

One might ask whether a stronger upper bound can be proved when \(\gamma+\gamma'>1\). The next theorem shows this is impossible: 
for some specific matrix and target marginals, the bound \(O(\log n - \log \varepsilon)\) in \cref{thm-rc-scaling-log} is tight.

\begin{theorem}\label{thm-tightness-dense-complexity}
There exist positive vectors $\boldsymbol{u}\in \mathbb{R}_{>0}^m,\boldsymbol{v}\in \mathbb{R}_{>0}^n$ with 
$\norm{\boldsymbol{u}}_1 = \norm{\boldsymbol{v}}_1 = 1$, feasible $\gamma,\gamma'$ with $\gamma + \gamma' > 1$, and a $(\gamma, \gamma')$-dense, $(\boldsymbol{u}, \boldsymbol{v})$-scalable matrix such that, given this matrix and $(\=u,\=v)$ as input, 
SK takes $\Omega(\log n - \log \varepsilon)$ iterations to output a matrix $A$ satisfying 
\[\norm{\=r\left(A\right)- \boldsymbol{u}}_{1} + \norm{\=c\left(A\right)- \boldsymbol{v}}_{1} \leq \varepsilon.\]
\end{theorem}

When the sum $\gamma + \gamma'$ falls below $1$, the parameter $\nu$ begins to feature in the time complexity of the SK algorithm.

\begin{theorem}\label{thm-lower-bound-main-rc}
Let $m, n \in \mathbb{Z}_{>0}$. Let $\boldsymbol{u} \in \mathbb{R}_{>0}^m$ and $\boldsymbol{v} \in \mathbb{R}_{>0}^n$ be vectors satisfying $\norm{\boldsymbol{u}}_1 = \norm{\boldsymbol{v}}_1 = 1$.
For any feasible $\gamma,\gamma'$ with $\gamma+\gamma'< 1$ and $\varepsilon \in \mathbb{R}_{>0}$ with $2\varepsilon< 1- \gamma-\gamma'$,
there exists a $(\gamma,\gamma')$-dense, \((\boldsymbol{u}, \boldsymbol{v})\)-scalable matrix $A$ such that, with $(A, (\=u,\=v))$ as input, 
SK takes $\Omega(\log (1-\gamma) + \log \gamma'-\log \nu(A) - \log \varepsilon)$ iterations to output a matrix $B$ satisfying 
\[\norm{\=r\left(B\right)- \boldsymbol{u}}_{1} + \norm{\=c\left(B\right)- \boldsymbol{v}}_{1} \leq \varepsilon.\]
Furthermore, for this constructed instance, $-\log \nu(A) = \Omega(n/\varepsilon)$, implying an overall iteration complexity of $\Omega(\log (1-\gamma) + \log \gamma' + n/\varepsilon)$.
The implicit constants may depend on $\gamma$ and $\gamma'$, but are independent of $m,n$ and $\varepsilon$.
\end{theorem}

One might ask whether a stronger lower bound can be proved when \(\gamma+\gamma'<1\). The next theorem shows this is impossible: 
for some specific $(\=u,\=v)$, the bound \(-\log \nu(A) - \log \varepsilon\) in \cref{thm-lower-bound-main-rc} is tight.

\begin{theorem}\label{thm-lower-bound-main-rc-tight}
There exist positive vectors $\boldsymbol{u},\boldsymbol{v}$ with 
$\norm{\boldsymbol{u}}_1 = \norm{\boldsymbol{v}}_1 = 1$ such that for any $\varepsilon \in (0,1]$,  any feasible $\gamma,\gamma'$ with $\gamma+\gamma' <1$, and any $(\gamma,\gamma')$-dense, \((\boldsymbol{u}, \boldsymbol{v})\)-scalable matrix $A$, with $(A,(\=u,\=v))$ as input, 
SK takes $O(-\log \nu(A) - \log \varepsilon)$ iterations to output a matrix $B$ satisfying 
\[\norm{\=r\left(B\right)- \boldsymbol{u}}_{1} + \norm{\=c\left(B\right)- \boldsymbol{v}}_{1} \leq \varepsilon.\]
\end{theorem}

The above theorem demonstrates that, for the SK algorithm, $(\boldsymbol{u}, \boldsymbol{v})$-scaling diverges significantly from $(\mathbf{1}, \mathbf{1})$-scaling. While one can construct a matrix $A$ with $\nu(A) = \text{poly}(n, \varepsilon)$ such that SK requires $\Omega(-\log \varepsilon)$ iterations for $(\mathbf{1}, \mathbf{1})$-scaling~\cite{he2025phase}, the situation changes for general marginals. Specifically, for certain $\boldsymbol{u}$ and $\boldsymbol{v}$, SK converges in $O(\log n - \log \varepsilon)$ iterations for any matrix $A$ satisfying $1/\nu(A) = O(\text{poly}(n, 1/\varepsilon))$ (see \cref{thm-tight-rc-scaling}).

As established in Theorems \ref{thm-rc-scaling-log} and \ref{thm-lower-bound-main-rc}, the iteration complexity of the SK algorithm is independent of $\nu(A)$ when $\gamma + \gamma' > 1$, but exhibits a dependence on this parameter when $\gamma + \gamma' < 1$. The following theorem further demonstrates that for specific marginals $\boldsymbol{u}$ and $\boldsymbol{v}$, the complexity can remain independent of $\nu(A)$ even in the boundary case where $\gamma + \gamma' = 1$.

\begin{theorem}\label{thm-lower-bound-main-rc-threshould-tight}
There exist positive vectors $\boldsymbol{u},\boldsymbol{v}$ with 
$\norm{\boldsymbol{u}}_1 = \norm{\boldsymbol{v}}_1 = 1$ such that for any $\varepsilon \in (0,1]$, any feasible $\gamma,\gamma'$ with $\gamma+\gamma' = 1$, and 
any $(\gamma,\gamma')$-dense, 
\((\boldsymbol{u}, \boldsymbol{v})\)-scalable matrix $A$, with $(A,(\=u,\=v))$ as input, 
SK takes $O(1/\varepsilon)$ iterations to output a matrix $B$ satisfying 
\[\norm{\=r\left(B\right)- \boldsymbol{u}}_{1} + \norm{\=c\left(B\right)- \boldsymbol{v}}_{1} \leq \varepsilon.\]
\end{theorem}

The following theorem demonstrates that pre-scaling the matrix with the target marginals accelerates the convergence of the SK algorithm by eliminating the $\log n$ term from the time complexity.
We remark that the matrix $A$ in the following theorem is precisely an approximate solution to the $(\={u}, \={v})$-scaling of $B$.

\begin{theorem}\label{thm-rc-scaling-prenormalization}
Under the same conditions and notation as in \cref{thm-rc-scaling-log},
if $\gamma+\gamma'>1$, then with $(\diag(\=u) \cdot B \cdot \diag(\=v),(\={u}, \={v}))$ as input, SK can output a matrix $A$ satisfying 
\[\norm{\=r\left(A\right)- \boldsymbol{u}}_{1} + \norm{\=c\left(A\right)- \boldsymbol{v}}_{1} \leq \varepsilon\]
in $O\left(\rho^{-14} \cdot\left(\gamma+\gamma'-1\right)^{-6}\left( - \log \varepsilon  - \log \rho - \log (\gamma+\gamma'-1)\right)\right)$ iterations.
\end{theorem}

\paragraph{Applications.}
Our framework applies naturally to two regimes where existing bounds become pessimistic. 
The first regime occurs when the entries of the cost matrix are drawn from a heavy-tailed distribution, causing $\|C\|_\infty$ to be dominated by rare extreme values rather than the typical scale of the data. The second regime involves cost matrices with pairwise or entrywise outliers: a subset of source-target pairs incurs abnormally high or inadmissible costs, yet every source and target point retains many moderate-cost partners. In both settings, our row/column-wise bulk condition captures the phenomenon that the convergence of the SK algorithm is governed by the well-behaved bulk of the cost matrix, rather than by its extreme entries.

We first illustrate the advantage of our framework using a simple heavy-tailed model. Consider the uniform-marginal case where $u_i=v_j=1/n$, and suppose the cost entries $C_{ij}$ are i.i.d. samples from a Pareto distribution with shape parameter $\alpha>1$; that is, $\mathbb{P}(C_{ij} > x) = x^{-\alpha}$ for each $x \ge 1$. Since the matrix contains $n^2$ entries, its maximum is on the scale of $n^{2/\alpha}$. More precisely, for each fixed $c>0$, $\mathbb{P}(\|C\|_\infty\le c n^{2/\alpha})\to \exp(-c^{-\alpha})$. Let $T = \eta C$. Classical SK convergence bounds that depend on the global scale of the kernel $K = \exp(-T)$ typically yield an iteration complexity of the form $O(\eta \|C\|_\infty/\varepsilon)$, whose random scale is $O_{\mathbb{P}}(\eta n^{2/\alpha}/\varepsilon)$~\cite{dvurechensky2018computational,luo2023improved}, which grows polynomially in $n$ for heavy-tailed cost distributions.

By contrast, our condition depends only on the row/column-wise bulk mass. To formalize this, we fix $\kappa=0.01$ and define $\tau$ such that $\mathbb{P}(C_{ij}\le \tau)=0.51$. For the Pareto distribution, this yields $\tau=(1-0.51)^{-1/\alpha}=0.49^{-1/\alpha}$. For each row $i$, let $X_i=\frac{1}{n}\sum_{j=1}^n \mathbf{1}_{\{C_{ij}\le \tau\}}$, and for each column $j$, let $Y_j=\frac{1}{n}\sum_{i=1}^n \mathbf{1}_{\{C_{ij}\le \tau\}}$. Both $X_i$ and $Y_j$ have an expectation of $0.51$. Because $1/2+\kappa/2=0.505$, Hoeffding's inequality and a union bound over all rows and columns imply that the joint event $\min_i X_i\ge 1/2+\kappa/2$ and $\min_j Y_j\ge 1/2+\kappa/2$ occurs with a probability of at least
$$1-2n\exp(-2(0.51-0.505)^2 n) = 1-2n\exp(-5\times 10^{-5} n).$$
Conditional on this event, we have $r_\tau(C;v)+c_\tau(C;u)\ge 1+\kappa$. Therefore, $C$ is $(\tau,\kappa)$-well-bounded with high probability. Equivalently, the scaled cost matrix $T=\eta C$ is $(\eta\tau,\kappa)$-well-bounded with high probability. Applying \cref{thm-upper-bound-general-ot-uv} yields the iteration bound $O(e^{14\eta\tau}(\eta\tau+\log n- \log \varepsilon))$ with $\tau=0.49^{-1/\alpha}$.
Thus, for fixed $\eta, \alpha, \kappa$, and a given target accuracy $\varepsilon$, the iteration complexity grows only logarithmically with $n$, whereas the sample maximum $\|C\|_\infty$ has scale $n^{2/\alpha}$, which is polynomial in $n$.
This bound can be further improved to $O(e^{14\eta\tau}(\eta\tau- \log \varepsilon))$ with \cref{thm-upper-bound-general-ot},
which is independent of $n$.
This demonstrates how our framework replaces the dependence on the extreme tail of the cost distribution with a dependence on a fixed bulk quantile.

A second application addresses outliers that occur at the level of source-target pairs, rather than at the level of individual source or target points. We emphasize that our well-boundedness condition is not designed to handle all possible outlier formulations in OT. In the standard geometric setting where $C_{ij}=\|x_i-y_j\|^2$, a single outlying source or target point can corrupt an entire row or column of the cost matrix. Such point-level outliers are better addressed by partial, trimmed, or unbalanced OT formulations, which relax the marginal constraints. 

Our results instead target a complementary regime: pairwise or entrywise cost outliers. In many computational OT applications, the cost matrix is not simply the Euclidean distance between raw observations; rather, it is learned, estimated, masked, constrained, or derived from noisy pairwise relations. In such scenarios, a subset of source-target pairs may incur abnormally large or inadmissible costs, while each individual point maintains numerous moderate-cost partners. This is precisely the setting captured by our row/column-wise bulk condition. 
For example, consider the uniform-marginal setting where $u_i=v_j=1/n$. 
Suppose that, independently, each entry $C_{ij}$ equals $C^{\mathrm{bulk}}_{ij}$ with probability $1-q$, and $M_{ij}$ with probability $q$, where $C^{\mathrm{bulk}}_{ij} \le \tau$ while $M_{ij}$ can be arbitrarily large.
If $q<1/2$, then for any $\delta\in(0,1/2-q)$, Hoeffding's inequality combined with a union bound implies that, with probability at least $1-2n\exp(-2\delta^2 n)$, every row and column contains a fraction of at least $1-q-\delta$ entries bounded by $\tau$. Consequently, the cost matrix is $(\tau, 1-2q-2\delta)$-well-bounded. Our theorems then yield iteration bounds that depend on $\tau$, but remain completely independent of the magnitude of the corrupted entries. Thus, our framework should be viewed not as a replacement for point-outlier-robust OT, but as a finite-accuracy convergence theory for SK under distributed pairwise cost contamination.

Nevertheless, our analysis does not generally improve the dependence on the cost scale when the maximum cost and the row- and column-wise bulk costs are of comparable magnitude. A representative example is a cost matrix with i.i.d. uniform $[0,1]$ entries and uniform marginals: satisfying our density condition with a fixed positive margin requires $\tau$ of order $\eta$ with high probability, so the factor $e^{14\tau}$ remains exponential in $\eta$. This exponential dependence on $\eta$ is qualitatively similar to that arising from the classical projective-metric approach~\cite{franklin1989scaling}. The main advantage of our analysis therefore lies in instances where exceptionally large costs inflate $\|C\|_\infty$, while sufficiently many moderate-cost entries remain in every row and column. In this regime, our bound can be insensitive to the magnitude of the extreme entries. Outliers are a recognized challenge in practical EOT applications, and our results address configurations in which these extreme costs are sufficiently dispersed to preserve the required row- and column-wise density.

\vspace{0.2cm}

The remainder of this paper is organized as follows. 
\Cref{sec-technique-overview} outlines the primary proof techniques driving our main results.
\Cref{sec-preliminary} introduces the necessary preliminaries and notations. 
\Cref{sec-upper-bounds} presents our core reduction framework from $(\=u, \=v)$-scaling to $(\=1, \=1)$-scaling. 
With this reduction in place, \cref{sec-ub-oneonescaling} establishes our results regarding the fast convergence of the SK algorithm for the $(\=1, \=1)$-scaling problem. 
\Cref{sec-ub-uv-scaling} then synthesizes these ingredients to conclude the upper bound analysis, providing the formal proofs for Theorems \ref{thm-upper-bound-general-ot-uv}, \ref{thm-upper-bound-general-ot}, \ref{thm-rc-scaling-log}, and \ref{thm-rc-scaling-prenormalization}. Next, \cref{sec-lb} is dedicated to the lower bound analysis, where we complete the proof of \cref{thm-lower-bound-main-rc}. Finally, \cref{sec-tightness} discusses the tightness of our bounds, detailing the proofs for Theorems \ref{thm-tightness-dense-complexity}, \ref{thm-lower-bound-main-rc-tight}, and \ref{thm-lower-bound-main-rc-threshould-tight}.

%% file: technique.tex
\section{Technique overview.}\label{sec-technique-overview}
In this section, we summarize the primary proof techniques utilized in this paper. 
Our results can be broadly categorized into two parts: upper bounds, which demonstrate the fast convergence of the SK algorithm (centered around Theorem \ref{thm-rc-scaling-log}), and lower bounds, which characterize scenarios where the SK algorithm converges slowly (principally Theorems \ref{thm-lower-bound-main-rc} and \ref{thm-lower-bound-main-rc-tight}). 
Below, we outline the core proof strategies for both.

\subsection{Techniques on upper bounds.}
Our approach to proving \cref{thm-rc-scaling-log} is inspired by \cite{he2025phase}, which established that the SK algorithm converges in $\mathcal{O}(\log n + \log(1/\varepsilon))$ iterations for the $(\boldsymbol{1},\boldsymbol{1})$-scaling of dense matrices. To extend this fast convergence guarantee to the $(\boldsymbol{u},\boldsymbol{v})$-scaling of dense matrices, we reduce the original $(\boldsymbol{u},\boldsymbol{v})$-scaling problem on an $m \times n$ matrix $A$ to a standard $(\boldsymbol{1},\boldsymbol{1})$-scaling problem on an $N \times N$ matrix $B$. This reduction is crucial, as it enables us to leverage combinatorial tools, such as the matrix permanent, which are naturally defined for square matrices.

For the $(\boldsymbol{1},\boldsymbol{1})$-scaling of $B$, the permanent serves as a potential function tracking the progress of the SK iterations. As long as the current matrix is far from being doubly stochastic, after each row normalization step, there exists a column whose sum deviates significantly from one; the subsequent column normalization then increases the permanent by a corresponding multiplicative factor. However, bounding the total number of iterations requires a strict upper bound on the maximum possible growth of this permanent. To address this, we establish a strengthened version of \emph{structural stability}. 
Roughly speaking, this property guarantees that once the iterates enter a sufficiently accurate regime, the entries of the resulting scaled matrices become comparable to those of the initially row-normalized matrix $B$.
Consequently, the permanent cannot grow excessively.

Once within this accurate regime, we transition to a different convergence argument. 
The structural stability property ensures that the scaled matrices preserve the essential density structure of the original matrix: a dominant fraction of the entries in every row and column remains sufficiently large. 
This preserved density enables us to prove that, in each column normalization step, either the deviation of the maximum column sum from one, or the deviation of the reciprocal of the minimum column sum from one, contracts by a constant factor.
Combined with the bounded permanent growth in the initial phase, this geometric contraction guarantees a tight upper bound on the overall iteration complexity.

In the remainder of this section, we present the main technical ingredients of our analysis, including the discretization-based reduction and the strengthened structural stability property.

\paragraph{Discretization-based reduction.} Notably, no linear transformation exists to directly reduce $(\boldsymbol{u},\boldsymbol{v})$-scaling to $(\boldsymbol{1},\boldsymbol{1})$-scaling. 
Instead, our reduction relies on discretization and subdivision. 
Given an instance $(A, (\boldsymbol{u},\boldsymbol{v}))$, we first choose a sufficiently large integer $L$. 
By appropriately rounding $(L\boldsymbol{u}, L\boldsymbol{v})$, we obtain positive integer vectors $(\={u}', \={v}')$ such that $\|\={u}'\|_1 = \|\={v}'\|_1$.
We then expand matrix $A$ into a block matrix $B$. 
This expanded matrix $B$ is constructed by partitioning each element $A_{i,j}$ into a block of $u'_i \times v'_j$ sub-entries, where each sub-entry is assigned a uniform value of $A_{i,j}/(u'_i \times v'_j)$. 
Through this process, we effectively reduce the instance $(A, (\boldsymbol{u},\boldsymbol{v}))$ to $(B, (\boldsymbol{1},\boldsymbol{1}))$.
To validate this reduction, we establish two critical components:
\begin{itemize}
\item Correctness of the Reduction: We prove that for any fixed iteration step $k$, by choosing a sufficiently large parameter $L$, the marginal error of the $(\boldsymbol{u}, \boldsymbol{v})$-scaling on $A$ at step $k$ can be made arbitrarily close to $1/L$ times the marginal error of the $(\={1}, \={1})$-scaling on the expanded matrix $B$ at step $k$. We achieve this in two steps. First, we establish an operational equivalence: performing $(\={u}', \={v}')$-scaling on matrix $A$ via the SK algorithm is strictly equivalent to performing standard $(\={1}, \={1})$-scaling on the expanded matrix $B$. This equivalence can be rigorously verified by tracing the row and column normalization steps throughout the SK iterations. Second, we prove that for a sufficiently large $L$, the marginal error of the $(\={u}, \={v})$-scaling on $A$ at step $k$ is tightly approximated by $1/L$ times the marginal error of the $(\={u}', \={v}')$-scaling on $A$. Combining these two results immediately confirms the correctness of our reduction.

\item Discrepancy Control and Structural Dynamics: 
A critical challenge arises during our reduction: even if the original matrix $A$ is dense with respect to  $(\={u}, \={v})$, the expanded matrix $B$ is generally not dense with respect to  $(\={1}, \={1})$. 
To bound the iteration complexity of the SK algorithm on the reduced input $(B, (\boldsymbol{1}, \boldsymbol{1}))$, we establish key properties concerning the dynamics of this dense structure. 
First, we show that although $B$ loses its density with respect to uniform marginals, this underlying dense structure can be recovered via appropriate row and column scalings. To see this connection, we introduce an intermediate matrix $C$, formed by partitioning each element $A_{i,j}$ into a $u'_i \times v'_j$ block of sub-entries, all set to the value $A_{i,j}$. Crucially, $B$ is simply a scaled version of $C$; it can be verified that $B = UCV$ for some positive diagonal matrices $U$ and $V$. Finally, we prove that for a sufficiently large $L$, this underlying matrix $C$ is indeed dense with respect to  $(\={1}, \={1})$.
Second, we rigorously characterize the discrepancy between $B$ and the well-structured matrix $C$. 
By comparing corresponding elements in their row-normalized counterparts, we demonstrate that the ratio $(C_{i,j}/r_i(C))/(B_{i,j}/r_i(B))$ on their common positive support is bounded by $\mathcal{O}(n)$ for fixed density parameters and sufficiently large $L$.
Together, these two insights provide a foundational characterization of the dynamics under reduction, allowing us to precisely measure the extent to which the normalized matrix $B$ deviates from being dense under $(\={1}, \={1})$-scaling.
\end{itemize}

Our reduction establishes a fundamental connection between $(\={u},\={v})$-scaling and $(\={1},\={1})$-scaling. 
It not only allows combinatorial techniques designed for $(\={1},\={1})$-scaling to be seamlessly transferred to $(\={u},\={v})$-scaling (and vice versa), but it also reduces the dynamic analysis of rectangular matrices to that of square matrices. 
Consequently, square-matrix-exclusive properties like the permanent can now be applied to analyze the dynamics of rectangular matrices, suggesting that our framework holds potential for broader matrix analysis applications.

Through this reduction, we observe a critical phenomenon: even if the input matrix $A$ exhibits a well-behaved structure, highly skewed target marginals $\={u}$ and $\={v}$ with extreme dynamic ranges can severely degrade the density of the reduced matrix $B$.
Since performing $(\={u},\={v})$-scaling on $A$ is fundamentally equivalent to performing $(\={1},\={1})$-scaling on $B$, the SK algorithm may require up to $\Theta(\log n - \log \varepsilon)$ iterations to converge (see \cref{thm-tightness-lognterm} for an example). 
Fortunately, the distortion introduced by $\={u}$ and $\={v}$ can be neutralized via pre-scaling, which accelerates the convergence of the SK algorithm by shaving off the $\log n$ term (\cref{thm-rc-scaling-prenormalization}). 
These results illustrate that our reduction uncovers the intrinsic structural properties of the SK algorithm, accurately capturing how the target probability vectors $(\={u},\={v})$ influence the structural dynamics of the input matrix.

\paragraph{Strengthened structural stability.} 
While powerful, our reduction introduces two primary analytical hurdles:
\begin{itemize}
\item As noted, the reduction destroys the dense structure of the matrix; $B$ is generally not dense with respect to  $(\boldsymbol{1}, \boldsymbol{1})$. 
Consequently, we must analyze the convergence time of the SK algorithm when the input is the stretched matrix $B = UCV$, rather than the perfectly dense matrix $C$. This requires relaxing the strict density requirement utilized in \cite{he2025phase}.

\item To guarantee the precision of the reduction, the dimension $N$ of the reduced matrix $B$ becomes exceptionally large. 
Consequently, for the $(\={1}, \={1})$-scaling, we must establish that the iteration complexity required for the SK algorithm to achieve an error of $N \varepsilon$ is independent of $N$. 
We set the target error to $N \varepsilon$ because, as mentioned above, an error of $N \varepsilon$ in the $(\mathbf{1}, \mathbf{1})$-scaling directly corresponds to an error of $\varepsilon$ in the $(\={u}, \={v})$-scaling (where $\|\={u}\|_1 = \|\={v}\|_1 = 1$). 
Notably, the iteration bound we prove here is significantly stronger than the one established in \cite{he2025phase}. 
Specifically, if we set $\varepsilon = 1$ (yielding a target error of $O(N)$), our result guarantees an iteration count completely independent of $N$. In stark contrast, Theorem 3.2 in \cite{he2025phase} demonstrates that achieving this exact same $O(N)$ error with the SK algorithm still necessitates $O(\log N)$ iterations.
\end{itemize}

To overcome these obstacles, we establish a stronger version of \emph{structural stability} for the SK algorithm. Structural stability, originally introduced in \cite{he2025phase}, describes a combinatorial invariance maintained by matrices during the iterative scaling process, allowing one to capture essential structural traits across a sequence of changing matrices. 
Specifically, recall that $C$ is a dense matrix with respect to $(\={1},\={1})$. Let $C'$ denote its row-normalized counterpart. 
Furthermore, let $D$ be any scaled matrix produced by the SK algorithm with input $(C,(\=1,\=1))$, provided that $D$ is sufficiently close to a doubly stochastic matrix.
An entry of $C'$ is considered ``considerable" if it is $\Theta(1/N)$. 
While \cite{he2025phase} demonstrated that if SK takes $C$ as input, the considerable entries in $C'$ remain $\Theta(1/N)$ in the scaled matrix $D$, we significantly reinforce this property in two directions:
\begin{itemize}
\item First, we prove that this structural stability holds even when the SK algorithm takes an arbitrarily scaled matrix $B = UCV$ as input, rather than relying on the unscaled dense matrix $C$ itself. 
This relaxation means our structural stability does not depend on the initial matrix $B$ having a good density structure, but rather depends exclusively on the well-behaved properties of the scaling orbit $\left\{ \mathrm{diag}(x) C \, \mathrm{diag}(y) \mid x \in \mathbb{R}_{>0}^N, y \in \mathbb{R}_{>0}^N \right\}$ generated by $C$.

\item Second, we additionally prove that \emph{every} entry of $D$ is bounded above by a constant multiple $\delta$ of its corresponding entry in $C'$ (Item \ref{lemma-structure-stablility-two} of \cref{lem-structural-stability}). 
This implies that the permanent of the matrix $D$ is bounded above by $\delta^N$ times the permanent of $C'$. 
By leveraging the permanent of $D$ to bound the number of SK iterations, we successfully prove that the iteration complexity required to reach an $N\varepsilon$ precision is independent of $N$. This critical enhancement allows structural properties to couple perfectly with the permanent, enabling a precise analysis of matrix dynamics.
\end{itemize}

In conclusion, the strengthened structural stability we establish relies solely on the intrinsic properties of the scaling orbit, yielding a robust upper bound for the permanent of near-doubly stochastic matrices within this orbit. 
This affords a fundamentally deeper understanding of the SK algorithm, illuminating its underlying matrix dynamics.

\subsection{Techniques on lower bounds.}
The proof of \cref{thm-lower-bound-main-rc} proceeds in two stages. First, we establish a lower bound for the $(\={1},\={1})$-scaling of a family of square matrices consisting of four constant blocks (\cref{thm-main-lb-rc-ab}). We then construct a matrix $A$ whose reduced instance reaches a matrix in this family during its initial SK iterations, allowing us to transfer the lower bound to the $(\={u},\={v})$-scaling of $A$.

To explain this connection, recall that our reduction uses a sufficiently large positive integer $L$ to approximate $L\={u}$ and $L\={v}$ by positive integer vectors with equal total mass. These integer vectors determine how each entry of $A$ is subdivided to form the reduced matrix $C$. Although this subdivision need not preserve the four constant blocks of $A$, \cref{lem-lowerbound-reduction-rctooneone} shows that, for our construction, four constant blocks reappear at $C^{(2)}$, after successive row, column, and row normalizations. We choose the entries of $A$ so that, for sufficiently large $L$, the matrix $B\triangleq C^{(2)}$ satisfies the hypotheses of \cref{thm-main-lb-rc-ab} with target error $2L\varepsilon$. Writing $z$ for the common value of its bottom-left entries, this theorem gives a lower bound of
\[
c_0\bigl(-\log z-\log(2L\varepsilon)\bigr)-b_0
\]
iterations for the SK algorithm starting from $B$. Here, $c_0,b_0>0$ depend only on the fixed density parameters, and are independent of $L$, the matrix dimensions, and $\varepsilon$.

The error threshold $2L\varepsilon$ accounts for the approximation in the reduction: over any prescribed finite range of iterations, choosing $L$ sufficiently large makes the marginal error on $A$ differ from $1/L$ times that on $C$ by at most $\varepsilon$. Thus, an error greater than $2L\varepsilon$ on $C$ implies an error greater than $\varepsilon$ on $A$. Our construction also ensures that $Lz=O(n\nu(A))$, where $n$ is the number of columns of $A$ and the implicit constant depends only on the density parameters. Substituting this estimate into the lower bound eliminates the dependence on $L$. Choosing $\nu(A)$ sufficiently small, independently of $L$, absorbs the resulting $\log n$ and constant terms and yields the desired $\Omega(-\log\nu(A)-\log\varepsilon)$ lower bound.

% The proof of \cref{thm-lower-bound-main-rc} proceeds as follows. To establish the $\Omega(-\log \nu(A) - \log \varepsilon)$ lower bound for the $(\=u, \=v)$-scaling of $A$, 
% our core strategy is to construct an elementary block matrix $B$ whose bottom-left entries equal $z$ and which requires at least $c_0(-\log z-\log(L\varepsilon))-C_0$ SK iterations to converge under $(\={1},\={1})$-scaling, where $L$ is the reduction scale, and subsequently embed it into $C$, the reduced $(\={1}, \={1})$-scaling instance derived from $(A, (\={u}, \={v}))$.
% Let $C^{(0)}, C^{(1)}, C^{(2)}, \dots$ denote the sequence of matrices generated by applying the SK algorithm to $C$. 
% Here, \cref{lem-lowerbound-reduction-rctooneone} plays a pivotal role in constructing our hard instance. 
% While our non-uniform reduction inherently destroys the block structure at the early state $C^{(0)}$, this lemma guarantees that the block nature of $A$ is completely recovered exactly at state $C^{(2)}$. 
% This property is highly advantageous: it allows us to safely focus on designing $B$ as a block matrix, 
% without having to worry about the intermediate structural distortion. 
% Thus, by carefully tuning the entries of $A$, we can seamlessly force $C^{(2)}$ to match $B$ exactly, while satisfying $Lz = O(n\nu(A))$, where the implicit constant depends only on the density parameters. Choosing $\nu(A)$ sufficiently small, independently of $L$, absorbs the additive constant and yields the desired lower bound for $A$.

In the following, we introduce the construction of $B$, which is the core of the proof of \cref{thm-lower-bound-main-rc}.
At a high level, our construction of $B$ bottlenecks the SK algorithm by combining an artificially tiny minimum entry with a massive initial marginal deviation.
We design $B$ as a block matrix with intentionally mismatched block dimensions and initialize its bottom-left block to an exponentially small value $\nu$.
Because the algorithm alternates between row and column normalizations, this dimensional imbalance creates a cascading ``push-pull" effect. 
During row normalizations (even iterations), the artificially tiny bottom-left entry forces the adjacent bottom-right block to absorb the bulk of the row mass. Due to the block size mismatch, this heavily inflates the subsequent column sums of the right columns. Symmetrically, during column normalizations (odd iterations), the tiny bottom-left entry forces the top-left block to absorb the bulk of the column mass, thereby inflating the subsequent row sums of the top rows. Regardless of the phase, these inflated marginals consistently compel the algorithm to shrink the top-right block. 
In essence, the top-right block is trapped in a decaying cycle until the bottom-left block accumulates enough mass.
This dynamic forces the SK algorithm to suffer through two distinct computational bottlenecks, which correspond exactly to the two parameter regimes in our formal analysis (the displayed asymptotic lower bounds apply when the corresponding logarithms are sufficiently large):
\begin{itemize}
\item The Growth Bottleneck (Escaping the initial trap): In the regime where the target error is relatively loose ($\varepsilon/n > n\nu$), the primary challenge for the algorithm is to grow the artificially tiny entry from $\nu$ up to a macroscopic scale ($\Theta(1/n)$) in order to eventually satisfy the marginal constraints. Because the SK updates are multiplicative, the per-iteration growth factor of this entry is strictly bounded by a constant. Consequently, this initial ``ramp-up" phase inescapably requires $\Omega(-\log n - \log \nu)$ iterations.
\item The Decay Bottleneck (Slow asymptotic convergence): In the regime where the target error is extremely tight ($\varepsilon/n \leq n\nu$), the bottleneck shifts to the agonizingly slow geometric decay of the scaling error. In this phase, while the matrix entries have reached the correct order of magnitude, each SK update only alters the relevant entries by a relative amount proportional to the current error. This means the residual error shrinks by at most a constant factor per iteration. Since the initial unscaled error is macroscopic ($\Theta(n)$), geometrically shrinking this error down to $\varepsilon$ demands at least $\Omega(\log n - \log \varepsilon)$ iterations.
\end{itemize}
In summary, by accounting for the iterations required to overcome both the initial growth trap and the subsequent slow error decay, we establish the overall $\Omega(-\log \nu - \log \varepsilon)$ lower bound.

To establish the tightness of our bound, \cref{thm-lower-bound-main-rc-tight} constructs a pair of target margins $(\=u,\=v)$ with engineered asymmetry, ensuring that for any $(\=u,\=v)$-scalable nonnegative $2 \times 2$ matrix $A'$, the $(\=u,\=v)$-scaling converges in at most $O(-\log \nu(A') - \log \varepsilon)$ iterations. Let $B'$ be the reduced $(\=1,\=1)$-instance derived from $(A',(\=u,\=v))$, which remains of constant size in our construction.

The core intuition behind our proof of \cref{thm-lower-bound-main-rc-tight} is to mirror the exact matrix dynamics we established for the hard instance $B$. Because the target margins $(\=u,\=v)$ are highly asymmetric, the reduced instance $B'$ naturally exhibits the same intentionally mismatched block dimensions discussed previously. 
By exploiting this structural asymmetry, we can tightly bound the SK iterations on $B'$ by decomposing the process into two distinct phases that perfectly parallel our previous analysis:
\begin{itemize}
\item Phase 1: Rapid Growth (Overcoming the initial trap). In the initial stage of the execution, the marginal errors need not be small. While the error exceeds a fixed threshold, the permanent of $B'$ increases by at least a constant factor at each normalization. The permanent of its initial row-normalized iterate is bounded below by a constant multiple of a fixed power of $\nu(B')$, whereas the permanent of every normalized iterate is at most $1$.  If the marginal error remained above the threshold, each normalization would increase the permanent by a fixed factor greater than $1$, eventually contradicting this upper bound. Consequently, the marginal error must fall below that threshold.
This initial ``ramp-up" phase takes at most $O(1-\log \nu(A'))$ iterations.
\item Phase 2: Dense Linear Convergence (Asymptotic decay).
Once the algorithm advances past the initial $O(1-\log \nu(A'))$ steps, the error drops below the aforementioned threshold. At this stage, the intermediate matrices generated in each iteration fundamentally inherit the asymmetric block structure of $B'$. Crucially, the combination of these mismatched block dimensions and the already-small marginal errors actively forces the intermediate matrices to remain in a strictly \emph{dense regime}. 
Consequently, we can invoke \cref{thm-rc-scaling-log} to guarantee that the SK algorithm linearly converges to an $\varepsilon$-error in an additional $O(-\log \varepsilon)$ iterations.
\end{itemize}
Combining these two phases, the total iteration complexity on $B'$ is strictly bounded by $O(-\log \nu(A') - \log \varepsilon)$, which ultimately completes the proof of \cref{thm-lower-bound-main-rc-tight}.

%% file: technicaltools.tex
\section{Preliminary.}\label{sec-preliminary}
Unless explicitly indicated otherwise, $\log$ denotes the natural logarithm throughout this paper. We use $\mathbb{Z}_{>0}$ to denote the set of strictly positive integers. Let $\mathbb{R}_{>0}$ and $\mathbb{R}_{\ge 0}$ denote the sets of strictly positive and nonnegative real numbers, respectively. For any integers $m, n \in \mathbb{Z}_{>0}$, we use $\mathbb{R}_{>0}^m$ to represent the set of $m$-dimensional vectors with strictly positive entries. Similarly, $\mathbb{R}_{\ge 0}^{m \times n}$ denotes the set of $m \times n$ matrices with nonnegative entries.

A square matrix $A \in \mathbb{R}_{\ge 0}^{n \times n}$ is called doubly stochastic if its row and column sums all equal one.

Given any $m\in \mathbb{Z}_{>0}$ and $\=u = (u_1,u_2,\dots,u_m)\in \mathbb{Z}_{>0}^m$, 
define
\begin{equation}
\begin{aligned}
\+D(\=u) \triangleq \diag\Bigl(
\underbrace{u_1,\ldots,u_1}_{u_1\text{ entries}},\;
\underbrace{u_2,\ldots,u_2}_{u_2\text{ entries}},\;
\ldots,\;
\underbrace{u_m,\ldots,u_m}_{u_m\text{ entries}}
\Bigr).\label{eq-def-diag}
\end{aligned}
\end{equation}

\paragraph{SK algorithm.} Let $m, n \in \mathbb{Z}_{>0}$. Let $\boldsymbol{u} \in \mathbb{R}_{>0}^m$ and $\boldsymbol{v} \in \mathbb{R}_{>0}^n$ be vectors satisfying $\norm{\boldsymbol{u}}_1 = \norm{\boldsymbol{v}}_1$, and 
$A \in \mathbb{R}_{\ge 0}^{m \times n}$ be a matrix with no zero rows or columns.
Given $(A,(\=u,\=v))$ as input, the SK algorithm iteratively generates a sequence of matrices $A^{(0)}, A^{(1)}, \ldots$ as follows:
\begin{itemize}
\item For each $i\in [m],j\in [n]$, let $A^{(0)}_{i,j} = u_i\cdot A_{i,j}/r_i(A)$;
\item For each integer $k>0$ and $i\in [m],j\in [n]$, if $k$ is odd, let $A^{(k)}_{i,j} = v_j\cdot A^{(k-1)}_{i,j}/c_j\left(A^{(k-1)}\right)$; 
otherwise, let $A^{(k)}_{i,j} = u_i\cdot A^{(k-1)}_{i,j}/r_i\left(A^{(k-1)}\right)$.
\end{itemize}

The following are some easy facts about the SK algorithm.

\begin{lemma}\label{fact-c-a-a0-a1}
Let $n \in \mathbb{Z}_{>0}$, and let
$A \in \mathbb{R}_{\ge 0}^{n \times n}$ be a matrix with no zero rows or columns.
Let $A^{(0)},A^{(1)},\dots$ be the generated sequence of matrices by the SK algorithm with input $(A,(\=1,\=1))$.
Then we have $A^{(k)}_{i,j} \in [0,1]$ for each $i\in [n],j\in[n]$ and $k\geq 0$.
Moreover, for each $k\geq 0$ and $i\in [n]$, if $k$ is even,  $r_i\left(A^{(k)}\right) = 1$. Otherwise, $c_i\left(A^{(k)}\right) = 1$.
\end{lemma}

The following lemma from \cite{sinkhorn1964relationship} demonstrates that the extremal (i.e., maximum and minimum) row and column sums behave monotonically in the SK algorithm for $(\=1,\=1)$-scaling.

\begin{lemma}\label{lem-monotone-rsum-csum}
Suppose the conditions in \cref{fact-c-a-a0-a1} hold. 
Then for any odd $k$, we have 
\[
\min_{i\in [n]}r_i\left(A^{(k)}\right) \leq  \min_{i\in [n]}r_i\left(A^{(k+2)}\right) \leq 1 \leq \max_{i\in [n]}r_i\left(A^{(k+2)}\right) \leq \max_{i\in [n]}r_i\left(A^{(k)}\right).
\]
Similarly, for any even $k$, we have
\[
\min_{j\in [n]}c_j\left(A^{(k)}\right) \leq  \min_{j\in [n]}c_j\left(A^{(k+2)}\right) \leq 1 \leq \max_{j\in [n]}c_j\left(A^{(k+2)}\right) \leq \max_{j\in [n]}c_j\left(A^{(k)}\right).
\]
Moreover, for any odd $k$, we have 
\[
\left(\max_{i\in [n]}c_i\left(A^{(k-1)}\right)\right)^{-1} \leq  \min_{i\in [n]}r_i\left(A^{(k)}\right) \leq 1 \leq \max_{i\in [n]}r_i\left(A^{(k)}\right) \leq \left(\min_{i\in [n]}c_i\left(A^{(k-1)}\right)\right)^{-1}.
\]
\end{lemma}

\paragraph{Permanent.}
For an $n\times n$ matrix $Z$,
its permanent is defined as
\[\perman(Z) \triangleq \sum_{\sigma }\prod_{i\in [n]}Z_{i,\sigma(i)},\]
where the sum is over all permutations $\sigma$ of $[n]$.

The following lower bound on the permanent of doubly stochastic matrices was first conjectured by
Van der Waerden and later proved independently by Falikman~\cite{falikman1981proof} and Egorychev~\cite{egorychev1981solution}.

\begin{lemma}\label{lem-perm-double-stochastic}
For any doubly stochastic matrix $A$ of size $n\times n$, we have $\perman(A)\geq n!/n^n$.
\end{lemma}

The following properties regarding the permanent of the matrices generated during the SK algorithm are well-established in the literature \cite{Linial2000Deterministic}.

\begin{lemma}\label{lem-facts-ak}
Suppose the conditions in \cref{fact-c-a-a0-a1} hold. 
For any $i\in [n]$, let $x^{(0)}_i = y^{(0)}_i = 1$.
For any $k > 0$ and $i\in [n]$, let
$x^{(k)}_i = 1/\prod_{j=0}^{k-1} r_i\left(A^{(j)}\right)$
and $y^{(k)}_i = 1/\prod_{j=0}^{k-1} c_i\left(A^{(j)}\right)$. Then we have the following facts:
\begin{itemize}
\item For any odd $k \geq 0$, we have 
\begin{align}\label{eq-ub-ri}
\prod_{i\in [n]} r_i\left(A^{(k)}\right) \leq 1
\end{align}
\begin{align}\label{eq-increase-perm-ri}
\perman\left(A^{(k+1)}\right) = \perman\left(A^{(k)}\right)\prod_{i\in [n]} r_i^{-1}\left(A^{(k)}\right). 
\end{align}
Similarly, for any even $k\geq 0$, we have 
\begin{align}\label{eq-ub-ci}
\prod_{i\in [n]} c_i\left(A^{(k)}\right) \leq 1
\end{align}
\begin{align}\label{eq-increase-perm-ci}
\perman\left(A^{(k+1)}\right) = \perman\left(A^{(k)}\right)\prod_{i\in [n]} c_i^{-1}\left(A^{(k)}\right). 
\end{align}
\item 
For any $k \geq 0$, 
\begin{align}\label{eq-ak-aprime-relation}
A^{(k)} = \diag\left(\frac{x^{(k)}_1}{r_1(A)},\dots,
\frac{x^{(k)}_n}{r_n(A)}\right)\cdot A\cdot\diag\left(y^{(k)}_1,\dots,y^{(k)}_n\right).
\end{align}
\end{itemize}
\end{lemma}

\paragraph{Accuracy.}
The following are some key quantities used in our proof.
% \begin{definition}
% A matrix $Z$ of size $n\times n$ is called \emph{standardized} if either $r_{i}(Z) = 1$ for each $i\in [n]$ or $c_{i}(Z) = 1$ for each $i\in [n]$.
% A matrix $Z$ has \emph{column-accuracy} $\boldsymbol{\alpha} = (\alpha_1,\dots,\alpha_n)$ if $r_{i}(Z) = 1$ for each $i\in [n]$ and 
% \begin{align}
% \forall j\in [n], \quad \abs{c_j(Z) - 1}\leq \alpha_j.
% \end{align}
% The definition of the \emph{row-accuracy} is similar.
% We say a matrix $Z$ has accuracy $\boldsymbol{\alpha}$ if $Z$ has column-accuracy $\boldsymbol{\alpha}$ or row-accuracy $\boldsymbol{\alpha}$.
% Given a matrix $Z$ with accuracy $\boldsymbol{\alpha}$,
% define
% \begin{align}
% \alpha(Z) &\triangleq \frac{2}{n}\cdot \sum_{i\in [n]}\alpha_i.\label{eq-def-alpha}
% \end{align}
% \end{definition}
% Intuitively, $\alpha(Z)$ depicts how far $Z$ is from a doubly stochastic matrix. Given a matrix $Z$, when the notation $\alpha(Z)$ is used, we always assume that $Z$ is standardized.
% We say that an $n\times n$ matrix $Z$ has a maximum deviation $t$ if 
% $\abs{r_i(Z) - 1}\leq t$ and $\abs{c_i(Z) - 1}\leq t$ for each $i\in [n]$. 

\begin{definition}
An $n\times n$ matrix $A$ is called \emph{standardized} if either $r_i(A) = 1$ for all $i \in [n]$, or $c_i(A) = 1$ for all $i \in [n]$.
We say $A$ has \emph{column-accuracy} $\boldsymbol{\alpha} = (\alpha_1, \dots, \alpha_n)$ if $r_i(A) = 1$ for all $i \in [n]$ and
$$\forall j \in [n], \quad \alpha_j = \abs{c_j(A) - 1}.$$
The notion of \emph{row-accuracy} is defined similarly.
A matrix has \emph{accuracy} $\boldsymbol{\alpha}$ if it has either column-accuracy or row-accuracy $\boldsymbol{\alpha}$. 
Given a matrix $A$ with accuracy $\boldsymbol{\alpha}$, we define
\begin{align}
\alpha(A) &\triangleq \frac{1}{n}\cdot \sum_{i\in [n]}\alpha_i.\label{eq-def-alpha}
\end{align}
Intuitively, $\alpha(A)$ quantifies how far $A$ is from being a doubly stochastic matrix. Henceforth, whenever the notation $\alpha(A)$ is used, we implicitly assume that $A$ is standardized.
\end{definition}

The following is an easy fact about the permanent of standardized matrices.

\begin{lemma}\label{lem-permanent-standardized-upper-bound}
For any standardized matrix $A\in\mathbb{R}_{\ge 0}^{n\times n}$, we have
\[
\perman(A)\leq 1.
\]
\end{lemma}

\paragraph{Pinsker's inequality.}
For strictly positive probability vectors $p,q\in\mathbb{R}_{>0}^n$, define their Kullback--Leibler divergence by
\[
D_{\mathrm{KL}}(p\|q) \triangleq \sum_{i=1}^n p_i\log\frac{p_i}{q_i}.
\]
Pinsker's inequality states that
\begin{equation}\label{eq-pinsker}
\norm{p-q}_1^2 \leq 2D_{\mathrm{KL}}(p\|q).
\end{equation}

%% file: Reduction.tex
\section{Reduction from \texorpdfstring{$(\boldsymbol{u},\boldsymbol{v})$}{(u,v)}-scaling to \texorpdfstring{$(\mathbf{1},\mathbf{1})$}{(1,1)}-scaling.}\label{sec-upper-bounds}

\subsection{Definition of the Reduction.}\label{sec-def-reduction}

Given an instance of $(\boldsymbol{u}, \boldsymbol{v})$-scaling, the following two definitions, Definitions \ref{def-integer-vector} and \ref{def-splitted-matrix}, serve to reduce it to an instance of $(\mathbf{1}, \mathbf{1})$-scaling.

\Cref{def-integer-vector} first discretizes $\boldsymbol{u},\boldsymbol{v}$ to integer vectors.
With these integer vectors and $A$, \cref{def-splitted-matrix} reduces this instance to another instance of $(\mathbf{1},\mathbf{1})$-scaling by subdividing each entry $A_{i,j}$ into a submatrix with identical subentries.

Given positive vectors $\boldsymbol{u}$ and $\boldsymbol{v}$, we discretize them into positive integer vectors $f_1(\boldsymbol{u},\boldsymbol{v},L)$ and $f_2(\boldsymbol{u},\boldsymbol{v},L)$ by multiplying each coordinate by a large integer $L$ and then rounding via a tailored rule. 
The rounding scheme is designed to satisfy the compatibility condition required by the SK algorithm, namely, $\norm{f_1(\boldsymbol{u},\boldsymbol{v},L)}_1 = \norm{f_2(\boldsymbol{u},\boldsymbol{v},L)}_1$.

\begin{definition}\label{def-integer-vector}
Let $m, n \in \mathbb{Z}_{>0}$,
and let $\boldsymbol{u} \in \mathbb{R}_{>0}^m$ and $\boldsymbol{v} \in \mathbb{R}_{>0}^n$ be vectors satisfying $\norm{\boldsymbol{u}}_1 = \norm{\boldsymbol{v}}_1 = 1$.
For any positive integer $L$, 
let
\[t \triangleq\sum_{i\in [m]} \lfloor L u_i \rfloor- \sum_{i\in [n]}\lfloor L v_i \rfloor.\]
If $t\geq 0$, define
\begin{align*}
f_1(\boldsymbol{u},\boldsymbol{v},L) \triangleq (\lfloor L u_1 \rfloor,\dots,\lfloor L u_m \rfloor),\quad 
f_2(\boldsymbol{u},\boldsymbol{v},L) \triangleq (\lfloor L v_1 \rfloor+1,\dots,\lfloor L v_t \rfloor+1,\lfloor L v_{t+1} \rfloor,\dots,\lfloor L v_n \rfloor).
\end{align*}
We remark that $f_2(\boldsymbol{u},\boldsymbol{v},L)$ is well-defined, because it can be verified that $t<n$ by the inequality
\[\sum_{i\in [m]}\lfloor L u_i \rfloor\leq L\norm{\boldsymbol{u}}_1 = L\norm{\boldsymbol{v}}_1 < \sum_{i\in [n]}(\lfloor L v_i \rfloor+1).\]
Similarly, if 
$t< 0$, define
\begin{align*}
f_1(\boldsymbol{u},\boldsymbol{v},L) \triangleq (\lfloor L u_1 \rfloor+1,\dots,\lfloor L u_{-t} \rfloor+1,\lfloor L u_{-t+1} \rfloor,\dots,\lfloor L u_{m} \rfloor),\quad
f_2(\boldsymbol{u},\boldsymbol{v},L) \triangleq (\lfloor L v_1 \rfloor,\dots,\lfloor L v_n \rfloor).
\end{align*}
It can be verified that $\norm{f_1(\boldsymbol{u},\boldsymbol{v},L)}_1 = \norm{f_2(\boldsymbol{u},\boldsymbol{v},L)}_1$.
We will always choose a sufficiently large integer 
$L$ such that both 
$f_1(\boldsymbol{u},\boldsymbol{v},L)$ and $f_2(\boldsymbol{u},\boldsymbol{v},L)$ are positive vectors.
Furthermore, define 
\[R(\boldsymbol{u},\boldsymbol{v},L) \triangleq \min\left\{\min_{i\in [m]} \lfloor L u_i \rfloor,\min_{j\in[n]} \lfloor L v_j \rfloor\right\}.\]
\end{definition}

Given an instance $\left(\={u},\={v},A\right)$
of $(\={u},\={v})$-scaling where $\={u},\={v}$ are positive integer vectors,
the following definition reduces it to another instance of $(\={1},\={1})$-scaling.

\begin{definition}\label{def-splitted-matrix}
Let $m, n \in \mathbb{Z}_{>0}$. 
Let $\boldsymbol{u} \in \mathbb{Z}_{>0}^m$ and $\boldsymbol{v} \in \mathbb{Z}_{>0}^n$ be vectors satisfying $\norm{\boldsymbol{u}}_1 = \norm{\boldsymbol{v}}_1$
and $A \in \mathbb{R}_{\ge 0}^{m \times n}$ be a nonzero matrix.
For each $i\in [m], j\in [n]$, 
let $S_i = \left[\sum_{k\leq i} u_k\right]\setminus \left[\sum_{k< i} u_k\right]$, $T_j = \left[\sum_{k\leq j} v_k\right]\setminus \left[\sum_{k< j} v_k\right]$.
Define $G(A,\boldsymbol{u},\boldsymbol{v})$ as
the matrix $B$ of size $\norm{\boldsymbol{u}}_1\times \norm{\boldsymbol{v}}_1$ where
\[\forall i\in [m],j\in [n], i'\in S_i,j'\in T_j,\quad  B_{i',j'} = \frac{A_{i,j}}{u_i\cdot v_j}.\]
Intuitively, $G(A,\boldsymbol{u},\boldsymbol{v})$ is the matrix obtained from $A$ by subdividing each entry $A_{i,j}$ into $u_i\cdot v_j$ identical subentries.
\end{definition}

Our reduction from $(\boldsymbol{u},\boldsymbol{v})$-scaling to $(\mathbf{1},\mathbf{1})$-scaling is by discretization and subdivision.
We remark that there is no trivial linear transformation for this reduction.

\subsection{Correctness of the Reduction.}
Utilizing Definitions \ref{def-integer-vector} and \ref{def-splitted-matrix}, we can reduce an instance $(A, (\={u}, \={v}))$ of $(\={u}, \={v})$-scaling to an instance $(G(A, \={u}', \={v}'), (\={1}, \={1}))$ of $(\={1}, \={1})$-scaling, where the scaled integer vectors are given by $\={u}' = f_1(\={u}, \={v}, L)$ and $\={v}' = f_2(\={u}, \={v}, L)$.
In this section, we prove the correctness of this reduction. 
Specifically, we show that for any fixed iteration $k$ of the SK algorithm, by choosing a sufficiently large $L$, the marginal error of the $(\={u}, \={v})$-scaling on $A$ can be made arbitrarily close to $1/L$ times that of the $(\={1}, \={1})$-scaling on the expanded matrix $G(A, \={u}', \={v}')$.

The proof relies on two main insights. 
First, to establish the error bounds, we compare two matrix sequences generated by the SK algorithm: the sequence $B^{(0)}, B^{(1)}, \dots$ obtained by applying $(\={u}, \={v})$-scaling on matrix $A$, and the sequence $D^{(0)}, D^{(1)}, \dots$ resulting from $(\={u}', \={v}')$-scaling on $A$. 
Letting $R = R(\={u}, \={v}, L)$, we can show that the initial ratio between $L \cdot B^{(0)}_{i,j}$ and $D^{(0)}_{i,j}$ is bounded by $R/(R-1)$. 
Because each iteration of the SK algorithm amplifies this ratio by at most a power of 3, the ratio $L \cdot B^{(k)}_{i,j}/D^{(k)}_{i,j}$ at the $k$-th step remains strictly bounded by $\left(R/(R-1)\right)^{3^{k+1}}$.
Consequently, for any fixed $k$, we can choose a sufficiently large $L$ such that $R/(R-1)$ approaches 1. This ensures that the upper bound $\left(R/(R-1)\right)^{3^{k+1}}$ is also arbitrarily close to 1, effectively controlling the discrepancy between $L \cdot B^{(k)}_{i,j}$ and $D^{(k)}_{i,j}$. 
As a result, the marginal error of $B^{(k)}$ can be tightly approximated by $1/L$ times the marginal error of $D^{(k)}$. 
That is, the difference between $\norm{\={r}(B^{(k)}) - \={u}}_{1} + \norm{\={c}(B^{(k)}) - \={v}}_{1}$ and $\frac{1}{L} \left(\norm{\={r}(D^{(k)}) - \={u}'}_{1} + \norm{\={c}(D^{(k)}) - \={v}'}_{1}\right)$ can be made negligibly small (\cref{lem-reduction-error}).

Second, we establish an operational equivalence: performing $(\={u}', \={v}')$-scaling on matrix $A$ via the SK algorithm is strictly equivalent to performing standard $(\={1}, \={1})$-scaling on the expanded matrix $G(A, \={u}', \={v}')$. 
This equivalence can be rigorously verified by tracing the row and column normalization steps throughout the SK iterations (\cref{lemma-reduction-key}).

Combining these two insights yields our main result: for any fixed iteration $k$ of the SK algorithm and sufficiently large $L$, the discrepancy between the weighted scaling on the original matrix and the uniform scaling on the expanded matrix vanishes proportionally to $1/L$.

The main result of this subsection is the following theorem.

\begin{theorem}\label{thm-reduction-precision}
Let $\varepsilon \in (0,1), m, n \in \mathbb{Z}_{>0}$. Let $\boldsymbol{u} \in \mathbb{R}_{>0}^m$ and $\boldsymbol{v} \in \mathbb{R}_{>0}^n$ be vectors such that $\norm{\boldsymbol{u}}_1 = \norm{\boldsymbol{v}}_1 = 1$, and let $B \in \mathbb{R}_{\ge 0}^{m \times n}$ be a matrix with no zero rows or columns. 
For any integer $t \ge 0$ and $L \in \mathbb{Z}_{>0}$ with $R(\boldsymbol{u}, \boldsymbol{v}, L) \geq 2$, let $A'$ and $A$ denote the outputs of the SK algorithm at step $t$ with inputs $(B, (\boldsymbol{u}, \boldsymbol{v}))$ and $(G(B, f_1(\boldsymbol{u}, \boldsymbol{v}, L), f_2(\boldsymbol{u}, \boldsymbol{v}, L)), (\boldsymbol{1}, \boldsymbol{1}))$, respectively.
Then, for any fixed iteration step $t \ge 0$, there exists a sufficiently large integer $\ell$ (which depends on $t$ and $\varepsilon$) such that for all $L \ge \ell$, we have
\begin{align}\label{thm-reduction-precision-condition}
\abs{\norm{\=r\left(A'\right)- \boldsymbol{u}}_{1} + \norm{\=c\left(A'\right)- \boldsymbol{v}}_{1} - \frac{\norm{\=r\left(A\right)- \boldsymbol{1}}_{1} + \norm{\=c\left(A\right)- \boldsymbol{1}}_{1}}{L}}\leq \varepsilon.
\end{align}
Moreover, if $L > 0$ is chosen such that $L\boldsymbol{u}$ and $L\boldsymbol{v}$ are integer vectors, then for any integer $t \ge 0$, we have
\begin{align}\label{thm-reduction-precision-condition-eq}
\norm{\=r\left(A'\right)- \boldsymbol{u}}_{1} + \norm{\=c\left(A'\right)- \boldsymbol{v}}_{1} = \frac{\norm{\=r\left(A\right)- \boldsymbol{1}}_{1} + \norm{\=c\left(A\right)- \boldsymbol{1}}_{1}}{L}.
\end{align}
\end{theorem}

% Let $\varepsilon \in (0,1), m, n \in \mathbb{Z}_{>0}$. 
% Let $\boldsymbol{u} \in \mathbb{R}_{>0}^m$ and $\boldsymbol{v} \in \mathbb{R}_{>0}^n$ be vectors satisfying $\norm{\boldsymbol{u}}_1 = \norm{\boldsymbol{v}}_1 = 1$
% and $B \in \mathbb{R}_{\ge 0}^{m \times n}$ be a nonzero matrix.
% Fix any $t,L \in \mathbb{Z}_{>0}$ with $R(\boldsymbol{u}, \boldsymbol{v}, L) \geq 2$,
% let $A'$ and $A$ denote the outputs of the SK algorithm at step $t$ with inputs $(B, (\boldsymbol{u}, \boldsymbol{v}))$ and $(G(B, f_1(\boldsymbol{u}, \boldsymbol{v}, L), f_2(\boldsymbol{u}, \boldsymbol{v}, L)), (\boldsymbol{1}, \boldsymbol{1}))$, respectively.
% Then, for any integer $t \geq 0$, there exists a sufficiently large $\ell$ such that for each $L\geq \ell$ we have 
% \begin{align}\label{thm-reduction-precision-condition}\abs{\norm{\=r\left(A'\right)- \boldsymbol{u}}_{1} + \norm{\=c\left(A'\right)- \boldsymbol{v}}_{1} - \frac{\norm{\=r\left(A\right)- \boldsymbol{1}}_{1} + \norm{\=c\left(A\right)- \boldsymbol{1}}_{1}}{L}}\leq \varepsilon.\end{align}
% Moreover, given any $L>0$ where $Lu_1,\dots,Lu_m,Lv_1,\dots,Lv_n$ are integers, for any integer $t\geq 0$ we have 
% \begin{align}\label{thm-reduction-precision-condition-eq}
% \norm{\=r\left(A'\right)- \boldsymbol{u}}_{1} + \norm{\=c\left(A'\right)- \boldsymbol{v}}_{1} = \frac{\norm{\=r\left(A\right)- \boldsymbol{1}}_{1} + \norm{\=c\left(A\right)- \boldsymbol{1}}_{1}}{L}.
% \end{align}

To prove \cref{thm-reduction-precision}, it suffices to establish \eqref{thm-reduction-precision-condition}, which follows directly from Lemmas \ref{lem-reduction-error} and \ref{lemma-reduction-key}. The proof of \eqref{thm-reduction-precision-condition-eq} for the case where $Lu_1,\dots,Lu_m,Lv_1,\dots,Lv_n$ are integers proceeds similarly.

The following lemma compares two matrices generated from matrix $A$ at iteration step $t$ of the SK algorithm: matrix $B$, obtained via $(\={u}, \={v})$-scaling, and matrix $D$, obtained via $(f_1(\boldsymbol{u},\boldsymbol{v},L),f_2(\boldsymbol{u},\boldsymbol{v},L))$-scaling. It establishes that for a sufficiently large $L$, the marginal error of $B$ can be tightly approximated by $1/L$ times the marginal error of $D$.

\begin{lemma}\label{lem-reduction-error}
Let $\varepsilon\in (0,1), m, n \in \mathbb{Z}_{>0}$ and $t \in \mathbb{Z}_{\geq 0}$. 
Let $\boldsymbol{u} \in \mathbb{R}_{>0}^m$ and $\boldsymbol{v} \in \mathbb{R}_{>0}^n$ be vectors satisfying $\norm{\boldsymbol{u}}_1 = \norm{\boldsymbol{v}}_1 = 1$
and $A \in \mathbb{R}_{\ge 0}^{m \times n}$ be a matrix with no zero rows or columns.
Given any positive integer $L$ with $R(\boldsymbol{u},\boldsymbol{v},L)\geq 2$,
let $B,D$ be the outputs of SK at step $t$ with input $(A,(\boldsymbol{u},\boldsymbol{v}))$ and $(A,f_1(\boldsymbol{u},\boldsymbol{v},L),f_2(\boldsymbol{u},\boldsymbol{v},L))$, respectively.
Then we have 
\begin{align*}
\abs{\norm{\=r\left(B\right) - \boldsymbol{u}}_{1}+\norm{\=c\left(B\right) - \boldsymbol{v}}_{1} - \frac{\norm{\=r\left(D\right) - \boldsymbol{u}'}_{1}+\norm{\=c\left(D\right) - \boldsymbol{v}'}_{1}}{L}}\leq \frac{m+n}{L} +\left(\frac{R(\boldsymbol{u},\boldsymbol{v},L)}{R(\boldsymbol{u},\boldsymbol{v},L)-1}\right)^{3^{t+1}}-1.
\end{align*}
\end{lemma}
\begin{proof}
Without loss of generality, we assume that $t$ is even.
For simplicity, let $R = R(\boldsymbol{u},\boldsymbol{v},L)$, $\boldsymbol{u}' \triangleq (u'_1,\dots,u'_m) = f_1(\boldsymbol{u},\boldsymbol{v},L)$, $\boldsymbol{v}' \triangleq (v'_1,\dots,v'_n) = f_2(\boldsymbol{u},\boldsymbol{v},L)$.
Let $B^{(0)},B^{(1)},\dots,B^{(t)} = B$ be the generated matrices by SK with $(A, (\boldsymbol{u},\boldsymbol{v}))$ as input. 
Similarly, let $D^{(0)},D^{(1)},\dots,D^{(t)} = D$ be the generated matrices by SK with $(A, (\boldsymbol{u}',\boldsymbol{v}'))$  as input.
For even $t\geq 2$, we have 
\begin{align*}
\forall i\in [m],\quad \sum_{j\in [n]}B^{(t)}_{i,j} = \sum_{j\in [n]}\frac{u_iB^{(t-1)}_{i,j}}{r_i\left(B^{(t-1)}
\right)} = u_i\cdot \sum_{j\in [n]}\frac{B^{(t-1)}_{i,j}}{r_i\left(B^{(t-1)}\right)} = u_i.
\end{align*}
The same row-sum identity holds at $t=0$ by the initialization of SK.
Thus, we have
$\norm{r\left(B\right)- \boldsymbol{u}}_{1} = \norm{r\left(B^{(t)}\right)- \boldsymbol{u}}_{1} = 0$.
Similarly, we also have 
$\norm{r\left(D\right)- \boldsymbol{u}'}_{1} = 0$.
Furthermore, define
\[
\forall k\geq 0, \quad \alpha(k) \triangleq \left(\frac{R}{R-1}\right)^{3^{k+1}}.
\]
We claim
\begin{equation}\label{eq-reduction-errorbound}
\begin{aligned}
\forall i\in[m],j\in [n], k\in \{0,1,\cdots,t\}, &\quad L\cdot B^{(k)}_{i,j}/\alpha(k)\leq D^{(k)}_{i,j}\leq L\cdot B^{(k)}_{i,j} \cdot \alpha(k).
\end{aligned}
\end{equation}
For each $j\in [n]$, by $\alpha(t)>1$ and \eqref{eq-reduction-errorbound} 
we have
\begin{align*}
\left(\sum_{i\in [m]} D^{(t)}_{i,j}\right) - Lv_j - \left(\left(\sum_{i\in [m]}L\cdot B^{(t)}_{i,j}\right) - Lv_j\right) = \left(\sum_{i\in [m]} D^{(t)}_{i,j}\right)  - \left(\sum_{i\in [m]}L\cdot B^{(t)}_{i,j}\right)
\leq \sum_{i\in [m]}L\cdot B^{(t)}_{i,j}\left(\alpha(t) - 1\right).
\end{align*}
By $\alpha(t)>1$ and \eqref{eq-reduction-errorbound}, we also have 
\begin{align*}
&\quad \left(\sum_{i\in [m]}L\cdot B^{(t)}_{i,j}\right) - Lv_j - \left(\left(\sum_{i\in [m]} D^{(t)}_{i,j}\right) - Lv_j\right)  
= \left(\sum_{i\in [m]}L\cdot B^{(t)}_{i,j}\right) - \left(\sum_{i\in [m]} D^{(t)}_{i,j}\right) \\
&\leq \sum_{i\in [m]}L\cdot B^{(t)}_{i,j}\left(1 - \frac{1}{\alpha(t)}\right) \leq 
\sum_{i\in [m]}L\cdot B^{(t)}_{i,j}\left(\alpha(t) - 1\right).
\end{align*}
In summary, we always have 
\begin{align}\label{eq-abssumiinm-lbtijminuslvj}
&\quad \abs{\left(\sum_{i\in [m]}L\cdot B^{(t)}_{i,j}\right) - Lv_j - \left(\left(\sum_{i\in [m]} D^{(t)}_{i,j}\right) - Lv_j\right)} \leq 
\sum_{i\in [m]}L\cdot B^{(t)}_{i,j}\left(\alpha(t) - 1\right).
\end{align}
Moreover, by \cref{def-integer-vector} we have
\[
\sum_{j\in [n]}\abs{Lv_j - v'_j} \leq 1\cdot n = n.
\]
Thus,
\begin{align*}
\sum_{j\in [n]}\abs{\left(\sum_{i\in [m]}D^{(t)}_{i,j}\right) - Lv_j} \leq \sum_{j\in [n]}\abs{\left(\sum_{i\in [m]}D^{(t)}_{i,j}\right) - v'_j} + \sum_{j\in [n]}\abs{Lv_j - v'_j} \leq n+\sum_{j\in [n]}\abs{\left(\sum_{i\in [m]}D^{(t)}_{i,j}\right) - v'_j}.
\end{align*}
Combined with \eqref{eq-abssumiinm-lbtijminuslvj}, we have
\begin{align*}
&\quad \sum_{j\in [n]}\abs{\left(\sum_{i\in [m]}L\cdot B^{(t)}_{i,j}\right) - Lv_j - \left(\left(\sum_{i\in [m]}D^{(t)}_{i,j}\right) - v'_j\right)} \\
&\leq n+\sum_{j\in [n]}\abs{\left(\sum_{i\in [m]}L\cdot B^{(t)}_{i,j}\right) - Lv_j - \left(\left(\sum_{i\in [m]}D^{(t)}_{i,j}\right) - L v_j\right)} \\
&\leq n + \sum_{j\in [n]}\sum_{i\in [m]}L\cdot B^{(t)}_{i,j}\left(\alpha(t)-1\right).
\end{align*}
Therefore, 
\begin{align*}
&\quad \abs{L\norm{\=c\left(B^{t}\right) - \boldsymbol{v}}_{1} - \norm{\=c\left(D^{t}\right) - \boldsymbol{v}'}_{1}}\leq \sum_{j\in [n]}\abs{\left(\sum_{i\in [m]}L\cdot B^{(t)}_{i,j}\right) - Lv_j - \left(\left(\sum_{i\in [m]}D^{(t)}_{i,j}\right) - v'_j\right)} \\
&\leq n + \sum_{j\in [n]}\sum_{i\in [m]}L\cdot B^{(t)}_{i,j}\left(\alpha(t)-1\right)
= n+ L\norm{\=r\left(B^{(t)}\right)}_{1}\left(\alpha(t)-1\right).
\end{align*}
Combined with $B = B^{(t)}$ and $D = D^{(t)}$, we have 
\begin{align*}
\abs{\norm{\=c\left(B\right) - \boldsymbol{v}}_{1} - \frac{1}{L}\cdot \norm{\=c\left(D\right) - \boldsymbol{v}'}_{1}  } \leq \frac{n}{L} + \norm{\=r\left(B\right)}_{1}(\alpha(t)-1).
\end{align*}
Recall that $\norm{\=r\left(B\right)- \boldsymbol{u}}_{1} = \norm{\=r\left(D\right)- \boldsymbol{u}'}_{1} = 0$ and $\norm{\=r\left(B\right)}_{1} = \norm{ \boldsymbol{u}}_{1}$ = 1.
Thus, the lemma follows immediately once we establish \eqref{eq-reduction-errorbound}. In the remainder of this proof, we demonstrate \eqref{eq-reduction-errorbound} by induction, thereby completing the argument.

The base step is $k = 0$. 
For any $i\in[m], j\in [n]$, we have
\begin{align*}
D^{(0)}_{i,j}  &= \frac{u'_iA_{i,j}}{r_{i}(A)} \geq \frac{\lfloor Lu_i \rfloor A_{i,j}}{r_{i}(A)} \geq \frac{Lu_iA_{i,j}}{r_{i}(A)}\cdot \frac{\lfloor Lu_i \rfloor}{\lfloor Lu_i \rfloor + 1} \geq \frac{Lu_iA_{i,j}}{r_{i}(A)}\cdot \frac{R}{R + 1} = L\cdot B^{(0)}_{i,j}\cdot \frac{R}{R + 1}\geq L\cdot B^{(0)}_{i,j}\left(1 - \frac{1}{R}\right),\\
D^{(0)}_{i,j}  &= \frac{u'_iA_{i,j}}{r_{i}(A)} \leq \frac{\left(\lfloor Lu_i \rfloor + 1\right) A_{i,j}}{r_{i}(A)} \leq \frac{Lu_iA_{i,j}}{r_{i}(A)}\frac{\lfloor Lu_i \rfloor+1}{\lfloor Lu_i \rfloor} \leq \frac{Lu_iA_{i,j}}{r_{i}(A)} \frac{R+1}{R} = L\cdot B^{(0)}_{i,j}\frac{R+1}{R}\leq L\cdot B^{(0)}_{i,j}\left(1 - \frac{1}{R}\right)^{-1}.
\end{align*}
Thus, \eqref{eq-reduction-errorbound} is immediate for $k = 0$.
The base step is proved.

For the inductive step where $k > 0$, without loss of generality, we assume $k$ is odd.
Thus, for any $j\in [n]$, we have
\begin{align*}
D^{(k)}_{i,j}  &= \frac{v'_jD^{(k-1)}_{i,j}}{c_{j}\left(D^{(k-1)}\right)} \geq \frac{\lfloor Lv_j \rfloor D^{(k-1)}_{i,j}}{c_{j}\left(D^{(k-1)}\right)} \geq \frac{Lv_jD^{(k-1)}_{i,j}}{c_{j}\left(D^{(k-1)}\right)}\cdot \frac{\lfloor Lv_j \rfloor}{\lfloor Lv_j \rfloor + 1} \\
&\geq \frac{Lv_jB^{(k-1)}_{i,j}/\alpha(k-1)}{\alpha(k-1)c_{j}\left(B^{(k-1)}\right)}\cdot \frac{R}{R + 1} 
\geq \frac{L\cdot B^{(k)}_{i,j}}{(\alpha(k-1))^{2}}\cdot \left(1-\frac{1}{R}\right) \geq \frac{L\cdot B^{(k)}_{i,j}}{\alpha(k)},
\end{align*}
where the third inequality follows from the induction hypothesis.
Similarly, we also have 
\begin{align*}
D^{(k)}_{i,j}  &= \frac{v'_jD^{(k-1)}_{i,j}}{c_{j}\left(D^{(k-1)}\right)} \leq \frac{\left(\lfloor Lv_j \rfloor +1\right)D^{(k-1)}_{i,j}}{c_{j}\left(D^{(k-1)}\right)} \leq \frac{Lv_jD^{(k-1)}_{i,j}}{c_{j}\left(D^{(k-1)}\right)}\cdot \frac{\lfloor Lv_j \rfloor+1}{\lfloor Lv_j \rfloor} \\
&\leq 
\frac{Lv_j\cdot \alpha(k-1)\cdot B^{(k-1)}_{i,j}}{c_{j}\left(B^{(k-1)}\right)/\alpha(k-1)}\cdot \frac{R+1}{R} 
\leq L\cdot B^{(k)}_{i,j} \cdot (\alpha(k-1))^2 \cdot \left(1-\frac{1}{R}\right)^{-1} \leq  L\cdot B^{(k)}_{i,j} \cdot \alpha(k),
\end{align*}
where the third inequality follows from the induction hypothesis.
Combining the above two inequalities, we see that \eqref{eq-reduction-errorbound} holds for step $k$. This completes the inductive step. Therefore, \eqref{eq-reduction-errorbound} is established, which concludes the proof of the lemma.
\end{proof}

The following lemma establishes that for any $\={u} \in \mathbb{Z}_{>0}^m$, $\={v} \in \mathbb{Z}_{>0}^n$, and $B \in \mathbb{R}_{\ge 0}^{m \times n}$,
performing $(\={u}, \={v})$-scaling on $B$ via the SK algorithm is strictly equivalent to performing standard $(\={1}, \={1})$-scaling on the expanded matrix $G(B, \={u}, \={v})$. 

\begin{lemma}\label{lemma-reduction-key}
Let $m,n\in \mathbb{Z}_{>0}$ and $t\in \mathbb{Z}_{\geq 0}$.
Let $\={u} \in \mathbb{Z}_{>0}^m$ and $\={v} \in \mathbb{Z}_{>0}^n$ be vectors satisfying $\norm{\={u}}_1 = \norm{\={v}}_1$
and $B \in \mathbb{R}_{\ge 0}^{m \times n}$ be a matrix with no zero rows or columns.
Let $A$ and $A'$ denote the outputs of SK at step $t$, using the input pairs $(G(B,\boldsymbol{u},\boldsymbol{v}), (\boldsymbol{1},\boldsymbol{1}))$ and $(B, (\boldsymbol{u}, \boldsymbol{v}))$, respectively.
Then
\begin{align}\label{eq-rcscaling-condition}
\norm{\=r\left(A\right)- \boldsymbol{1}}_{1} + \norm{\=c\left(A\right)- \boldsymbol{1}}_{1} = \norm{\=r\left(A'\right)- \boldsymbol{u}}_{1} + \norm{\=c\left(A'\right)- \boldsymbol{v}}_{1}.
\end{align}
\end{lemma}

\begin{proof}
For simplicity, let $C \triangleq G(B,\boldsymbol{u},\boldsymbol{v})$. 
Let $B^{(0)}, B^{(1)}, \ldots, B^{(t)} = A'$ and $C^{(0)}, C^{(1)}, \ldots, C^{(t)} = A$ denote the sequences of matrices generated by the SK algorithm on inputs $(B,(\boldsymbol{u},\boldsymbol{v}))$ and $(C,(\boldsymbol{1},\boldsymbol{1}))$, respectively.
For each $i\in [m], j\in [n]$, 
define $S_i = \left[\sum_{k\leq i} u_k\right]\setminus \left[\sum_{k< i} u_k\right]$, $T_j = \left[\sum_{k\leq j} v_k\right]\setminus \left[\sum_{k< j} v_k\right]$.
We claim that
\begin{align}\label{eq-reduction-key}
\forall i\in [m],j\in [n],k\in \{0,\ldots,t\},h,h'\in S_i, \ell,\ell'\in T_j, \quad C^{(k)}_{h',\ell'} = C^{(k)}_{h,\ell}, \quad B^{(k)}_{i,j} = \sum_{i'\in S_i, j'\in T_j}C^{(k)}_{i',j'}.
\end{align}
Hence, by $C^{(t)} = A$ and $B^{(t)} = A'$ we have
\begin{align}\label{eq-def-aprime}
\forall i\in [m],j\in [n],\quad A'_{i,j} = B^{(t)}_{i,j}  = \sum_{i'\in S_i, j'\in T_j}C^{(t)}_{i',j'} =  \sum_{i'\in S_i, j'\in T_j}A_{i',j'}.
\end{align}
Thus,
\begin{equation*}
\begin{aligned}
\forall i\in [m], \quad r_i\left(A'\right) =\sum_{j\in [n]}A'_{i,j} =
\sum_{j\in [n]} \sum_{i'\in S_i,j'\in T_j,}A_{i',j'} = \sum_{i'\in S_i}\sum_{j\in [n],j'\in T_j} A_{i',j'}  = \sum_{i'\in S_i}r_{i'}\left(A\right).
\end{aligned}
\end{equation*}
Therefore, 
\begin{align}\label{eq-riaprime-minusui}
\forall i\in [m], \quad r_i\left(A'\right) - u_i = r_i\left(A'\right) - \abs{S_i} = \sum_{i'\in S_i}(r_{i'}\left(A\right) - 1).
\end{align}
In addition, by \eqref{eq-reduction-key} and $C^{(t)} = A$, for any $i\in[m],j\in [n], h,h'\in S_i, \ell\in T_j$ we have
$A_{h',\ell} = A_{h,\ell}$.
Thus, $r_{h'}(A) = r_{h}(A)$, which implies that for all $i' \in S_i$, the values $r_{i'}(A) - 1$ share the same sign.
Combined with \eqref{eq-riaprime-minusui}, we have
$\abs{r_i\left(A'\right) - u_i} = \sum_{i'\in S_i}\abs{r_{i'}\left(A\right) - 1}$.
Hence, 
\begin{align*}
\norm{\=r\left(A\right)- \boldsymbol{1}}_{1} = \norm{\=r\left(A'\right)- \boldsymbol{u}}_{1}.
\end{align*}
Similarly, we also have 
\begin{align*}
\norm{\=c\left(A\right)- \boldsymbol{1}}_{1} =  \norm{\=c\left(A'\right)- \boldsymbol{v}}_{1}.
\end{align*}
Thus, \eqref{eq-rcscaling-condition} follows immediately from the two identities above. Next, we prove \eqref{eq-reduction-key} by induction, which completes the proof of the lemma.

The base step is $k = 0$. 
For any $i\in[m],j\in [n], h,h'\in S_i, \ell,\ell'\in T_j$, 
by $C = G(B,\boldsymbol{u},\boldsymbol{v})$ and \cref{def-splitted-matrix},
we have 
$C_{h',\ell'} = C_{h,\ell}$.
Thus, we also have
$r_{h'}(C) = r_{h}(C)$.
Therefore,
\begin{align*}
\quad C^{(0)}_{h',\ell'} = \frac{C_{h',\ell'}}{r_{h'}(C)} = \frac{C_{h,\ell}}{r_{h}(C)} =  C^{(0)}_{h,\ell}.
\end{align*}
Moreover, by \cref{def-splitted-matrix}, we also have
$B_{i,j} = \abs{T_j}\cdot \abs{S_i}\cdot C_{h,\ell}$.
Thus,
\begin{align*}
B^{(0)}_{i,j} &= u_i \cdot \frac{ B_{i,j}}{r_i(B)} = u_i \cdot \frac{\abs{T_j}\cdot \abs{S_i}\cdot C_{h,\ell}}{\sum_{i'\in S_i} r_{i'}(C)}
= u_i\cdot \frac{\abs{T_j}\cdot \abs{S_i}\cdot C_{h,\ell}}{\abs{S_i}\cdot r_{h}(C)}
=  \frac{\abs{T_j}\cdot \abs{S_i}\cdot C_{h,\ell}}{r_{h}(C)}\\
& =  \abs{T_j}\cdot \abs{S_i}\cdot C^{(0)}_{h,\ell} =   \sum_{i'\in S_i, j'\in T_j}C^{(0)}_{i',j'}.
\end{align*}
The base step is proved.

For the inductive step where $k > 0$, without loss of generality, we assume that $k$ is odd.
By the induction hypothesis, for any $i\in[m],j\in [n], h,h'\in S_i, \ell,\ell'\in T_j$ we have
$C^{(k-1)}_{h',\ell'} = C^{(k-1)}_{h,\ell}$.
Thus, we also have $c_{\ell'}\left(C^{(k-1)}\right) = c_{\ell}\left(C^{(k-1)}\right)$.
Therefore, 
\begin{align}\label{eq-reduction-key-first}
C^{(k)}_{h',\ell'} = \frac{C^{(k-1)}_{h',\ell'}}{c_{\ell'}\left(C^{(k-1)}\right)} =\frac{C^{(k-1)}_{h,\ell}}{c_{\ell}\left(C^{(k-1)}\right)} = C^{(k)}_{h,\ell}.
\end{align}
Moreover, by the induction hypothesis, 
we also have 
\begin{align}\label{eq-reduction-key-second}
B^{(k-1)}_{i,j} =\sum_{i'\in S_i, j'\in T_j}C^{(k-1)}_{i',j'} = \abs{T_j}\cdot \abs{S_i}\cdot C^{(k-1)}_{h,\ell}.
\end{align}
Furthermore, by $c_{\ell'}\left(C^{(k-1)}\right) = c_{\ell}\left(C^{(k-1)}\right)$ for each $h,h'\in S_i, \ell,\ell'\in T_j$,
we have 
\[\sum_{j'\in T_j}c_{j'}\left(C^{(k-1)}\right) = \abs{T_j}\cdot c_{\ell}\left(C^{(k-1)}\right).\]
Combined with the induction hypothesis, we have
\begin{equation}
\begin{aligned}\label{eq-reduction-key-third}
c_j\left(B^{(k-1)}\right) &=\sum_{i\in [m]}B^{(k-1)}_{i,j} =
\sum_{i\in [m]} \sum_{i'\in S_i, j'\in T_j}C^{(k-1)}_{i',j'} = \sum_{j'\in T_j}\sum_{i\in [m],i'\in S_i} C^{(k-1)}_{i',j'}  = \sum_{j'\in T_j}c_{j'}\left(C^{(k-1)}\right) \\
&= \abs{T_j}\cdot c_{\ell}\left(C^{(k-1)}\right).
\end{aligned}
\end{equation}
Therefore, by \eqref{eq-reduction-key-first}, \eqref{eq-reduction-key-second}, and \eqref{eq-reduction-key-third}, we have
\begin{align*}
\quad B^{(k)}_{i,j} = v_j \cdot \frac{ B^{(k-1)}_{i,j}}{c_j\left(B^{(k-1)}\right)} =  v_j \cdot \frac{\abs{T_j}\cdot \abs{S_i}\cdot C^{(k-1)}_{h,\ell}}{\abs{T_j}\cdot c_{\ell}\left(C^{(k-1)}\right)} =  \frac{\abs{T_j}\cdot \abs{S_i}\cdot C^{(k-1)}_{h,\ell}}{c_{\ell}\left(C^{(k-1)}\right)} = \abs{T_j}\cdot \abs{S_i}\cdot C^{(k)}_{h,\ell}
=  \sum_{i'\in S_i, j'\in T_j}C^{(k)}_{i',j'},
\end{align*}
The induction is finished, and the lemma is proved.
\end{proof}

Now we can prove \cref{thm-reduction-precision}.
\begin{proof}[Proof of \cref{thm-reduction-precision}]
Choose an integer $\ell$ sufficiently large such that 
$$
R(\boldsymbol{u},\boldsymbol{v},\ell)\ge 2,\quad\quad \frac{m+n}{\ell} \leq \frac{\varepsilon}{2},
\quad\quad
\left(\frac{R(\boldsymbol{u},\boldsymbol{v},\ell)}{R(\boldsymbol{u},\boldsymbol{v},\ell)-1}\right)^{3^{t+1}}-1\leq\frac{\varepsilon}{2}.
$$
Given any $L\geq \ell$, for simplicity, let $R = R(\boldsymbol{u},\boldsymbol{v},L)$, $\boldsymbol{u}' \triangleq (u'_1,\dots,u'_m) = f_1(\boldsymbol{u},\boldsymbol{v},L)$, $\boldsymbol{v}' \triangleq (v'_1,\dots,v'_n) = f_2(\boldsymbol{u},\boldsymbol{v},L)$.
By $L \geq \ell$, one can verify that
$$
R\ge 2,\quad\quad \frac{m+n}{L} \leq \frac{\varepsilon}{2},
\quad\quad
\left(\frac{R}{R-1}\right)^{3^{t+1}}-1\leq\frac{\varepsilon}{2}.
$$
Let $A''$ be the outputs of SK at step $t$ with input $(B,(\boldsymbol{u}' ,\boldsymbol{v}') )$.
By \cref{lem-reduction-error}, we have 
\begin{align*}
\abs{\norm{\=r\left(A'\right)- \boldsymbol{u}}_{1} + \norm{\=c\left(A'\right)- \boldsymbol{v}}_{1} - \frac{1}{L} \cdot \left(\norm{\=r\left(A''\right) - \boldsymbol{u}'}_{1}+\norm{\=c\left(A''\right) - \boldsymbol{v}'}_{1}\right)}
\leq \frac{m+n}{L} +\left(\frac{R}{R-1}\right)^{3^{t+1}}-1\leq
\varepsilon.
\end{align*}
Moreover, by \cref{lemma-reduction-key} we have
\begin{align*}
(\norm{\=r\left(A''\right)- \boldsymbol{u}'}_{1} + \norm{\=c\left(A''\right)- \boldsymbol{v}'}_{1}) = \norm{\=r\left(A\right)- \boldsymbol{1}}_{1} + \norm{\=c\left(A\right)- \boldsymbol{1}}_{1}.
\end{align*}
Thus, the lemma is immediate by combining the above two inequalities.
\end{proof}

\subsection{Dynamics of the Dense Structure under Reduction.}

Given an instance $(A, (\boldsymbol{u}, \boldsymbol{v}))$ of $(\boldsymbol{u}, \boldsymbol{v})$-scaling, we can reduce it to an instance $(B, (\boldsymbol{1}, \boldsymbol{1}))$ of standard uniform scaling, utilizing Definitions \ref{def-integer-vector} and \ref{def-splitted-matrix},
where $\boldsymbol{u}' = f_1(\boldsymbol{u}, \boldsymbol{v}, L)$, $\boldsymbol{v}' = f_2(\boldsymbol{u}, \boldsymbol{v}, L)$, $B = G(A, \boldsymbol{u}', \boldsymbol{v}')$.
However, a critical challenge arises during this reduction: even if the original matrix $A$ is dense with respect to  $(\boldsymbol{u}, \boldsymbol{v})$, the expanded matrix $B$ is generally not dense with respect to  $(\boldsymbol{1}, \boldsymbol{1})$. 
To successfully bound the iteration complexity of the SK algorithm on the input $(B, (\boldsymbol{1}, \boldsymbol{1}))$, we establish the following key properties regarding the dynamics of the dense structure.

First, we demonstrate that while $B$ loses its density with respect to  $(\boldsymbol{1}, \boldsymbol{1})$, the dense structure can be recovered through appropriate row and column scalings. 
Concretely, we show that for a sufficiently large $L$, the scaled matrix $C = \+D(\boldsymbol{u}') \cdot B \cdot \+D(\boldsymbol{v}')$ recovers the dense structure of $A$, becoming dense with respect to  $(\boldsymbol{1}, \boldsymbol{1})$ (\cref{lemma-gamma-rho-dense-error}), where $\+D(\cdot)$ is defined in \eqref{eq-def-diag}.
The intuition behind this is twofold:
\begin{itemize}
\item By choosing a sufficiently large $L$, the integer vectors $(\boldsymbol{u}', \boldsymbol{v}')$ become arbitrarily close to $(L\boldsymbol{u}, L\boldsymbol{v})$. According to \cref{def-gamma-rho-dense}, if $A$ is dense with respect to $(\boldsymbol{u}, \boldsymbol{v})$, it strictly preserves this density with respect to the scaled targets $(L\boldsymbol{u}, L\boldsymbol{v})$.
Due to this arbitrary closeness, it follows that $A$ is also dense with respect to $(\boldsymbol{u}', \boldsymbol{v}')$.

\item Furthermore, based on the construction of the expanded matrix (\cref{def-splitted-matrix}), $B$ partitions each entry $A_{i,j}$ into a block of $u'_i \times v'_j$ sub-entries, each holding the value $A_{i,j}/(u'_i \times v'_j)$. 
Therefore, the scaling operation $C = \+D(\boldsymbol{u}') \cdot B \cdot \+D(\boldsymbol{v}')$ effectively scales each sub-entry back up, restoring the values to the original $A_{i,j}$ elements. 
Thus, if $A$ is dense with respect to $(\boldsymbol{u}', \boldsymbol{v}')$, it inherently implies that $C$ is dense with respect to $(\boldsymbol{1}, \boldsymbol{1})$ (\cref{lem-reduction-dense-structure}).
\end{itemize}

Second, to bound the iteration complexity of SK on the input matrix $B$, we must meticulously characterize the discrepancy between $B$ and the nicely structured matrix $C$. The raw absolute deviation between the entries of $B$ and $C$ can be arbitrarily large because the target marginals $\boldsymbol{u}$ and $\boldsymbol{v}$ may contain extremely large or small values.
However, we prove that this deviation is significantly mitigated after applying row-normalization. 
Specifically, by comparing the corresponding elements in the row-normalized matrices, we prove that the ratio $(C_{i,j}/r_i(C))/(B_{i,j}/r_i(B))$ on their common positive support is bounded by $O(n)$ for fixed density parameters and sufficiently large $L$ (\cref{thm-dynamic-dense}).

Together, these two insights provide a foundational characterization of the dynamics under reduction, allowing us to precisely measure the extent to which the reduced matrix $B$ deviates from being dense under $(\={1}, \={1})$-scaling. 
Consequently, this structural preservation and discrepancy control equip us with the mathematical tools needed to rigorously bound the iteration complexity of the SK algorithm on the reduced instance $(B, (\boldsymbol{1}, \boldsymbol{1}))$.

Finally, to bound the iteration complexity on the pre-scaled input $(G(\diag(\boldsymbol{u}) \cdot A \cdot \diag(\boldsymbol{v}), \boldsymbol{u}', \boldsymbol{v}'), (\boldsymbol{1}, \boldsymbol{1}))$, we establish \cref{lemma-gamma-rho-dense-error-prenormalization}, which follows a completely analogous proof structure to \cref{lemma-gamma-rho-dense-error}.

As discussed above, while the reduction maps the original instance $(A, (\boldsymbol{u}, \boldsymbol{v}))$ to one of $(\boldsymbol{1}, \boldsymbol{1})$-scaling, the expanded matrix $G(A, f_1(\boldsymbol{u}, \boldsymbol{v}, L), f_2(\boldsymbol{u}, \boldsymbol{v}, L))$ generally loses the dense property of $A$. The following lemma formalizes our first key insight: provided the parameter $L$ is sufficiently large, the dense structure of the original matrix can be rigorously restored via appropriate row and column scalings.

\begin{lemma}\label{lemma-gamma-rho-dense-error}
Let $\gamma \in (0,1]$, $\gamma'\in(1-\gamma,1]$, $\rho\in (0,1]$, $m,n\in \mathbb{Z}_{>0}$.
Let $\boldsymbol{u} \in \mathbb{R}_{>0}^m$ and $\boldsymbol{v} \in \mathbb{R}_{>0}^n$ be vectors satisfying $\norm{\boldsymbol{u}}_1 = \norm{\boldsymbol{v}}_1 = 1$,
and let $A \in \mathbb{R}_{\geq 0}^{m \times n}$ be a nonzero matrix.
If $A$ is at least $(\gamma,\gamma',\rho)$-dense with respect to $(\boldsymbol{u},\boldsymbol{v})$, then there exists a sufficiently large $\ell>0$, such that for each integer $L\geq \ell$,
$\+D(\boldsymbol{u}')\cdot G(A,\boldsymbol{u}' ,\boldsymbol{v}') \cdot \+D(\boldsymbol{v}')$ is at least $( \alpha\gamma,  \alpha\gamma',\rho)$-dense with respect to  $(\boldsymbol{1},\boldsymbol{1})$ where $\alpha =  (\gamma+\gamma'+1)/(2(\gamma+\gamma')) $ , $\boldsymbol{u}' = f_1(\boldsymbol{u},\boldsymbol{v},L)$, and $\boldsymbol{v}' = f_2(\boldsymbol{u},\boldsymbol{v},L)$.
\end{lemma}

In \cref{def-splitted-matrix}, even if \(A\) is dense with respect to  \((\boldsymbol{u},\boldsymbol{v})\), the matrix \(G(A,\boldsymbol{u},\boldsymbol{v})\) need not remain dense with respect to  \((\mathbf{1},\mathbf{1})\), since its entries are normalized by different scaling factors. In other words, the reduction in \cref{def-splitted-matrix} may destroy the dense structure of the original matrix. 
Fortunately, as shown in the following lemma, the resulting matrix \(G(A,\boldsymbol{u},\boldsymbol{v})\) can be made dense again via appropriate row and column scalings.

\begin{lemma}\label{lem-reduction-dense-structure}
Assume the conditions in \cref{def-splitted-matrix}.
Let $D = \+D(\=u) \cdot G(A,\boldsymbol{u},\boldsymbol{v})\cdot \+D(\=v)$. 
Let $\gamma,\gamma'\in  (0,1]$.
Then $D$ is of size $\norm{\boldsymbol{u}}_1\times \norm{\boldsymbol{v}}_1$ where
\[\forall i\in [m],j\in [n], i'\in S_i,j'\in T_j,\quad  D_{i',j'} = A_{i,j}.\]
Thus, if 
$A$ is $(\gamma,\gamma',\rho)$-dense with respect to  $(\boldsymbol{u},\boldsymbol{v})$, then $D$ is  $(\gamma,\gamma',\rho)$-dense with respect to  $(\boldsymbol{1},\boldsymbol{1})$. 
\end{lemma}

\Cref{lem-reduction-dense-structure} is immediate by Definitions \ref{def-splitted-matrix} and \ref{def-gamma-rho-dense}.

Now we can prove \cref{lemma-gamma-rho-dense-error}.

\begin{proof}[Proof of \cref{lemma-gamma-rho-dense-error}]
By \cref{lem-reduction-dense-structure}, we have
if $A$ is at least $( \alpha\gamma,  \alpha\gamma',\rho)$-dense with respect to  $(\boldsymbol{u}' ,\boldsymbol{v}' )$, then 
$\+D(\boldsymbol{u}') \cdot G(A,\boldsymbol{u}' ,\boldsymbol{v}') \cdot \+D(\boldsymbol{v}' )$ is at least $( \alpha\gamma,  \alpha\gamma',\rho)$-dense with respect to  $(\boldsymbol{1},\boldsymbol{1})$.
Thus, to prove this lemma, it is sufficient to show that $A$ is at least $( \alpha\gamma,  \alpha\gamma',\rho)$-dense with respect to  $(\boldsymbol{u}' ,\boldsymbol{v}' )$.
In the following, we prove this conclusion.

Note that each of the following functions increases monotonically as $\ell$ increases:
\begin{align}\label{eq-def-three-functions}
R(\boldsymbol{u}, \boldsymbol{v}, \ell),\quad \frac{R(\boldsymbol{u},\boldsymbol{v},\ell)}{R(\boldsymbol{u},\boldsymbol{v},\ell)+1}\cdot \frac{\ell\norm{\boldsymbol{v}}_1}{n+\ell\norm{\boldsymbol{v}}_1},\quad \frac{R(\boldsymbol{u},\boldsymbol{v},\ell)}{R(\boldsymbol{u},\boldsymbol{v},\ell)+1}\cdot \frac{\ell\norm{\boldsymbol{u}}_1}{m+\ell\norm{\boldsymbol{u}}_1}.
\end{align}
In addition, we have 
\[\lim_{\ell\rightarrow \infty}R(\boldsymbol{u}, \boldsymbol{v}, \ell) = \infty,\quad \lim_{\ell\rightarrow \infty}\frac{R(\boldsymbol{u},\boldsymbol{v},\ell)}{R(\boldsymbol{u},\boldsymbol{v},\ell)+1}\cdot \frac{\ell\norm{\boldsymbol{v}}_1}{n+\ell\norm{\boldsymbol{v}}_1} = 1,\quad \lim_{\ell\rightarrow \infty}\frac{R(\boldsymbol{u},\boldsymbol{v},\ell)}{R(\boldsymbol{u},\boldsymbol{v},\ell)+1}\cdot \frac{\ell\norm{\boldsymbol{u}}_1}{m+\ell\norm{\boldsymbol{u}}_1} = 1.\]
Combined with $\alpha =  (\gamma+\gamma'+1)/(2(\gamma+\gamma')) < 1$, one can choose an integer $\ell$ sufficiently large such that
\[R(\boldsymbol{u},\boldsymbol{v},\ell) > 10^4, \quad   \frac{R(\boldsymbol{u},\boldsymbol{v},\ell)}{R(\boldsymbol{u},\boldsymbol{v},\ell)+1}\cdot \frac{\ell\norm{\boldsymbol{v}}_1}{n+\ell\norm{\boldsymbol{v}}_1} \geq  \alpha, \quad \frac{R(\boldsymbol{u},\boldsymbol{v},\ell)}{R(\boldsymbol{u},\boldsymbol{v},\ell)+1}\cdot \frac{\ell\norm{\boldsymbol{u}}_1}{m+\ell\norm{\boldsymbol{u}}_1} \geq  \alpha.
\]
For simplicity, given a fixed $L\geq \ell$,
let $R = R(\boldsymbol{u},\boldsymbol{v},L)$, $\boldsymbol{u} = (u_1,\dots,u_m), \boldsymbol{v} = (v_1,\dots,v_n)$, $\boldsymbol{u}' = (u'_1,\dots,u'_m) $, $\boldsymbol{v}'= (v'_1,\dots,v'_n)$.
By the monotonicity of the functions in \eqref{eq-def-three-functions} and the fact that $L\geq \ell$, we have 
\begin{align}\label{eq-three-functions-conditions}
R \geq R(\boldsymbol{u},\boldsymbol{v},\ell) > 10^4, \quad   \frac{R}{R+1}\cdot \frac{L\norm{\boldsymbol{v}}_1}{n+L\norm{\boldsymbol{v}}_1} \geq  \alpha, \quad \frac{R}{R+1}\cdot \frac{L\norm{\boldsymbol{u}}_1}{m+L\norm{\boldsymbol{u}}_1} \geq  \alpha.
\end{align}

Let 
\[
t \triangleq \max_{i\in[m],j\in[n]}A_{i,j}.
\]
Recall that $A$ is at least $(\gamma,\gamma',\rho)$-dense with respect to  $(\boldsymbol{u},\boldsymbol{v})$ for some $\gamma'>1-\gamma$.
By \cref{def-gamma-rho-dense}, we have for any $i\in [m]$,
\begin{align}\label{eq-kinn-vkovernormv1}
\sum_{k\in [n]}\frac{v_k}{\norm{\boldsymbol{v}}_1}\cdot\id{A_{i,k}\geq \rho t} \geq \gamma.
\end{align}
In addition, by \cref{def-integer-vector}, we also have for each $k\in [n]$,
\begin{align*}
\frac{v'_k}{\norm{\boldsymbol{v}'}_1} &\geq \frac{\lfloor L v_k \rfloor}{\sum_{j\in [n]}v'_j } \geq \frac{\lfloor L v_k \rfloor}{\sum_{j\in [n]} \left( \lfloor L v_j \rfloor + 1\right)} \geq  \frac{Lv_k\cdot \lfloor L v_k \rfloor}{(\lfloor L v_k \rfloor+1)(n+L\sum_{j\in [n]}  v_j) } \\
&\geq \frac{R}{R+1}\cdot \frac{L v_k}{n+L\norm{\boldsymbol{v}}_1} = \frac{R}{R+1}\cdot \frac{L\norm{\boldsymbol{v}}_1}{n+L\norm{\boldsymbol{v}}_1} \cdot \frac{v_k}{\norm{\boldsymbol{v}}_1}.
\end{align*}
Combined with \eqref{eq-three-functions-conditions}, we have
\begin{align}
\frac{v'_k}{\norm{\boldsymbol{v}'}_1} \geq \alpha\cdot \frac{v_k}{\norm{\boldsymbol{v}}_1}.
\end{align}
Combined with \eqref{eq-kinn-vkovernormv1}, we have
\begin{align*}
\sum_{k\in [n]}\frac{v'_k}{\norm{\boldsymbol{v}'}_1}\cdot\id{A_{i,k}\geq \rho t }&\geq \sum_{k\in [n]}\alpha\cdot \frac{v_k}{\norm{\boldsymbol{v}}_1}\cdot\id{ A_{i,k} \geq \rho t} \geq  \alpha\gamma. 
\end{align*}
Similarly, one can also prove that for any $j\in [n]$,
\begin{align*}
\sum_{k\in [m]}\frac{u'_k}{\norm{\boldsymbol{u}'}_1}\cdot\id{A_{k,j}\geq \rho t } \geq \alpha\gamma'. 
\end{align*}
Combining the above two inequalities with \cref{def-gamma-rho-dense}, 
we have $A$ is at least $( \alpha\gamma, \alpha\gamma',\rho)$-dense with respect to  $(\boldsymbol{u}', \boldsymbol{v}')$.
The lemma is proved.
\end{proof}

The following theorem characterizes the discrepancy between $G(A,\={u'} ,\={v'})$ and the nicely structured matrix $\+D(\={u'})\cdot G(A,\={u'} ,\={v'}) \cdot \+D(\={v'})$ where $\=u' = f_1(\={u},\={v},L) , \=v' = f_2(\={u},\={v},L)$.

\begin{theorem}\label{thm-dynamic-dense}
Let $\gamma \in (0,1]$, $\gamma'\in(1-\gamma,1]$, $\rho\in (0,1]$, $m,n \in \mathbb{Z}_{>0}$. 
Let $\boldsymbol{u} \in \mathbb{R}_{>0}^m$ and $\boldsymbol{v} \in \mathbb{R}_{>0}^n$ be vectors satisfying $\norm{\boldsymbol{u}}_1 = \norm{\boldsymbol{v}}_1 = 1$
and $A \in \mathbb{R}_{\ge 0}^{m \times n}$ be a nonzero matrix at least $(\gamma,\gamma',\rho)$-dense with respect to  $(\boldsymbol{u},\boldsymbol{v})$.
Given any positive integer $L$ with $R(\boldsymbol{u},\boldsymbol{v},L)\geq 2$, 
let 
\begin{equation}
\begin{aligned}\label{eq-thm-dynamic-dense-condition}
\=u' &\triangleq f_1(\boldsymbol{u},\boldsymbol{v},L) ,\quad\quad \quad \ \=v' \triangleq f_2(\boldsymbol{u},\boldsymbol{v},L), \quad\quad \quad \quad \quad \quad N \triangleq \norm{\=v'}_1, \quad\\
B &\triangleq G(A,\boldsymbol{u}' ,\boldsymbol{v}'),\quad \quad \quad C \triangleq\+D(\boldsymbol{u}')\cdot B \cdot \+D(\boldsymbol{v}'),\quad  \quad\quad\alpha \triangleq \frac{\gamma+\gamma'+1}{2(\gamma+\gamma')}. 
\end{aligned}
\end{equation}
Then there exists a sufficiently large $\ell>0$, such that for each integer $L\geq \ell$,
\begin{align}\label{eq-thm-dynamic-dense}
 \min_{i,j\in [N]:C_{i,j}>0} \left\{ \frac{ r_i(C)\cdot B_{i,j}}{r_i(B) \cdot C_{i,j}} \right\}\geq \frac{\alpha \rho \gamma}{n}.
\end{align}
\end{theorem}
\begin{proof}
Assume
\begin{align}\label{eq-def-u1touelll-v1tovell}
\boldsymbol{v}' = (v'_1,\dots,v'_{n}), \quad \quad
\+D(\boldsymbol{u}') = \diag\left(U_1,\dots,U_{N}\right),\quad \quad 
\+D(\boldsymbol{v}') = \diag\left(V_1,\dots,V_{N}\right).
\end{align}
By \cref{def-splitted-matrix}, we have $B$ and $C$ are of size $N\times N$.
Define
\begin{align}\label{eq-tau-def-reduction}
\tau = \max_{i\leq N,j\in N}C_{i,j}.
\end{align}
By \eqref{eq-thm-dynamic-dense-condition} and \eqref{eq-def-u1touelll-v1tovell}, 
we obtain
\begin{align}\label{eq-azeroij-relation-b}
\forall i,j\leq N, \quad  \quad \frac{B_{i,j}}{r_{i}\left(B\right)} = \frac{U^{-1}_iC_{i,j}V^{-1}_j}{\sum_{t\in [N]} U^{-1}_iC_{i,t}V^{-1}_t} = \frac{C_{i,j}V^{-1}_j}{\sum_{t\in [N]} C_{i,t}V^{-1}_t}.
\end{align}
Furthermore, by \eqref{eq-def-u1touelll-v1tovell} and the definition of $\+D(\cdot)$ in \eqref{eq-def-diag}, it follows that 
\begin{align*}
\max_{j\in [N]} V_j = \max_{j\in [n]} v'_j.
\end{align*}
Combined with $\=v' = f_2(\boldsymbol{u},\boldsymbol{v},L)$ and \cref{def-integer-vector}, we have
\begin{align}\label{eq-reduction-upperbound-vj}
\max_{j\in [N]} V_j = \max_{j\in [n]} v'_j \leq 
\norm{\=v'}_1 = N.
\end{align}
Moreover, 
by \cref{lemma-gamma-rho-dense-error} and that
$A$ is at least $(\gamma,\gamma',\rho)$-dense with respect to  $(\boldsymbol{u},\boldsymbol{v})$, 
we have there exists a sufficiently large $\ell>0$, such that for each integer $L\geq \ell$,
$C$ is at least $( \alpha\gamma,  \alpha\gamma',\rho)$-dense with respect to  $(\boldsymbol{1},\boldsymbol{1})$.
Fix any $L \geq \ell$.
We have each row of $C$ contains at least 
$\alpha \gamma \cdot N$ entries no less than $\rho \tau$. 
Thus,
\begin{align}\label{eq-reduction-ri-upperbound}
\forall i\in [N], \quad r_i(C) \geq \alpha \gamma \rho \tau N.
\end{align}
Moreover, let $T_j = \left[\sum_{k\leq j} v'_k\right]\setminus \left[\sum_{k< j} v'_k\right]$ for each $j\in [n]$.
By \cref{def-splitted-matrix} and \eqref{eq-def-diag}, we have
\begin{align*}
\forall i\leq N,\; \sum_{t\leq N} C_{i,t}V^{-1}_t = \sum_{j \in [n]}\sum_{t \in T_j} C_{i,t}V^{-1}_t = \sum_{j \in [n]}\sum_{t \in T_j} C_{i,t}(v'_j)^{-1} \leq \max_{k \in [N]} C_{i,k}\sum_{j \in [n]}\sum_{t \in T_j} (v'_j)^{-1} = n \max_{k \in [N]} C_{i,k}.
\end{align*}
Combined with \eqref{eq-tau-def-reduction}, we have
\begin{align}\label{eq-reduction-sumbv-upperbound}
\forall i\leq N,\quad \sum_{t\leq N} C_{i,t}V^{-1}_t \leq n \max_{k \in [N]} C_{i,k} \leq n\tau.
\end{align}
By \eqref{eq-azeroij-relation-b}, \eqref{eq-reduction-upperbound-vj}, \eqref{eq-reduction-ri-upperbound} and \eqref{eq-reduction-sumbv-upperbound}, we have for any $i\leq N,j\leq N$ where $C_{i,j}>0$,
\begin{align}\label{eq-thm-dynamic-dense-1}
\frac{ r_i(C)\cdot B_{i,j}}{r_i(B) \cdot C_{i,j}}= \frac{r_i(C)\cdot V^{-1}_j }{\sum_{t\in [N]} C_{i,t}V^{-1}_t} \geq \frac{\alpha \gamma \rho \tau N}{ n\tau N} \geq \frac{\alpha \rho \gamma}{n}.
\end{align}
The theorem is proved.
\end{proof}

Analogous to \cref{lemma-gamma-rho-dense-error}, which characterizes the density of matrices when the splitting operation (\cref{def-splitted-matrix}) precedes scaling, the following lemma addresses the reverse order: scaling followed by splitting.

\begin{lemma}\label{lemma-gamma-rho-dense-error-prenormalization}
Suppose the notations and conditions of \cref{lemma-gamma-rho-dense-error} hold. If $A$ is at least $(\gamma,\gamma',\rho)$-dense with respect to $(\boldsymbol{u},\boldsymbol{v})$, then there exists a sufficiently large $\ell > 0$ such that for any integer $L \geq \ell$, the matrix $G(\diag(\boldsymbol{u}) \cdot A \cdot \diag(\boldsymbol{v}), \boldsymbol{u}', \boldsymbol{v}')$ is at least $(\alpha\gamma, \alpha\gamma', \rho/2)$-dense with respect to $(\boldsymbol{1}, \boldsymbol{1})$.
\end{lemma}

\begin{proof}
Let
$C \triangleq D(\boldsymbol u')\cdot G(A,\boldsymbol u',\boldsymbol v')\cdot D(\boldsymbol v')$
and
$B \triangleq G(\diag(\boldsymbol u)\cdot A\cdot \diag(\boldsymbol v),
              \boldsymbol u',\boldsymbol v')$.
By \cref{lemma-gamma-rho-dense-error}, for all sufficiently large \(L\),
\(C\) is at least \((\alpha\gamma,\alpha\gamma',\rho)\)-dense with respect to
\((\boldsymbol 1,\boldsymbol 1)\).
It remains to compare \(B\) with \(C\). 
For each $i\in [m], j\in [n]$, 
let $S_i = \left[\sum_{k\leq i} u'_k\right]\setminus \left[\sum_{k< i} u'_k\right]$, $T_j = \left[\sum_{k\leq j} v'_k\right]\setminus \left[\sum_{k< j} v'_k\right]$.
For \(a\in S_i\) and \(b\in T_j\), we have
\[
    C_{ab}=A_{ij},\quad \quad  B_{ab}
    =
    \frac{u_i v_j}{u_i'v_j'}A_{ij}.
\]
Define diagonal matrices \(P,Q\) by
\[
P_{aa}\triangleq\frac{u_i}{u_i'} \quad \forall \ a\in S_i,
\qquad
Q_{bb}\triangleq\frac{v_j}{v_j'} \quad \forall \ b\in T_j.
\]
Then \(P\) and \(Q\) are positive diagonal matrices and $B=PCQ$.
Moreover, by \cref{def-integer-vector}, for every fixed
\(\delta>0\) and all sufficiently large \(L\),
\[
    \frac{1-\delta}{L}
    \le P_{aa}\le
    \frac{1+\delta}{L},
    \qquad
    \frac{1-\delta}{L}
    \le Q_{bb}\le
    \frac{1+\delta}{L}.
\]
Let \(M_C=\max_{a,b} C_{ab}\) and \(M_B=\max_{a,b} B_{ab}\).
We have
\[
    M_B= \max_{a,b} B_{ab} = \max_{a,b} P_{aa}C_{ab}Q_{bb} \leq \left(\max_{a}P_{aa}\right)\cdot \left(\max_{a,b} C_{ab}\right)\cdot \left(\max_{b}Q_{bb}\right)
    \le \frac{(1+\delta)^2}{L^2}M_C.
\]
Hence, whenever \(C_{ab}\ge \rho M_C\), we have
\[
    B_{ab} = P_{aa}C_{ab}Q_{bb}
    \ge
    \frac{(1-\delta)^2}{L^2}\rho M_C
    \ge
    \rho\left(\frac{1-\delta}{1+\delta}\right)^2 M_B.
\]
Choosing \(\delta>0\) small enough so that
\[
    \left(\frac{1-\delta}{1+\delta}\right)^2\ge \frac12,
\]
we obtain $B_{ab}\ge \rho M_B/2$
for every \((a,b)\) such that \(C_{ab}\ge \rho M_C\)
and all sufficiently large \(L\).
Since \(C\) is at least \((\alpha\gamma,\alpha\gamma',\rho)\)-dense with respect to
\((\boldsymbol 1,\boldsymbol 1)\),  
it follows that
\(B\) is at least $(\alpha\gamma,\alpha\gamma',\rho/2)$-dense with respect to \((\boldsymbol 1,\boldsymbol 1)\).
The lemma is proved.
\end{proof}

\subsection{Dynamics of the Block Structure under Reduction.}
The main result of this subsection is the following lemma, whose proof is omitted as it follows from a straightforward computation of the SK scaling updates. 
Note that even if the original matrix admits a block structure, the reduction introduced in \cref{def-splitted-matrix} generally destroys it, as the entries of the expanded matrix are divided by non-uniform scaling factors. 
However, this structural loss is only temporary. Specifically, letting $A$ denote the expanded matrix obtained from the reduction, the lemma demonstrates that the $2 \times 2$ block structure of the original matrix is completely recovered in $A^{(2)}$.

\begin{lemma}\label{lem-lowerbound-reduction-rctooneone}
Let $n,t,s\in \mathbb{Z}_{>0}$ with $t<s< n$, and $d \in \mathbb{R}_{>0}$.
Let $A$ be a nonnegative matrix of size $n\times n$ and $\=u = (u_1,\dots,u_n)$, $\=v = (v_1,\dots,v_n)$ be positive vectors where 
\begin{align*}
\forall i\leq t, j \leq s, \quad A_{i,j} &= \frac{1}{u_i\cdot v_j},\\
\forall i\leq t, j > s, \quad A_{i,j} &=  \frac{1}{u_i\cdot v_j},\\
\forall i> t, j \leq s, \quad A_{i,j} &= \frac{d}{u_i\cdot v_j},\\
\forall i> t, j > s, \quad A_{i,j} &= \frac{1}{u_i\cdot v_j}.
\end{align*}
Define 
\begin{align}\label{eq-def-lambda-s1-s2}
S_1 \triangleq \sum_{j\leq s}\frac{1}{v_j},\quad S_2 \triangleq \sum_{s< j\leq n}\frac{1}{v_j},\quad \lambda \triangleq \frac{S_1 + S_2}{d S_1 + S_2}.
\end{align}
Let $A^{(0)}, A^{(1)},\dots$ denote the sequence of matrices generated by SK with $A$ and $(\=1,\=1)$ as input.
Define
\[x \triangleq A^{(2)}_{1,1},\quad \quad  y \triangleq A^{(2)}_{1,s+1}, \quad  \quad z \triangleq A^{(2)}_{t+1,1},\quad \quad  q \triangleq A^{(2)}_{t+1,s+1}.\] 
Then we have 
\begin{align}\label{eq-reduction-block-partitioned-matrix}
\setlength{\jot}{12pt} 
A^{(2)}_{i,j}=
\begin{cases}
x = \dfrac{t+\lambda(n-t)}
{\,nt+\lambda(n-t)\bigl(s+d(n-s)\bigr)\,}, & i\le t,\ j\le s,\\[12pt]
y = \dfrac{t+d\lambda(n-t)}
{\,nt+\lambda(n-t)\bigl(s+d(n-s)\bigr)\,}, & i\le t,\ j> s,\\[12pt]
z = \dfrac{d\bigl(t+\lambda(n-t)\bigr)}
{\,t\bigl(n-s+ds\bigr)+d\lambda n(n-t)\,}, & i> t,\ j\le s,\\[12pt]
q = \dfrac{t+d\lambda(n-t)}
{\,t\bigl(n-s+ds\bigr)+d\lambda n(n-t)\,}, & i> t,\ j> s.
\end{cases}
\end{align}
\end{lemma}

%% file: Permanent.tex
\section{Upper Bound for \texorpdfstring{$(\mathbf{1},\mathbf{1})$}{(1,1)}-scaling.}
\label{sec-ub-oneonescaling}

In this section, we focus on the $(\boldsymbol{1}, \boldsymbol{1})$-scaling and prove \cref{thm-two-phases}.

\begin{theorem}\label{thm-two-phases}
Let $\gamma, \rho,\varepsilon \in (0,1]$, $\gamma' \in (1-\gamma,1]$, and $n\in \mathbb{Z}_{>0}$. 
Suppose $B\in \mathbb{R}_{\ge 0}^{n \times n}$ is at least $(\gamma,\gamma',\rho)$-dense with respect to $(\mathbf{1},\mathbf{1})$. 
Let $A$ be a matrix satisfying $UAV = B$ for some strictly positive diagonal matrices $U$ and $V$, and define
\begin{align}\label{eq-def-h-maximumratio}
h \triangleq \max_{i,j\in [n]:A_{i,j}>0} \left\{ \frac{r_i(A)\, B_{i,j}}{r_i(B) \, A_{i,j} }\right\}.
\end{align}
Let $A^{(0)}, A^{(1)}, \dots$ denote the sequence of matrices generated by the SK algorithm with input $(A, (\mathbf{1}, \mathbf{1}))$. 
Then there exists an iteration index
\begin{align}\label{eq-k-twophase}
k = O\left( \frac{\log h - \log \varepsilon - \log \rho - \log (\gamma+\gamma'-1)}{\rho^{14} (\gamma+\gamma'-1)^6} \right)
\end{align}
such that for all $\ell \geq k$, we have
\begin{align*}
\norm{\=r\left(A^{(\ell)}\right)- \boldsymbol{1}}_{1} + \norm{\=c\left(A^{(\ell)}\right)- \boldsymbol{1}}_{1} \leq n\varepsilon.
\end{align*}
\end{theorem}

Compared to the convergence results established in \cite{he2025phase}, our theorem significantly strengthens the prior bound and operates under a more generalized setting.
First, our theorem allows the SK algorithm to operate on an arbitrarily stretched matrix $A = U^{-1}BV^{-1}$, relaxing the requirement in \cite{he2025phase} that directly utilizes the dense matrix $B$. 
Consequently, the initial input matrix $A$ is not required to be dense.
To quantify the degree of stretching from $B$ to $A$, we introduce a parameter $h$, defined as the maximum ratio between the corresponding entries of $A$ and $B$ after row normalization. 
We normalize each entry by its row sum because the SK algorithm begins with row normalization and is therefore invariant to the absolute scale of individual rows. 
Due to this initial stretch, our final time complexity incorporates an additional $O(\log h)$ term.

Second, we establish a strictly stronger, dimension-independent convergence rate. We prove that for an input matrix $A$ with a stretch factor $h$, 
achieving an error of at most $n\varepsilon$ requires only $O(\log h - \log \varepsilon)$ iterations. 
If both $h$ and $\varepsilon$ are constants, the required number of iterations is $O(1)$. We specifically target an error bound of $n\varepsilon$ because translating the error from a $(1,1)$-scaling to a $(\={u},\={v})$-scaling (where $|\={u}|_1 = |\={v}|_1 = 1$) inherently incurs the loss of a factor of $n$. By comparison, Theorem 3.2 in \cite{he2025phase} demonstrates that even when the SK algorithm is fed the unscaled dense matrix $B$, achieving an $O(n)$ error still requires $O(\log n)$ iterations. This scenario corresponds to $h$ and $\varepsilon$ being constants in our setting. Thus, our $O(1)$ time complexity bound is strictly stronger than the $O(\log n)$ bound in prior work.

In summary, we must overcome two main difficulties: 
\begin{itemize}
\item The input matrix $A$ fed into the SK algorithm lacks the guaranteed density properties of $B$. 
\item The targeted time complexity must be proven to be independent of the matrix dimension $n$.
\end{itemize}

For simplicity, define 
\[
\forall k\geq 0,\quad \Delta^{(k)} \triangleq\norm{\=r\left(A^{(k)}\right)- \boldsymbol{1}}_1 + \norm{\=c\left(A^{(k)}\right)- \boldsymbol{1}}_1.
\]
Also define 
\begin{align}\label{eq-def-K-set}
K &\triangleq \{0,1\}\cup \left\{k \geq 2 \mid \max\left\{\Delta^{(k-2)},\Delta^{(k-1)} ,\Delta^{(k)} \right\} > \frac{4n}{5}\left(1 - \frac{1}{\gamma+\gamma'}\right)\right\}.
\end{align}
Intuitively, 
\(K\) collects all indices \(k\) for which the error in at least one of the previous three steps exceeds the threshold. For convenience, we also include \(0,1\) in \(K\).

\subsection{Structural Stability.}
In this subsection, we establish the structural stability of the scaled matrices. 
Following the notation of \cref{thm-two-phases}, let $C$ denote the row-normalized counterpart of $B$. We define an entry of $C$ to be \emph{considerable} if it exceeds $\rho/n$. 
Our structural stability results differ from those established in \cite{he2025phase} in several key aspects.

First, we prove a strictly stronger form of structural stability. While \cite{he2025phase} demonstrated that if SK takes $B$ as input, the considerable entries in $C$ stay $\Theta(1/n)$ in the scaled matrix $A^{(k)}$ (provided $A^{(k)}$ is sufficiently close to being doubly stochastic), 
we additionally prove that every entry of $A^{(k)}$ is bounded above by a constant multiple of its corresponding entry in $C$ (Item \ref{lemma-structure-stablility-two} of \cref{lem-structural-stability}). Second, our SK algorithm takes an arbitrarily scaled matrix $A$ (obtained from $B$) as input, rather than the unscaled matrix $B$ itself.

This arbitrary initialization removes the reliance on the restrictive conditions used in prior work, but it introduces new analytical difficulties. In \cite{he2025phase}, feeding $B$ directly into the algorithm ensures that the initial matrix is exactly $C$ (i.e., $A^{(0)} = C$). Because $B$ is at least $(\gamma,\gamma',\rho)$-dense, it is straightforward to obtain explicit upper and lower bounds on the row and column sums of $A^{(k)}$.
Furthermore, by expressing $A^{(k)} = X A^{(0)} Y = X C Y$ (where the diagonal matrices $X$ and $Y$ accumulate the respective normalizations), 
the author can exploit the property that the permanents of $X$ and $Y$ are strictly greater than $1$ to yield the desired result. By contrast, in our setting, the initial matrix $A^{(0)}$ is a stretched version of $C$, meaning $A^{(0)} = X' C Y'$. Consequently, we cannot directly bound the row and column sums of $A^{(k)}$. Moreover, the cumulative scaling becomes $A^{(k)} = X (X' C Y') Y = (X X') C (Y' Y)$, where the permanents of the composite matrices $(X X')$ and $(Y' Y)$ are no longer guaranteed to be greater than $1$.

To overcome these hurdles, we develop the following techniques:
\begin{itemize}
\item \emph{Bounding Row/Column Sums via Relative Mass:} 
For any $k \notin K$ (defined in \eqref{eq-def-K-set}), assume without loss of generality that $A^{(k-2)}$ is column-normalized. We first prove that each entry accounts for an $O(1/n)$ fraction of its respective row mass (\cref{lem-upper-bound-max-element-initial}).
Symmetrically, we show that each entry of $A^{(k-1)}$ constitutes an $O(1/n)$ fraction of its column mass. By leveraging these two relative properties, we successfully derive the necessary bounds for the row sums of $A^{(k)}$.

\item \emph{Symmetric Shrinking Construction:} To overcome the limitations of the composite scaling matrices, we utilize the degrees of freedom inherent in matrix scaling. We construct alternative diagonal matrices $X''$ and $Y''$ satisfying $A^{(k)} = X'' C Y''$. In particular, we carefully choose $X''$ and $Y''$ such that their shrinking effects are perfectly balanced for the entry in $A^{(k)}$ that experiences the maximum relative shrinking compared to $C$. By leveraging the favorable properties of these newly constructed matrices $X''$ and $Y''$, we establish our final bounds (Lemmas \ref{lem-lower-bound-elements-new} and \ref{lem-lower-bound-elements}).
\end{itemize}

The main result of this subsection is the following lemma.

\begin{lemma}\label{lem-structural-stability}
Assume the conditions in \cref{thm-two-phases}. Fix any $k\not\in K$.
Let $C$ denote the matrix 
\begin{align}\label{eq-def-c-normlized-b}
\forall i,j\in [n],\quad C_{i,j} = \frac{B_{i,j}}{r_i(B)}.
\end{align}
\begin{enumerate}[(a)]
\item \label{lemma-structure-stablility-regime} We have 
\begin{align}
&\forall i\in [n],\quad \frac{\rho^2(\gamma + \gamma' - 1) }{10}<  r_i\left(A^{(k)}\right)< \frac{10}{\rho^2(\gamma + \gamma' - 1) },\label{eq-structure-stablility-regime-row}\\
&\forall j\in [n],\quad \frac{\rho^2(\gamma + \gamma' - 1) }{10}< c_j\left(A^{(k)}\right)< \frac{10}{\rho^2(\gamma + \gamma' - 1)},\label{eq-structure-stablility-regime-column}\\
&\forall i,j\in [n],\quad \quad \quad\quad \ \ A^{(k)}_{i,j}  < \frac{10}{\rho^2(\gamma + \gamma' - 1)n }.\label{eq-structure-stablility-regime-item}
\end{align}

\item \label{lemma-structure-stablility-one}
For any $i,j\in [n]$ with $C_{i,j}\geq \rho/n$, we have $A^{(k)}_{i,j} > \theta/n$ where 
\begin{align}\label{eq-structure-stability-1}
\theta \triangleq 10^{-5}\cdot \rho^{14} \cdot \gamma^3 \cdot (\gamma + \gamma' - 1) ^{5}.
\end{align}
Thus, 
\begin{align}\label{eq-structure-stability-2}
\forall i\in [n], \quad
\abs{\left\{t\mid A^{(k)}_{i,t}> \frac{\theta}{n}\right\}}\geq \lceil\gamma n \rceil.
\end{align}

\item \label{lemma-structure-stablility-two}
For any $i,j\in [n]$, we have $A^{(k)}_{i,j} \leq \delta C_{i,j}$ where 
\begin{align}\label{eq-structure-stability-3}
\delta\triangleq 10^{3}\cdot \rho^{-7} (\gamma + \gamma' - 1) ^{-4}.
\end{align}
Thus,
\begin{align}\label{eq-structure-stability-4}
\frac{\perman\left(A^{\left(k\right)}\right)}{\perman\left(A^{\left(0\right)}\right)} \leq (h\delta)^n.
\end{align}
\end{enumerate} 
\end{lemma}

The following lemma, adapted from Lemma 3.5 in \cite{he2025phase}, is used in our proof of \cref{lem-structural-stability}. In \cite{he2025phase}, the author establishes that all entries of $A^{(k)}$ are $O(1/n)$ by utilizing explicit upper and lower bounds on its row and column sums, combined with the fact that $\alpha\left(A^{(k)}\right)$ is small. 
In our setting, however, such bounds on the row and column sums are not directly available. 
Instead, we leverage the smallness of $\alpha\left(A^{(k)}\right)$ to demonstrate a relative property: when $A^{(k)}$ is row-normalized, each entry accounts for an $O(1/n)$ fraction of its respective column mass. Symmetrically, when $A^{(k)}$ is column-normalized, each entry constitutes an $O(1/n)$ fraction of its respective row mass. The detailed proof is deferred to the appendix.

\begin{lemma}\label{lem-upper-bound-max-element-initial}
Let \( \gamma,\rho \in (0,1] \), \(\gamma' \in (1-\gamma,1]\), and $n\in \mathbb{Z}_{>0}$. Let \( A \in \mathbb{R}_{\ge 0}^{n \times n}\) be at least $(\gamma,\gamma',\rho)$-dense with respect to  $(\=1,\=1)$, and let \( X \) and \( Y \) be diagonal matrices with positive diagonal entries. Suppose that 
\( B = XAY \) satisfies the following conditions:
\begin{itemize}
\item $B$ is standardized and has entries in \([0,1]\),
\item $\gamma + \gamma' - 1 - \alpha(B) > 0$.
\end{itemize}
If $r_t(B) = 1$ for each $t\in [n]$, we have
\begin{align}\label{eq-upper-bound-max-element-older}
\forall i,j\in [n],\quad B_{i,j} \leq \frac{c_j(B)}{\rho^2(\gamma + \gamma' - 1 - \alpha(B))n }.
\end{align}
Otherwise, $c_t(B) = 1$ for each $t\in [n]$. We have
\begin{align}\label{eq-upper-bound-max-element-older-column}
\forall i,j\in [n],\quad B_{i,j} \leq \frac{r_i(B)}{\rho^2(\gamma + \gamma' - 1 - \alpha(B))n }.
\end{align} 
\end{lemma}

The next lemma, a key ingredient in the proof of \cref{lem-structural-stability}, shows that each entry of the scaled matrix $A^{(k)}$ is bounded above by a constant multiple of the corresponding entry of $C$.

\begin{lemma}\label{lem-lower-bound-elements-new}
Suppose \(a>0 \), \(b>0 \), \(\gamma,\rho
\in (0,1]\), \(\gamma' \in (1-\gamma,1]\), and $n\in \mathbb{Z}_{>0}$.
Let \( A \) be an \( n \times n \) matrix with entries in \([0,1]\).
Let $B = XAY$ where $X,Y$ are positive diagonal matrices. 
Assume $A$ and $B$ 
satisfy the following conditions:
\begin{itemize}
\item $\forall i,j\in [n]$, $A_{i,j} \leq a/n$ and $B_{i,j} \leq b/n$;
\item \(B\) is standardized and $\alpha(B)< 1- 1/(\gamma+\gamma')$;
\item each row of $A$ contains at least \(\gamma n\) entries with values at least \(\rho/n\);
\item each column of $A$ contains at least \(\gamma' n\) entries with values at least \(\rho/n\).
\end{itemize}
For any $i,j\in [n]$, we have $B_{i,j} \leq \delta A_{i,j}$ where 
\begin{align}\label{eq-def-theta-new-constant}
\delta\triangleq \frac{ab^2}{\rho^2((\gamma+\gamma')(1 -\alpha(B)) - 1)}.
\end{align}
\end{lemma}

\begin{proof}
Assume for contradiction that there exists some $k,\ell\in [n]$ where $B_{k,\ell} > \delta A_{k,\ell}$.
Combined with $B_{k,\ell} = x_kA_{k,\ell}y_{\ell}$, 
we have $x_{k}y_{\ell} > \delta$.
Thus, we have either $x_{k} \geq \sqrt{\delta}$ or $y_{\ell} \geq \sqrt{\delta}$.
Since the factorization 
$B=XAY$ is invariant under the rescaling 
$(X,Y)\rightarrow (cX,Y/c)$ for any 
$c>0$, there is one degree of freedom in the choice of the diagonal scalings. Hence, without loss of generality, we may rescale so that 
$x_{k} = y_{\ell}$, which together with 
$x_{k}y_{\ell} > \delta$ implies
$x_{k} = y_{\ell} > \sqrt{\delta}$.

Let $C = \{j\in [n]| A_{k,j}\geq \rho/n\}$.
For each $j\in C$, we have
\[
B_{k,j} = x_kA_{k,j}y_{j}\leq \frac{b}{n}.\]
Combined with $x_{k}> \sqrt{\delta}$ and $A_{k,j}\geq \rho/n$,
we have 
\begin{equation*}
\begin{aligned}
y_j\leq \frac{b}{n x_kA_{k,j}} < \frac{b}{\rho\sqrt{\delta}}.
\end{aligned}
\end{equation*}
Combined with \eqref{eq-def-theta-new-constant}, we have 
\begin{equation}
\begin{aligned}\label{eq-yj-small-new}
y_j< \frac{b}{\rho\sqrt{\delta}} = \sqrt{\frac{((\gamma+\gamma')(1 -\alpha(B)) - 1)}{a}}.
\end{aligned}
\end{equation}
Similarly, 
let $R = \{i\in [n]| A_{i,\ell}\geq \rho/n\}$.
For each $i\in R$, we have
\[
B_{i,\ell} = x_iA_{i,\ell}y_{\ell}\leq \frac{b}{n}.\]
Combined with $y_{\ell}> \sqrt{\delta}$ and $A_{i,\ell}\geq \rho/n$,
we have 
\begin{equation*}
\begin{aligned}
x_i\leq \frac{b}{n y_\ell A_{i,\ell}} < \frac{b}{\rho\sqrt{\delta}}.
\end{aligned}
\end{equation*}
Combined with \eqref{eq-def-theta-new-constant}, we have 
\begin{equation}
\begin{aligned}\label{eq-yj-small-new-xi}
x_i< \frac{b}{\rho\sqrt{\delta}} = \sqrt{\frac{((\gamma+\gamma')(1 -\alpha(B)) - 1)}{a}}.
\end{aligned}
\end{equation}
Combining \eqref{eq-yj-small-new}, \eqref{eq-yj-small-new-xi} with $A_{i,j}\leq a/n$ for each $i,j\in [n]$,
we have 
\begin{equation}\label{eq-app-upper-bound-rc-new}
\begin{aligned}
&\quad \sum_{i\in R}\sum_{j\in C}x_{i}A_{i,j}y_{j}\\
&< n^2\cdot \frac{a}{n}\cdot \frac{((\gamma+\gamma')(1 -\alpha(B)) - 1)}{a}.\\
&= ((\gamma+\gamma')(1 -\alpha(B)) - 1)n.
\end{aligned}
\end{equation}

Assume without loss of generality that $c_t(B) = 1$ for each $t\in [n]$.
The case where $r_t(B) = 1$ for each $t\in [n]$ can be proved analogously.
By \eqref{eq-def-alpha},
we have 
\[(\abs{C}+\abs{R})\alpha(B)\geq (\gamma + \gamma')n\alpha(B) \geq n\alpha(B) \geq \sum_{r\in [n]}\abs{\left(\sum_{t\in [n]}B_{r,t}\right) -1} \geq \sum_{r\in R}\abs{\left(\sum_{t\in [n]}B_{r,t}\right) -1} \geq \abs{R} - \sum_{r\in R}\sum_{t\in [n]}B_{r,t}.\]
Thus, 
we have
\[\sum_{r\in R}\sum_{t\in [n]}B_{r,t} \geq \abs{R} -\alpha(B)\cdot(\abs{C}+\abs{R}).\]
In addition, 
\[\sum_{r\in [n]}\sum_{t\in [n]}B_{r,t} = \sum_{t\in [n]}\sum_{r\in [n]}B_{r,t}  = \sum_{t\in [n]}c_t(B) = n .\]
Hence, 
\[\sum_{r\not\in R}\sum_{t\in [n]}B_{r,t} \leq n - \abs{R} + \alpha(B)\cdot(\abs{C}+\abs{R}).\]
Moreover, we also have
\[\sum_{r\in [n]}\sum_{t\not\in C}B_{r,t} =\sum_{t\not\in C}c_t(B) = n - \abs{C}.\]
Thus, we have 
\begin{align*}
\sum_{r\in R}\sum_{t\in C} B_{r,t} \geq n - \sum_{r\not\in R}\sum_{t\in [n]}B_{r,t} - \sum_{r\in [n]}\sum_{t\not\in C}B_{r,t} \geq n - (2n - (1 -\alpha(B))(\abs{R} + \abs{C})).
\end{align*}
Combined with $\abs{R} + \abs{C} \geq (\gamma'+\gamma) n$, we have
\begin{align*}
\sum_{r\in R}\sum_{t\in C} B_{r,t}\geq n - (2n - (1 -\alpha(B))\cdot (\gamma + \gamma')n) \geq ((\gamma+\gamma')(1 -\alpha(B)) - 1)n.
\end{align*}
This is contradictory with \eqref{eq-app-upper-bound-rc-new}.
The lemma is proved.
\end{proof}

The next lemma, adapted from Lemma 3.6 in \cite{he2025phase}, shows that any considerable entry of $C$ stays $\Omega(1/n)$ in the scaled matrix $A^{(k)}$.
The proof of the lemma is provided in the appendix.

% Let $A^{(k)} = X A^{(0)} Y$, where the diagonal matrices $X$ and $Y$ represent the cumulative products of all row and column normalizations, respectively, applied during the execution of SK from $A^{(0)}$ to $A^{(k)}$. 
% In \cite{he2025phase}, the initial matrix is exactly $C$ (i.e., $A^{(0)} = C$). 
% Thus, the desired result follows directly from the favorable properties of $X$ and $Y$, specifically that their permanents are strictly greater than $1$.
% In our setting, however, the SK algorithm takes a stretched matrix of $B$ as input, meaning $A^{(0)}$ is also a stretched version of $C$, expressed as $A^{(0)} = X' C Y'$. 
% This yields $A^{(k)} = (X X') C (Y' Y)$, which prevents us from directly invoking the properties of $X$ and $Y$. 
% To circumvent this issue, we exploit the degrees of freedom inherent in matrix scaling to define specific diagonal matrices $X''$ and $Y''$ such that $A^{(k)} = X'' C Y''$. 
% Crucially, we choose $X''$ and $Y''$ such that their shrinking effects are perfectly balanced for the entry in $A^{(k)}$ that experiences the maximum relative shrinking compared to $C$. 
% By leveraging the favorable properties of these newly constructed $X''$ and $Y''$, we establish the final result.

\begin{lemma}\label{lem-lower-bound-elements}
Assume the conditions in \cref{lem-lower-bound-elements-new}.
Suppose further that there exists $d>0$
such that $r_i(B)\geq d$ and 
$c_i(B)\geq d$ for each $i\in [n]$.
Then for any $i,j\in [n]$ with $A_{i,j}\geq \rho/n$, we have $B_{i,j} > \theta/n$ where \(\theta\) is any positive constant satisfying
\begin{align}\label{eq-def-theta-constant}
\theta < \frac{\rho^3 d^2((\gamma+\gamma')(1 -\alpha(B)) - 1)}{a^3b^2}.
\end{align}
\end{lemma}

Now we can prove \cref{lem-structural-stability}.

\begin{proof}[Proof of \cref{lem-structural-stability}]
Assume without loss of generality that $k$ is odd.
We first prove \ref{lemma-structure-stablility-regime}.  
By $k\not \in K$ and \eqref{eq-def-K-set} we have
\[
\Delta^{(k-2)} < \frac{9n}{10}\left(1 - \frac{1}{\gamma+\gamma'}\right).
\]
Combined with \eqref{eq-def-alpha} and $\gamma+\gamma'>1$,
we have
\begin{align}\label{eq-upperbound-alphakminustwo}
\alpha\left(A^{(k-2)}\right) = \frac{1}{n}\cdot\sum_{i\in [n] }\alpha_{i} = \frac{1}{n}\cdot \Delta^{(k-2)} < \frac{9}{10}\left(1 - \frac{1}{\gamma+\gamma'}\right) < \frac{9}{10}\left(\gamma+\gamma' - 1\right).
\end{align}
In addition, 
by \cref{fact-c-a-a0-a1},
we have $A^{(k-2)}_{i,j}\in [0,1]$ for each $i,j\in [n]$.
Since $k-2$ is odd, we have $c_j\left(A^{(k-2)}\right) = 1$ for each $j\in [n]$.
Moreover, by \eqref{eq-ak-aprime-relation}, we have 
$A^{(k-2)} = XAY$ for some diagonal $X$ and $Y$. 
Combined with $UAV = B$,
we have $A^{(k-2)} = XU^{-1}BV^{-1}Y$. 
Thus, one can apply \cref{lem-upper-bound-max-element-initial} to $A^{(k-2)}$ and obtain
\begin{align}\label{eq-upper-bound-max-element-older-kminusone}
\forall i,j\in [n],\quad A^{(k-2)}_{i,j} \leq \frac{r_i\left(A^{(k-2)}\right)}{\rho^2(\gamma + \gamma' - 1 - \alpha(A^{(k-2)}))n } < \frac{10 \cdot r_i\left(A^{(k-2)}\right)}{\rho^2(\gamma + \gamma' - 1)n } .
\end{align}
Since $k-1$ is even, we have
\begin{align*}
\forall i,j\in [n],\quad A^{(k-1)}_{i,j}  = \frac{A^{(k-2)}_{i,j}}{r_i\left(A^{(k-2)}\right)} < \frac{10}{\rho^2(\gamma + \gamma' - 1)n }.
\end{align*}
Therefore, we have 
\begin{align*}
\forall j\in [n],\quad c_j\left(A^{(k-1)}\right) < \frac{10}{\rho^2(\gamma + \gamma' - 1) }.
\end{align*}
Combined with \cref{lem-monotone-rsum-csum}, we have
\begin{align*}
\forall i\in [n],\quad  r_i\left(A^{(k)}\right) \geq \left(\max_{j\in [n]}c_j\left(A^{(k-1)}\right)\right)^{-1} >\frac{\rho^2(\gamma + \gamma' - 1) }{10}.
\end{align*}
Similar to \eqref{eq-upper-bound-max-element-older-kminusone}, by applying \cref{lem-upper-bound-max-element-initial} to $A^{(k-1)}$, we have 
\begin{align*}
\forall i,j\in [n],\quad A^{(k-1)}_{i,j} \leq \frac{c_j\left(A^{(k-1)}\right)}{\rho^2(\gamma + \gamma' - 1 - \alpha(A^{(k-1)}))n } < \frac{10 \cdot c_j\left(A^{(k-1)}\right)}{\rho^2(\gamma + \gamma' - 1)n }.
\end{align*}
Thus,
\begin{align}\label{eq-structure-stablility-regime-item-new}
\forall i,j\in [n],\quad A^{(k)}_{i,j}  = \frac{A^{(k-1)}_{i,j}}{c_j\left(A^{(k-1)}\right)} < \frac{10}{\rho^2(\gamma + \gamma' - 1)n }.
\end{align}
Therefore, we have 
\begin{align*}
\forall i\in [n],\quad r_i\left(A^{(k)}\right) < \frac{10}{\rho^2(\gamma + \gamma' - 1) }.
\end{align*}
In summary, \eqref{eq-structure-stablility-regime-row} is proved.
Furthermore, by $k$ is odd, we also have 
\begin{align*}
\forall j\in [n],\quad \frac{\rho^2(\gamma + \gamma' - 1) }{10} <  c_j\left(A^{(k)}\right) = 1 < \frac{10}{\rho^2(\gamma + \gamma' - 1) }.
\end{align*}
% In summary, we have 
% $r_t\left(A^{(k)}\right)\geq \rho^2(\gamma + \gamma' - 1)/10$ and 
% $c_t\left(A^{(k)}\right)\geq \rho^2(\gamma + \gamma' - 1) /10$ for each $t\in [n]$.
Thus, \eqref{eq-structure-stablility-regime-column} is proved.
Furthermore, \eqref{eq-structure-stablility-regime-item} is immediate by \eqref{eq-structure-stablility-regime-item-new}.
Hence, \ref{lemma-structure-stablility-regime} is proved.

In the next, we prove \ref{lemma-structure-stablility-one}.  
We claim that all the conditions in \cref{lem-lower-bound-elements} are satisfied by $C$ and $A^{(k)}$ with $a = 1/(\rho\gamma)$, $b =10/(\rho^2(\gamma + \gamma' - 1)) $, $d = \rho^2(\gamma + \gamma' - 1)/10$.
\begin{itemize}
\item Let $\ell$ denote $\max_{i,j} B_{i,j}$.
Since $B$ is at least $(\gamma,\gamma',\rho)$-dense with respect to $(\=1,\=1)$,
each row of $B$ contains at least \(\gamma n\) entries with values at least \(\rho \ell\) and each column contains at least \(\gamma' n\) entries with values at least \(\rho \ell\).
Moreover, we have $r_i(B) \leq \ell n$ for each $i\in [n]$.
Combined with \eqref{eq-def-c-normlized-b}, 
we have 
each row of $C$ contains at least \(\gamma n\) entries with values at least \(\rho/n\) and
each column contains at least \(\gamma' n\) entries with values at least \(\rho/n\).

\item By each row of $B$ contains at least \(\gamma n\) entries with values at least \(\rho \ell\), we have $r_i(B) \geq \rho \ell \cdot \gamma n$ for each $i\in [n]$.
Combined with \eqref{eq-def-c-normlized-b}, 
we have $C_{i,j}\leq \ell/(\rho \ell\gamma n) = 1/(\rho \gamma n)$ for each $i,j\in [n]$.
\item By \eqref{eq-def-c-normlized-b}, we have 
\[C =  \left(\diag(r_1(B),\dots,r_n(B))\right)^{-1}B.\]
Thus, 
$B = \diag(r_1(B),\dots,r_n(B)) \cdot C$.
In addition, by \eqref{eq-ak-aprime-relation} we have 
$A^{(k)} = XAY$ for some diagonal $X$ and $Y$. 
Combined with $UAV = B$,
we have $A^{(k)} = XU^{-1}BV^{-1}Y$.
Combined with $B = \diag(r_1(B),\dots,r_n(B))\cdot  C$,
we have $A^{(k)} = XU^{-1}\cdot \diag(r_1(B),\dots,r_n(B))\cdot CV^{-1}Y$. 

\item By \ref{lemma-structure-stablility-regime}, we have 
$r_t\left(A^{(k)}\right)\geq \rho^2(\gamma + \gamma' - 1)/10$ and 
$c_t\left(A^{(k)}\right)\geq \rho^2(\gamma + \gamma' - 1) /10$ for each $t\in [n]$.
\item By \ref{lemma-structure-stablility-regime}, we also have $A^{(k)}_{i,j}  \leq 10/\left(\rho^2(\gamma + \gamma' - 1)n\right)$ for each $i,j\in [n]$.
\item Similar to \eqref{eq-upperbound-alphakminustwo}, by $k\not\in K$ we have 
\begin{align}\label{eq-upperbound-alohaak-structure-stable}
\alpha\left(A^{(k)}\right)< \frac{9}{10}\left(1 - \frac{1}{\gamma+\gamma'}\right) < 1 - \frac{1}{\gamma+\gamma'}.
\end{align}
\end{itemize}
Thus, 
by \cref{lem-lower-bound-elements} we have $A^{(k)}_{i,j} > \theta/n$ for any $i,j\in [n]$ with $C_{i,j}\geq \rho/n$,
because  
\begin{align*}
&\quad \frac{\rho^3 d^2\left((\gamma+\gamma')\left(1 -\alpha\left(A^{(k)}\right)\right) - 1\right)}{a^3b^2} \\
&= \rho^3 \cdot 10^{-4}\cdot  \rho^{8}(\gamma + \gamma' - 1) ^{4}\left((\gamma+\gamma')\left(1 -\alpha\left(A^{(k)}\right)\right) - 1\right)\cdot (\rho \gamma)^3 \\
&= 10^{-4}\cdot \rho^{14} \cdot \gamma^3 \cdot (\gamma + \gamma' - 1) ^{4}\left((\gamma+\gamma')\left(1 -\alpha\left(A^{(k)}\right)\right) - 1\right)\\
(\text{by \eqref{eq-upperbound-alohaak-structure-stable}})\quad &> 10^{-5}\cdot \rho^{14} \cdot \gamma^3 \cdot (\gamma + \gamma' - 1) ^{5}\\
& = \theta.
\end{align*}
Combined with that each row of $C$ contains at least \(\lceil\gamma n \rceil\) entries with values at least \(\rho/n\), \eqref{eq-structure-stability-2} is immediate. 

In the following, we prove \ref{lemma-structure-stablility-two}.
Similar to the proof of \ref{lemma-structure-stablility-one}, we also have that all the conditions in \cref{lem-lower-bound-elements-new} are satisfied by $C$ and $A^{(k)}$ with $a = 1/(\rho\gamma)$, $b =10/(\rho^2(\gamma + \gamma' - 1)) $.
Thus, by \cref{lem-lower-bound-elements-new} we have $A^{(k)}_{i,j} \leq \delta C_{i,j}$ for any $i,j\in [n]$, because 
\begin{align*}
&\quad \frac{ab^2}{\rho^2\left((\gamma+\gamma')\left(1 -\alpha\left(A^{(k)}\right)\right) - 1\right)}\\
&=100\cdot \rho^{-2}\left((\gamma+\gamma')\left(1 -\alpha\left(A^{(k)}\right)\right) - 1\right)^{-1}(\rho\gamma)^{-1}(\rho^2(\gamma + \gamma' - 1))^{-2}\\
&=100\cdot \rho^{-7}\gamma^{-1}(\gamma + \gamma' - 1)^{-2}\left((\gamma+\gamma')\left(1 -\alpha\left(A^{(k)}\right)\right) - 1\right)^{-1}\\
(\text{by \eqref{eq-upperbound-alohaak-structure-stable}})\quad &< 10^{3}\cdot \rho^{-7}\gamma^{-1} (\gamma + \gamma' - 1) ^{-3}\\
(\text{by $\gamma'\leq 1$})\quad &\leq 10^{3}\cdot \rho^{-7} (\gamma + \gamma' - 1) ^{-4}\\
& = \delta.
\end{align*}
Moreover, by \eqref{eq-def-h-maximumratio} and \eqref{eq-def-c-normlized-b},
we have 
\begin{align*}
\forall i,j\in[n],\quad A^{(0)}_{i,j} = \frac{A_{i,j}}{r_i(A)} \geq \frac{C_{i,j}}{h}.
\end{align*}
Combined with $A^{(k)}_{i,j} \leq \delta C_{i,j}$ for any $i,j\in [n]$, we have 
\begin{align*}
\forall i,j\in[n],\quad A^{(0)}_{i,j} \geq \frac{A^{(k)}_{i,j}}{h\delta}.
\end{align*}
Therefore,
\begin{align*}
\perman\left(A^{(k)}\right) &= \sum_{\sigma }\prod_{i\in [n]}A^{(k)}_{i,\sigma(i)} \leq \sum_{\sigma }\prod_{i\in [n]}\left(A^{(0)}_{i,\sigma(i)}\cdot h\delta\right) =  \perman\left(A^{(0)}\right)\cdot \left(h
\delta\right)^n.
\end{align*}
Hence, \eqref{eq-structure-stability-4} is immediate.
The lemma is proved.
\end{proof}

\subsection{Rapid Decay of Error.} In this subsection, we prove \cref{thm-two-phases}.

The proof proceeds in two main phases. The first phase establishes an upper bound on the size of the set $K$ defined in \eqref{eq-def-K-set}, demonstrating that the error rapidly decreases to below $O(n)$. During this stage, \cref{lemma-condition-upperbound} guarantees that the permanent of the matrix $A^{(k)}$ generated in each iteration of the SK algorithm increases by a specific factor. Concurrently, Item \ref{lemma-structure-stablility-two} of \cref{lem-structural-stability} imposes a strict upper bound on the permanent of $A^{(k)}$ relative to the initial matrix $A^{(0)}$. By combining this guaranteed per-round growth with the global upper bound, we can deduce a theoretical maximum for the size of $K$.

The second phase proves that once the error falls below the $O(n)$ threshold, its subsequent decay is exponential. For iterations beyond $K$, Item \ref{lemma-structure-stablility-one} of \cref{lem-structural-stability} ensures that each row of the matrix generated by the SK algorithm contains $\lceil\gamma n \rceil$ elements of magnitude $\Omega(1/n)$, and similarly, each column contains $\lceil\gamma' n \rceil$ such elements. 
According to \cref{lem-max-accuracy-decrease}, if $\gamma > 1/2$, the deviation from $1$ of at least one of the maximum column sum and the reciprocal of the minimum column sum decays by a constant factor in each iteration.
By symmetry, if  $ \gamma' > 1/2$, the corresponding deviation for the row sums decays at the same rate.
Given the condition $\gamma+\gamma'>1$, at least one of $\gamma$ or $\gamma'$ must be strictly greater than $1/2$. Consequently, the deviation in either the row sums or the column sums must experience rapid decay in every subsequent step, which ultimately drives the exponential convergence of the total error. 

Combining the analyses of these two phases yields the proof of \cref{thm-two-phases}.

The following two lemmas, adapted from \cite{he2025phase}, are used in our proof.

\begin{lemma}\label{lemma-condition-upperbound}
Assume the conditions in \cref{thm-two-phases}. 
Then the set $K$ defined in \eqref{eq-def-K-set} is finite.
Furthermore, let $L$ be the maximum number in $K$.
We have 
\begin{align*}
\perman\left(A^{(L+1)}\right) \geq \perman\left(A^{(0)}\right)\exp\left( \frac{n(\abs{K}-2)}{24}\cdot \left(\frac{9}{10}\left(1 - \frac{1}{\gamma+\gamma'}\right) \right)^2\right).
\end{align*}
\end{lemma}

\begin{proof}
Let $t\triangleq 1-1/(\gamma+\gamma')$ and $J=\{j\geq 0:\Delta^{(j)}>4nt/5\}$. By the definition of $K$ in \eqref{eq-def-K-set}, we have $J\subseteq K$.
The support graph of $B$ satisfies Hall's condition: a nonempty set of at most $\gamma n$ rows has at least $\gamma n$ neighbors, while a set of more than $\gamma n>(1-\gamma')n$ rows meets every column. Since $B=UAV$ with strictly positive diagonal $U,V$, $A$ and $B$ have identical supports. The matrix $A^{(0)}$ is obtained from $A$ by positive row scaling, so it has the same support; hence Hall's condition holds for $A^{(0)}$, implying that $\perman(A^{(0)})>0$. Each subsequent SK update also rescales rows or columns by positive factors, so every iterate retains this support and has positive permanent. By \cref{lem-permanent-standardized-upper-bound}, $\perman(A^{(j)})\leq 1$ since $A^{(j)}$ is standardized.
For each iteration index $j$, let $s_i=c_i(A^{(j)})$ if $j$ is even and $s_i=r_i(A^{(j)})$ if $j$ is odd, for every $i\in[n]$. Then $\sum_i s_i=n$ and $\sum_i|s_i-1|=\Delta^{(j)}$. Let $p=(1/n)_{i=1}^n$ and $q=(s_i/n)_{i=1}^n$. Since $A^{(j)}$ has no zero rows or columns, $s_i>0$; together with $\sum_i s_i=n$, this shows that $p$ and $q$ are strictly positive probability vectors. By \cref{lem-facts-ak} and \eqref{eq-pinsker},
\begin{align}\label{eq-pinsker-permanent-increment-bound}
\log\frac{\perman(A^{(j+1)})}{\perman(A^{(j)})}=-\sum_{i=1}^n\log s_i=n\sum_{i=1}^n\frac{1}{n}\log\frac{1/n}{s_i/n}=nD_{\mathrm{KL}}(p\|q)\geq \frac{n}{2}\norm{p-q}_1^2=\frac{1}{2n}\left(\sum_i|s_i-1|\right)^2=\frac{(\Delta^{(j)})^2}{2n}.
\end{align}
Furthermore, the definition of $J$ gives
\begin{equation}\label{eq-J-permanent-increment-bound}
\begin{aligned}
\forall j\in J,\qquad \log\frac{\perman(A^{(j+1)})}{\perman(A^{(j)})}\geq \frac{(\Delta^{(j)})^2}{2n}> \frac{(4nt/5)^2}{2n}= \frac{8nt^2}{25}.
\end{aligned}
\end{equation}
By \eqref{eq-pinsker-permanent-increment-bound}, all log-permanent increments are nonnegative, while \eqref{eq-J-permanent-increment-bound} gives a lower bound for every increment indexed by $J$. Thus, for each $m\geq 0$,
\[
\begin{aligned}
\frac{8nt^2}{25}\abs{J\cap\{0,\ldots,m\}}\leq \sum_{j\in J\cap\{0,\ldots,m\}}\log\frac{\perman(A^{(j+1)})}{\perman(A^{(j)})}\leq \sum_{j=0}^{m}\log\frac{\perman(A^{(j+1)})}{\perman(A^{(j)})}=\log\frac{\perman(A^{(m+1)})}{\perman(A^{(0)})}\leq -\log\perman(A^{(0)}),
\end{aligned}
\]
where the last inequality follows from \cref{lem-permanent-standardized-upper-bound}, since $A^{(m+1)}$ is standardized. This uniform bound on $\abs{J\cap\{0,\ldots,m\}}$ implies that $J$ is finite and
\[
\abs{J}\leq \frac{25}{8nt^2}\left(-\log\perman(A^{(0)})\right).
\]
By the definition of $K$, for every $k\in K\setminus\{0,1\}$, at least one of $\Delta^{(k-2)},\Delta^{(k-1)},\Delta^{(k)}$ exceeds $4nt/5$. Hence there is some $j\in J\cap\{k-2,k-1,k\}$, and therefore $k\in\{j,j+1,j+2\}$. It follows that
\[
K\setminus\{0,1\}\subseteq \bigcup_{j\in J}\{j,j+1,j+2\}.
\]
Since each set in this union has at most three elements, $|K|-2\leq 3|J|$; in particular, $K$ is finite. Since $J\subseteq K$ and $L=\max K$, every index in $J$ is at most $L$. In particular, the preceding Pinsker bound gives, for every $j\in\{0,\ldots,L\}$,
\[
\log\frac{\perman(A^{(j+1)})}{\perman(A^{(j)})}
\geq \frac{(\Delta^{(j)})^2}{2n}\geq 0.
\]
Thus the increments indexed outside $J$ are nonnegative. Splitting the telescoping sum into indices in $J$ and those outside $J$, we obtain
\[
\begin{aligned}
\log\frac{\perman(A^{(L+1)})}{\perman(A^{(0)})}&=\sum_{j=0}^{L}\log\frac{\perman(A^{(j+1)})}{\perman(A^{(j)})}=\sum_{j\in J}\log\frac{\perman(A^{(j+1)})}{\perman(A^{(j)})}
+\sum_{j\in\{0,\ldots,L\}\setminus J}\log\frac{\perman(A^{(j+1)})}{\perman(A^{(j)})}\\
&\geq \sum_{j\in J}\log\frac{\perman(A^{(j+1)})}{\perman(A^{(j)})}\geq \frac{8n}{25}|J|t^2\geq \frac{8n}{75}(|K|-2)t^2\geq \frac{n(|K|-2)}{24}\left(\frac{9t}{10}\right)^2.
\end{aligned}
\]
\end{proof}

\begin{lemma}\label{lem-max-accuracy-decrease}
Given any $n\in \mathbb{Z}_{> 0}$,
let $A\in \mathbb{R}_{\geq 0}^{n \times n}$ be a matrix with no zero rows or columns,
and $A^{(0)},A^{(1)},\dots$ be the sequence of matrices generated by SK with input $(A,(\=1,\=1))$.
Let $L > n/2$ be an integer and $\theta>0$. 
Given any even $k\geq 0$ where 
\begin{align*}
\forall i\in [n], \quad
\abs{\left\{t\mid A^{(k)}_{i,t}> \frac{\theta}{n}\right\}}\geq L,
\end{align*} 
we have at least one of the following inequalities holds:
\begin{align*}
\max_{j\in [n]} c_j\left(A^{(k+2)}\right) - 1 &\leq
 \left(1 - \theta\left(\frac{L}{n} - \frac{1}{2} \right)\right)\cdot\left(\max_{j\in [n]} c_j\left(A^{(k)}\right) - 1\right),\\
\left(\min_{j\in [n]} c_j\left(A^{(k+2)}\right)\right)^{-1} - 1 &\leq \left(1 - \theta\left(\frac{L}{n} - \frac{1}{2} \right)\right)\cdot\left(\left(\min_{j\in [n]} c_j\left(A^{(k)}\right)\right)^{-1} - 1\right).
\end{align*}
% Similarly, given any odd $k\geq 0$ where 
% \begin{align*}
% \forall i\in [n],
% \abs{\left\{t\mid A^{(k)}_{t,i}>\frac{\theta}{n}\right\}}\geq L,
% \end{align*} 
% we have at least one of the following inequalities holds:
% \begin{align*}
% \max_{j\in [n]} r_j\left(A^{(k+2)}\right) - 1 &\leq
% \tau\cdot\left(\max_{j\in [n]} r_j\left(A^{(k)}\right) - 1\right)\\
% \left(\min_{j\in [n]} r_j\left(A^{(k+2)}\right)\right)^{-1} - 1 &\leq \tau\cdot\left(\left(\min_{j\in [n]} r_j\left(A^{(k)}\right)\right)^{-1} - 1\right).
% \end{align*}
\end{lemma}

\begin{proof}
The proof of Lemma~3.7 in \cite{he2025phase} uses only the stated row condition and applies unchanged.
\end{proof}

Now we can prove \cref{thm-two-phases}
\begin{proof}[Proof of \cref{thm-two-phases}]
It suffices to prove the case \(\gamma\ge \gamma'\), since the case \(\gamma'>\gamma\) follows by an entirely analogous argument. Thus, in the remainder of the proof, we assume \(\gamma\ge \gamma'\).
Recall the set $K$ in \eqref{eq-def-K-set}.
We claim that there exists some even $k$ where 
\begin{align}\label{eq-def-k-ub-two-phase}
k - 2\abs{K} = O\left(1+\rho^{-14}\cdot\left(\gamma+\gamma'-1\right)^{-6}\cdot\left(- \log \varepsilon - \log \rho - \log (\gamma+\gamma'-1)\right)\right)
\end{align}
such that one of the following holds:
\begin{align}
\max_{i\in [n]} c_i\left(A^{(k)}\right) \leq 1 +  \frac{\varepsilon}{2},\label{eq-conditon-maxcak}\\
\left(\min_{i\in [n]} c_i\left(A^{(k)}\right)\right)^{-1} \leq 1+  \frac{\varepsilon}{2}.\label{eq-conditon-mincak-inverse}
\end{align}
Assume without loss of generality that \eqref{eq-conditon-maxcak} holds.
The case \eqref{eq-conditon-mincak-inverse} holds can be proved analogously.
For each even $\ell \geq k$, by \eqref{eq-conditon-maxcak} and \cref{lem-monotone-rsum-csum}, we also have 
\begin{align*}
\max_{i\in [n]} c_i\left(A^{(\ell)}\right) \leq 1 +  \frac{\varepsilon}{2}.
\end{align*}
Thus, we have
\begin{align*}
\norm{\=c\left(A^{(\ell)}\right)- \boldsymbol{1}}_1 &= \sum_{i\in [n]} \left(c_i\left(A^{(\ell)}\right) - 1\right)\cdot \id{c_i\left(A^{(\ell)}\right)>1} + \sum_{i\in [n]} \left(1 - c_i\left(A^{(\ell)}\right)\right)\cdot \id{c_i\left(A^{(\ell)}\right)<1}\\
&=2\sum_{i\in [n]} \left(c_i\left(A^{(\ell)}\right) - 1\right)\cdot \id{c_i\left(A^{(\ell)}\right)>1}\\
&\leq 2n\cdot\left(\max_{i\in [n]} c_i\left(A^{(\ell)}\right) - 1 \right)\\
&\leq n\varepsilon.
\end{align*}
Moreover, since $\ell$ is even, we have $\norm{\=r\left(A^{(\ell)}\right)- \boldsymbol{1}}_1 = 0$.
Thus, 
\begin{align*}
\forall \text{ even } \ell \geq k, \quad \norm{\=r\left(A^{(\ell)}\right)- \boldsymbol{1}}_{1} + \norm{\=c\left(A^{(\ell)}\right)- \boldsymbol{1}}_{1} \leq n\varepsilon.
\end{align*}
In addition, for each odd $\ell >k$, by \eqref{eq-conditon-maxcak} and \cref{lem-monotone-rsum-csum}, we have 
\begin{align}
\left(\min_{i\in [n]}r_i\left(A^{(\ell)}\right)\right)^{-1} \leq \max_{i\in [n]}c_i\left(A^{(k)}\right) \leq 1 + \frac{\varepsilon}{2}.
\end{align}
Thus, by $\min_{i\in [n]} r_i\left(A^{(\ell)}\right)\leq 1$ we have
\begin{align*}
1 - \min_{i\in [n]} r_i\left(A^{(\ell)}\right) = \min_{i\in [n]} r_i\left(A^{(\ell)}\right)\cdot\left(\left(\min_{i\in [n]} r_i\left(A^{(\ell)}\right)\right)^{-1} - 1\right)  \leq \left(\min_{i\in [n]} r_i\left(A^{(\ell)}\right)\right)^{-1} - 1 \leq \frac{\varepsilon}{2}.
\end{align*}
Therefore, 
\begin{align*}
\norm{\=r\left(A^{(\ell)}\right)- \boldsymbol{1}}_1 &= \sum_{i\in [n]} \left(r_i\left(A^{(\ell)}\right) - 1\right)\cdot \id{r_i\left(A^{(\ell)}\right)>1} + \sum_{i\in [n]} \left(1 - r_i\left(A^{(\ell)}\right)\right)\cdot \id{r_i\left(A^{(\ell)}\right)<1}\\
&=2\sum_{i\in [n]} \left(1 - r_i\left(A^{(\ell)}\right)\right)\cdot \id{r_i\left(A^{(\ell)}\right)<1}\\
&\leq 2n\cdot\left(1 - \min_{i\in [n]} r_i\left(A^{(\ell)}\right)\right)\\
&\leq n\varepsilon.
\end{align*}
Moreover, since $\ell$ is odd, we have $\norm{\=c\left(A^{(\ell)}\right)- \boldsymbol{1}}_1 = 0$.
Therefore, we have
\begin{align*}
\forall \text{ odd } \ell \geq k, \quad \norm{\=r\left(A^{(\ell)}\right)- \boldsymbol{1}}_{1} + \norm{\=c\left(A^{(\ell)}\right)- \boldsymbol{1}}_{1} \leq n\varepsilon.
\end{align*}
In summary, we always have
\begin{align*}
\forall \ell \geq k, \quad \norm{\=r\left(A^{(\ell)}\right)- \boldsymbol{1}}_{1} + \norm{\=c\left(A^{(\ell)}\right)- \boldsymbol{1}}_{1} \leq n\varepsilon.
\end{align*}
Furthermore, let $L$ be the maximum number in $K$. We have $L+1\not \in K$.
By \cref{lem-structural-stability}, we have 
\begin{align*}
\frac{\perman\left(A^{\left(L+1\right)}\right)}{\perman\left(A^{\left(0\right)}\right)} \leq \left(10^{3}\cdot h\rho^{-7} (\gamma + \gamma' - 1) ^{-4}\right)^n.
\end{align*}
Combined with \cref{lemma-condition-upperbound}, we have 
\begin{align*}
\abs{K} = O\left((\gamma+\gamma' - 1)^{-2}\left(\log h -\log (\gamma + \gamma' - 1)  -  \log \rho\right)\right).
\end{align*}
Combined with \eqref{eq-def-k-ub-two-phase}, we have \eqref{eq-k-twophase} is satisfied.
In the following, we prove the claim that there exists some even $k$ satisfying \eqref{eq-def-k-ub-two-phase} such that one of \eqref{eq-conditon-maxcak} and \eqref{eq-conditon-mincak-inverse} is true. Then the theorem is immediate.

Define 
\begin{align}
q &\triangleq  1 - 10^{-5} \rho^{14} \gamma^3 (\gamma + \gamma' - 1) ^{5}\left(\gamma -\frac{1}{2}\right),\label{eq-def-g1-gamma}\\
t & \triangleq (\log \varepsilon + 2\log \rho + \log (\gamma+\gamma'-1)- 5)/\log q, \label{eq-def-t}\\
k & \triangleq \min \left\{i \geq 0\mid i \text{ is even and } i\geq 2\abs{K} + 4t + 2\right\}.\label{eq-def-k-upper-bound}
\end{align}
One can verify that 
\begin{align}\label{eq-log-ggamma}
-\log q \geq 10^{-5} \rho^{14} \gamma^3 (\gamma + \gamma' - 1) ^{5}\left(\gamma -\frac{1}{2}\right).
\end{align}
Recall that $\gamma\geq \gamma'$ and $\gamma' \in (1-\gamma,1]$.
We have 
\[2\left(\gamma - \frac{1}{2}\right) = 2\gamma - 1 \geq   \gamma+ \gamma' -  1 >0.\]
Combined with \eqref{eq-log-ggamma},
we have 
\begin{align}\label{eq-log-ggamma-new}
-\log q \geq 10^{-6} \rho^{14} \gamma^3 (\gamma + \gamma' - 1) ^{6} \geq 10^{-7} \rho^{14} (\gamma + \gamma' - 1) ^{6} .
\end{align}
By \eqref{eq-def-t}, \eqref{eq-def-k-upper-bound}, and \eqref{eq-log-ggamma-new},
we have
\eqref{eq-def-k-ub-two-phase} is satisfied.

Define 
$
T \triangleq \left\{0\leq j< k\ | j \text{ is even and } j\not\in K   \right\}.$
We have 
\begin{align}\label{eq-size-T}
\abs{T}\geq (k - 2)/2 - \abs{K} \geq 2t.
\end{align}
For each $j\in T$, we have $j$ is even and $j\not\in K$.
By \cref{lem-structural-stability}, we have
\begin{align*}
\forall j\in T, i\in [n], \quad
\abs{\left\{\ell\mid A^{(j)}_{i,\ell}> \frac{ \rho^{14} \gamma^3 (\gamma + \gamma' - 1) ^{5}}{10^{5} n}\right\}}\geq \lceil\gamma n \rceil.
\end{align*}
Combined with \cref{lem-max-accuracy-decrease}, we have
one of the following two inequalities is true:
\begin{align}
\max_{i\in [n]} c_i\left(A^{(j+2)}\right) - 1 &\leq q\left(\max_{i\in [n]} c_i\left(A^{(j)}\right) - 1\right),\label{eq-maxcj-decrease}\\
\left(\min_{i\in [n]} c_i\left(A^{(j+2)}\right)\right)^{-1} - 1 &\leq q\left(\left(\min_{i\in [n]} c_i\left(A^{(j)}\right)\right)^{-1} - 1\right).\label{eq-maxcj-inverse-decrease}
\end{align}

Let $S \triangleq \left\{j\in T\ | \ \eqref{eq-maxcj-decrease} \text{ holds for $j$}\right\}.$ 
At first, consider the case $\abs{S} \geq \abs{T}/2$. 
By \eqref{eq-size-T}, we have $\abs{S} \geq t$.
Let $\ell_0$ be the minimum element in $T$. We have
\begin{align*}
\max_{i\in [n]} c_i\left(A^{(k)}\right) - 1 = 
\left(\max_{i\in [n]} c_i\left(A^{(\ell_0)}\right) - 1\right)\cdot \prod_{j= \ell_0/2}^{k/2-1}\frac{\left(\max_{i\in [n]} c_i\left(A^{(2j+2)}\right) - 1\right)}{\left(\max_{i\in [n]} c_i\left(A^{(2j)}\right) - 1\right)}.
\end{align*}
By \cref{lem-monotone-rsum-csum},
we have 
\[
\forall j\geq 0,\quad \frac{\left(\max_{i\in [n]} c_i\left(A^{(2j+2)}\right) - 1\right)}{\left(\max_{i\in [n]} c_i\left(A^{(2j)}\right) - 1\right)} \leq 1.
\]
Recall that $k$ is even.
By the definitions of $T$ and $S$, one can verify that $j$ is even and $\ell_0 \leq j \leq k-2$ for each $j\in S$.
Combined with the above two inequalities, we have 
\begin{align*}
\max_{i\in [n]} c_i\left(A^{(k)}\right) - 1 \leq 
\left(\max_{i\in [n]} c_i\left(A^{(\ell_0)}\right) - 1\right)\cdot\prod_{j\in S}\frac{\left(\max_{i\in [n]} c_i\left(A^{(j+2)}\right) - 1\right)}{\left(\max_{i\in [n]} c_i\left(A^{(j)}\right) - 1\right)}.
\end{align*}
Combined with \eqref{eq-maxcj-decrease} and $\abs{S}\geq t$, we have
\begin{align*}
\max_{i\in [n]} c_i\left(A^{(k)}\right) - 1 \leq 
q^{\abs{S}}\cdot\left(\max_{i\in [n]} c_i\left(A^{(\ell_0)}\right) - 1\right) \leq q^{t}\cdot\left(\max_{i\in [n]} c_i\left(A^{(\ell_0)}\right) - 1\right).
\end{align*}
Meanwhile, by $\ell_0\in T$, we have $\ell_0 \not \in K$.
Combined with \cref{lem-structural-stability}, we have
$c_i\left(A^{(\ell_0)}\right) < 10/\left(\rho^2(\gamma + \gamma' - 1)\right)$ for each $i\in [n]$.
Hence, 
\begin{align}\label{eq-maxci-upperbound-ggamma}
\max_{i\in [n]} c_i\left(A^{(k)}\right) - 1 < q^{t}\left(10\rho^{-2}(\gamma + \gamma' - 1)^{-1} - 1\right)< 10q^{t}\rho^{-2}(\gamma + \gamma' - 1)^{-1} < \frac{\varepsilon}{2},
\end{align}
where the last inequality is by \eqref{eq-def-t}.
Thus, \eqref{eq-conditon-maxcak} is proved.
At last, consider the other case $\abs{S} < \abs{T}/2$.
By \eqref{eq-size-T}, we have $\abs{T\setminus S}\geq t$.
Similar to \eqref{eq-maxci-upperbound-ggamma}, we can also prove \eqref{eq-conditon-mincak-inverse}.
In summary, there exists some even $k$ satisfying \eqref{eq-def-k-ub-two-phase} such that one of \eqref{eq-conditon-maxcak} and \eqref{eq-conditon-mincak-inverse} is true. 
The theorem is proved.
\end{proof}

%% file: target.tex
\section{Upper bound for \texorpdfstring{$(\={u},\={v})$}{(u,v)}-scaling.}
\label{sec-ub-uv-scaling}

\subsection{Proof of Theorems \ref{thm-rc-scaling-log} and \ref{thm-rc-scaling-prenormalization}.}

Now we can prove Theorems \ref{thm-rc-scaling-log} and \ref{thm-rc-scaling-prenormalization}.

\begin{proof}[Proof of \cref{thm-rc-scaling-log}]
For any sufficiently large integer $L$, let $R = R(\boldsymbol{u},\boldsymbol{v},L)$, $\boldsymbol{u}' \triangleq (u'_1,\dots,u'_m) = f_1(\boldsymbol{u},\boldsymbol{v},L)$, $\boldsymbol{v}' \triangleq (v'_1,\dots,v'_n) = f_2(\boldsymbol{u},\boldsymbol{v},L)$,
$G = G(B,f_1(\boldsymbol{u},\boldsymbol{v},L),f_2(\boldsymbol{u},\boldsymbol{v},L))$,
$G' = \+D\left(\boldsymbol{u}'\right)\cdot G \cdot \+D\left(\boldsymbol{v}'\right)$.
In addition, let $t \triangleq \norm{f_1(\boldsymbol{u},\boldsymbol{v},L)}_1$.
By \cref{def-integer-vector}, one can verify that 
$t = \norm{f_2(\boldsymbol{u},\boldsymbol{v},L)}_1$
and $t \leq L\norm{\boldsymbol{u}}_1 = L$.
Combined with \cref{def-splitted-matrix}, we have $G$ is a matrix of size $t\times t$.

By \cref{lemma-gamma-rho-dense-error}, we have there exists a sufficiently large $\ell_1>0$, such that for each integer $L\geq \ell_1$,
$G'$ is at least $(\alpha\gamma, \alpha\gamma',\rho)$-dense with respect to $(\boldsymbol{1},\boldsymbol{1})$ where $\alpha =  (\gamma+\gamma'+1)/(2(\gamma+\gamma'))$.
Fix any $L \geq \ell_1$.
By $\gamma + \gamma'>1$, we have
$\alpha\gamma+\alpha\gamma' > 1$.
Thus, all the assumptions of \cref{thm-two-phases} are satisfied, with \(G\) and \(G'\) playing the roles of the matrices \(A\) and \(B\) in that theorem.
Let $G^{(0)},G^{(1)},\dots$ denote the sequence of matrices generated by SK
with $(G,(\boldsymbol{1},\boldsymbol{1}))$ as input.
Define
\begin{align*}
h \triangleq \max_{i,j\in [t]:G_{i,j}>0} \left\{ \frac{r_i(G)\cdot G'_{i,j}}{r_i(G') \cdot G_{i,j} }\right\}.
\end{align*}
By \cref{thm-two-phases}, we have the conclusion that there exists some 
\begin{align}\label{eq-k-thm-rc-scaling-log-upperbound}
k &= O\left(\rho^{-14} \cdot\left(\alpha\gamma+\alpha\gamma'-1\right)^{-6}\left(\log h - \log \left(\frac{\varepsilon}{2}\right)  - \log \rho-  \log (\alpha\gamma+\alpha\gamma'-1)\right)\right)
\end{align}
such that 
\begin{align}\label{eq-k-thm-rc-scaling-log-error-1}
\norm{r\left(G^{(k)}\right)- \boldsymbol{1}}_{1} + \norm{c\left(G^{(k)}\right)- \boldsymbol{1}}_{1} \leq \frac{t\varepsilon}{2} \leq  \frac{L\varepsilon}{2}.
\end{align}

% The above bound on \(k\) is uniform for all sufficiently large \(L\),
% because Theorem 3.8 gives
% \(h\le n/(\alpha\rho\gamma)\).
% Therefore, we may choose \(L\) sufficiently large so that
% the approximation guarantee in Theorem 3.3 holds for all iterations
% up to this uniform bound.

Moreover, by \cref{thm-dynamic-dense} we have there exists a sufficiently large $\ell_2>0$, such that 
$h \leq n/(\alpha \rho \gamma)$
for each integer $L\geq \ell_2$.
Combined with \eqref{eq-k-thm-rc-scaling-log-upperbound}, we have
\begin{align*}
k &=  O\left(\rho^{-14} \cdot\left(\gamma+\gamma'-1\right)^{-6}\left(\log n - \log \varepsilon - \log \gamma - \log \rho - \log (\gamma+\gamma'-1)\right)\right).
\end{align*}
By $\gamma'\leq 1$, we have $\gamma+\gamma'-1\leq \gamma$.
Thus, 
\begin{align}\label{eq-k-thm-rc-scaling-log-upperbound-new}
k &=  O\left(\rho^{-14} \cdot\left(\gamma+\gamma'-1\right)^{-6}\left(\log n - \log \varepsilon - \log \rho - \log (\gamma+\gamma'-1)\right)\right).
\end{align}

The bound in \eqref{eq-k-thm-rc-scaling-log-upperbound-new} is uniform for all $L\geq\max\{\ell_1,\ell_2\}$. 
Since \cref{thm-two-phases} guarantees the error bound at every subsequent iteration, we may choose a fixed integer $K$, independent of $L$ and satisfying the asymptotic bound in \eqref{eq-k-thm-rc-scaling-log-upperbound-new}, such that \eqref{eq-k-thm-rc-scaling-log-error-1} holds with $k=K$ for every such $L$.
Let $A$ denote the output of SK after $K$ iterations with $(B,(\boldsymbol{u}, \boldsymbol{v}))$ as input. 
Applying \cref{thm-reduction-precision} at this fixed $K$, there exists a sufficiently large $\ell_3>0$ such that for each integer $L\geq \ell_3$,
\begin{align}\label{eq-k-thm-rc-scaling-log-error-2}
\abs{\norm{r\left(A\right)- \boldsymbol{u}}_{1} + \norm{c\left(A\right)- \boldsymbol{v}}_{1} - \frac{\norm{r\left(G^{(K)}\right)- \boldsymbol{1}}_{1} + \norm{c\left(G^{(K)}\right)- \boldsymbol{1}}_{1}}{L}}\leq \frac{\varepsilon}{2}.
\end{align}

Fix any $L \geq \max\{\ell_1,\ell_2,\ell_3\}$.
By \eqref{eq-k-thm-rc-scaling-log-error-1} and \eqref{eq-k-thm-rc-scaling-log-error-2}, 
we have  
\[\norm{r\left(A\right)- \boldsymbol{u}}_{1} + \norm{c\left(A\right)- \boldsymbol{v}}_{1} \leq \varepsilon.\]
Combined with \eqref{eq-k-thm-rc-scaling-log-upperbound-new},
the theorem is proved.
\end{proof}

\begin{proof}[Proof of \cref{thm-rc-scaling-prenormalization}]
For any sufficiently large integer $L$, let $R = R(\boldsymbol{u},\boldsymbol{v},L)$, $\boldsymbol{u}' \triangleq (u'_1,\dots,u'_m) = f_1(\boldsymbol{u},\boldsymbol{v},L)$, $\boldsymbol{v}' \triangleq (v'_1,\dots,v'_n) = f_2(\boldsymbol{u},\boldsymbol{v},L)$,
$D = G(\diag(\=u)\cdot B\cdot \diag(\=v),f_1(\boldsymbol{u},\boldsymbol{v},L),f_2(\boldsymbol{u},\boldsymbol{v},L))$.
In addition, let $t \triangleq \norm{f_1(\boldsymbol{u},\boldsymbol{v},L)}_1$.
By \cref{def-integer-vector}, one can verify that 
$t = \norm{f_2(\boldsymbol{u},\boldsymbol{v},L)}_1$
and $t \leq L\norm{\boldsymbol{u}}_1 = L$.
Combined with \cref{def-splitted-matrix}, we have $D$ is a matrix of size $t\times t$.

By \cref{lemma-gamma-rho-dense-error-prenormalization}, we have there exists a sufficiently large $\ell_1>0$, such that for each integer $L\geq \ell_1$,
$D$ is at least $(\alpha\gamma, \alpha\gamma',\rho/2)$-dense with respect to  $(\boldsymbol{1},\boldsymbol{1})$ where $\alpha =  (\gamma+\gamma'+1)/(2(\gamma+\gamma'))$.
Fix any $L \geq \ell_1$.
By $\gamma + \gamma'>1$, we have
$\alpha\gamma+\alpha\gamma' > 1$.
Thus, all the assumptions of \cref{thm-two-phases} are satisfied, with \(D\) playing the roles of the matrices \(A\) and \(B\) and $h = 1$ in that theorem.
Let $D^{(0)},D^{(1)},\dots$ denote the sequence of matrices generated by SK
with $(D,(\boldsymbol{1},\boldsymbol{1}))$ as input.
By \cref{thm-two-phases}, we have the conclusion that there exists some 
\begin{align}\label{eq-k-thm-rc-scaling-log-upperbound-new-prenorm}
k &= O\left(2^{14}\rho^{-14} \cdot\left(\alpha\gamma+\alpha\gamma'-1\right)^{-6}\left(-\log \left(\frac{\varepsilon}{2}\right)  - \log \left(\frac{\rho}{2}\right)-  \log (\alpha\gamma+\alpha\gamma'-1)\right)\right)
\end{align}
such that 
\begin{align}\label{eq-k-thm-rc-scaling-log-error-1-prenorm}
\norm{r\left(D^{(k)}\right)- \boldsymbol{1}}_{1} + \norm{c\left(D^{(k)}\right)- \boldsymbol{1}}_{1} \leq \frac{t\varepsilon}{2} \leq  \frac{L\varepsilon}{2}.
\end{align}

The bound in \eqref{eq-k-thm-rc-scaling-log-upperbound-new-prenorm} is independent of $L$. Since \cref{thm-two-phases} guarantees the error bound at every subsequent iteration, choose a fixed integer $K$, independent of $L$ and satisfying the asymptotic bound in \eqref{eq-k-thm-rc-scaling-log-upperbound-new-prenorm}, such that \eqref{eq-k-thm-rc-scaling-log-error-1-prenorm} holds with $k=K$ for every $L\geq\ell_1$.
Let $A$ denote the output of SK after $K$ iterations with $(\diag(\=u)\cdot B\cdot \diag(\=v),(\boldsymbol{u}, \boldsymbol{v}))$ as input. 
Applying \cref{thm-reduction-precision} at this fixed $K$, there exists a sufficiently large $\ell_2>0$ such that for each integer $L\geq \ell_2$,
\begin{align}\label{eq-k-thm-rc-scaling-log-error-2-prenorm}
\abs{\norm{r\left(A\right)- \boldsymbol{u}}_{1} + \norm{c\left(A\right)- \boldsymbol{v}}_{1} - \frac{\norm{r\left(D^{(K)}\right)- \boldsymbol{1}}_{1} + \norm{c\left(D^{(K)}\right)- \boldsymbol{1}}_{1}}{L}}\leq \frac{\varepsilon}{2}.
\end{align}

Fix any $L \geq \max\{\ell_1,\ell_2\}$.
By \eqref{eq-k-thm-rc-scaling-log-error-1-prenorm} and \eqref{eq-k-thm-rc-scaling-log-error-2-prenorm}, 
we have  
\[\norm{r\left(A\right)- \boldsymbol{u}}_{1} + \norm{c\left(A\right)- \boldsymbol{v}}_{1} \leq \varepsilon.\]
Combined with \eqref{eq-k-thm-rc-scaling-log-upperbound-new-prenorm},
the theorem is proved.
\end{proof}

\subsection{Proof of Theorems \ref{thm-upper-bound-general-ot-uv} and \ref{thm-upper-bound-general-ot}.}
Finally, we can prove Theorems \ref{thm-upper-bound-general-ot-uv} and \ref{thm-upper-bound-general-ot}.

\begin{proof}[Proof of \cref{thm-upper-bound-general-ot-uv}]
Let $T = \eta C, K = \exp(-T), \gamma = r_{\tau}(T;\={v}), \gamma' = c_{\tau}(T;\={u})$.
Since $T$ is $(\tau,\kappa)$-well-bounded with respect to  $(\={u},\={v})$,
we have $\gamma+\gamma'\geq 1+\kappa$ where 
\[
\gamma
=\min_{i\in[m]} \sum_{j\in[n]} v_j\,\id{T_{ij}\leq\tau},
\qquad
\gamma'
=
\min_{j\in[n]} \sum_{i\in[m]} u_i\,\id{T_{ij}\leq \tau}.
\]
Hence, 
\[
\gamma
=\min_{i\in[m]} \sum_{j\in[n]} v_j\,\id{K_{ij}\geq \exp(-\tau)},
\qquad
\gamma'
=
\min_{j\in[n]} \sum_{i\in[m]} u_i\,\id{K_{ij}\geq \exp(-\tau)}.
\]
Moreover, since $C$ is positive, every entry of $K$ is at most $\exp(0) = 1$. 
Combined with $\norm{\boldsymbol{u}}_1 = \norm{\boldsymbol{v}}_1 = 1$, it follows from \cref{def-gamma-rho-dense} that $K$ is at least $(\gamma,\gamma',\exp(-\tau))$-dense with respect to  $\=u,\=v$.
By $\gamma+\gamma'\geq 1 + \kappa$ and \cref{thm-rc-scaling-log}, we have with $(\exp(-\eta C),(\=u,\=v))$ as input, SK can output a matrix $A$ satisfying 
\[\norm{\=r\left(A\right)- \boldsymbol{u}}_{1} + \norm{\=c\left(A\right)- \boldsymbol{v}}_{1} \leq \varepsilon.\]
in $O\left(\exp(14\tau) \cdot\left(\gamma+\gamma'-1\right)^{-6}\left(\log n - \log \varepsilon + \tau - \log (\gamma+\gamma'-1)\right)\right)$ iterations.
The theorem is immediate.
\end{proof}

\begin{proof}[Proof of \cref{thm-upper-bound-general-ot}]
Recall that $\exp(-\eta C)$ is $(\gamma,\gamma',\exp(-\tau))$-dense with respect to  $\=u,\=v$ for some $\gamma+\gamma'\geq 1 + \kappa$
as shown in the proof of \cref{thm-upper-bound-general-ot-uv}.
Combined with \cref{thm-rc-scaling-prenormalization}, 
we have with $(\diag(\=u) \cdot \exp(-\eta C) \cdot \diag(\=v),(\=u,\=v))$ as input, SK can output a matrix $A$ satisfying 
\[\norm{\=r\left(A\right)- \boldsymbol{u}}_{1} + \norm{\=c\left(A\right)- \boldsymbol{v}}_{1} \leq \varepsilon.\]
in $O\left(\exp(14\tau) \cdot\left(\gamma+\gamma'-1\right)^{-6}\left(\tau - \log \varepsilon  - \log (\gamma+\gamma'-1)\right)\right)$ iterations.
The theorem is immediate.
\end{proof}

%% file: Expansity.tex
\section{Lower bound.}\label{sec-lb}
In this section, we prove \cref{thm-lower-bound-main-rc}.

To construct a hard instance $A$ for $(\={u}, \={v})$-scaling that requires $\Omega(-\log \nu(A) - \log \varepsilon)$ SK iterations to converge, a convenient approach is to design $A$ as a block matrix. Let $C = G(A, f_1(\boldsymbol{u}, \boldsymbol{v}, L), f_2(\boldsymbol{u}, \boldsymbol{v}, L))$ be the corresponding $(\boldsymbol{1}, \boldsymbol{1})$-scaling instance reduced from $(A,(\=u,\=v))$, and let $C^{(0)}, C^{(1)}, C^{(2)}, \dots$ denote the matrix sequence generated by applying the SK algorithm on $(C, (\boldsymbol{1}, \boldsymbol{1}))$. 
The main idea of the proof proceeds as follows.

First, we construct a base matrix $B$ whose bottom-left block entries equal $z$ and which requires at least $c_0(-\log z-\log(L\varepsilon))-b_0$ iterations to converge under standard $(\boldsymbol{1}, \boldsymbol{1})$-scaling, for positive constants $c_0,b_0$ depending only on the block proportions. To facilitate our subsequent analysis, $B$ is specifically designed to be a block matrix.
Ideally, we would like the initial reduced matrix $C$ to exactly equal our hard instance $B$. However, as in \cref{def-splitted-matrix}, the reduction process divides the entries by non-uniform scaling factors, inherently destroying the block structure of the original matrix. Consequently, we cannot directly embed $B$ as the matrix $C$.

Fortunately, \cref{lem-lowerbound-reduction-rctooneone} establishes that if $A$ is a block matrix, this structural loss is only temporary: the block structure is completely recovered at state $C^{(2)}$, after a sequence of row, column, and row normalizations. 
Leveraging this property, we can carefully tune the initial entries of $A$ such that the intermediate state $C^{(2)}$ exactly matches our constructed hard instance $B$, with $Lz \leq 2n\nu(A)/(\gamma\gamma')$.
Ultimately, by choosing $\nu(A)$ sufficiently small to absorb the additive constants, the lower bound for the SK algorithm starting from state $C^{(2)} = B$ implies that the overall iteration complexity for the $(\boldsymbol{u}, \boldsymbol{v})$-scaling of $A$ is $\Omega(-\log \nu(A) - \log \varepsilon)$ for fixed $\gamma,\gamma'$.

\subsection{Counter-example for \texorpdfstring{$(\=1,\=1)$}{(1,1)}-scaling.}
In this subsection, we prove the following theorem, which gives a lower bound of $c_0(-\log \nu-\log\varepsilon)-b_0$ iterations for $(\={1},\={1})$-scaling. The additive constant is needed because the permitted accuracy may already be attained by the initial matrix.

\begin{theorem}\label{thm-main-lb-rc-ab}
Let $n,s,t \in \mathbb{Z}_{>0}$.
Assume that there exist constants $\alpha,\beta,\gamma \in (0,1)$ such that
\begin{align}\label{eq-relation-s-t-n}
\alpha n \le t < s \le \beta n, \quad \quad  \quad s - t \geq \gamma n.
\end{align}
Let $\varepsilon \in (0, s-t)$ and $\nu > 0$.
Suppose $\theta_{1,1}, \theta_{1,2}, \theta_{2,1}$, and $\theta_{2,2}$ are positive real numbers satisfying the following conditions:
\begin{align}
&\theta_{2,1} = \nu, \label{eq-invariant-a1thetaij-qminusaqthetak2j-2-theta21}\\
&\theta_{1,2}<  
\frac{6n}
{6n^2 + 5s(n-t)},\label{eq-invariant-a1thetaij-qminusaqthetak2j-2-theta12}\\
&\forall i\leq 2,  \quad s\theta_{i,1} + (n-s)\theta_{i,2} = 1,\label{eq-invariant-a1thetaij-qminusaqthetak2j-2-kequal0}\\
&\frac{t-n}{n}< t \theta_{1,1} + (n-t) \theta_{2,1} - 1 < 0,\label{eq-invariant-a1thetaij-qminusaqthetak2j-1-kequal0-1}\\
& \frac{st(s-t)}
{4n^3}< t\cdot \theta_{1,2} + (n - t)\cdot \theta_{2,2} - 1<\frac{s(n-t)}{n(n-s)}.
\label{eq-invariant-a1thetaij-qminusaqthetak2j-1-kequal0-2}
\end{align}
Let \(A\) be an \(n \times n\) matrix partitioned into four blocks defined as follows:
\begin{itemize}
\item The top-left block is of size \(t \times s\), where all entries are equal to \(\theta_{1,1}\).
\item The bottom-left block is of size \((n - t) \times s\), where all entries are equal to \(\theta_{2,1}\).
\item The top-right block is of size \(t \times (n - s)\), where all entries are equal to \(\theta_{1,2}\).
\item The bottom-right block is of size \((n - t) \times (n - s)\), where all entries are equal to \(\theta_{2,2}\).
\end{itemize}
Let $A^{(0)},A^{(1)},\dots$
denote the sequence of matrices generated by SK with $(A,(\mathbf{1}, \mathbf{1}))$ as input.
Let $K$ be the minimum integer $k$ such that 
\begin{align}\label{eq-condition-thm-main-lb-rc-ab}
\norm{\=r\left(A^{(k)}\right)- \=1}_{1} + \norm{\=c\left(A^{(k)}\right)- \=1}_{1} \leq \varepsilon.
\end{align}
There exist constants $c_0,b_0>0$, depending only on $\alpha,\beta,\gamma$, such that
\[
K \geq c_0\left(-\log \nu-\log\varepsilon\right)-b_0.
\]
In particular, $K=\Omega(-\log\nu-\log\varepsilon)$ whenever $-\log\nu-\log\varepsilon\geq 2b_0/c_0$.
\end{theorem}

Under the conditions of \cref{thm-main-lb-rc-ab}, one can verify that each matrix $A^{(k)}$ where $k\geq 0$ can be partitioned into four blocks, with all entries within each block being identical:
\begin{itemize}
\item The top-left block has size \(t \times s\), and its entries are denoted by \(\theta^{(k)}_{1,1}\).
\item The bottom-left block has size \((n-t) \times s\), and its entries are denoted by \(\theta^{(k)}_{2,1}\).
\item The top-right block has size \(t \times (n-s)\), and its entries are denoted by \(\theta^{(k)}_{1,2}\).
\item The bottom-right block has size \((n-t) \times (n-s)\), and its entries are denoted by  \(\theta^{(k)}_{2,2}\).
\end{itemize}

The intuition of our construction is as follows.
Initially, $\theta^{(0)}_{2,1}$ is tiny.
For each even $k\geq 0$, the algorithm normalizes every row sum to $1$, so in particular,
$s\cdot \theta^{(k)}_{2,1} + (n-s) \cdot \theta^{(k)}_{2,2} = 1$.
If $\theta^{(k)}_{2,1}$ is small, then $(n-s) \cdot \theta^{(k)}_{2,2}$ must be very close to $1$.
Using $s>t$, this implies that the column sum
$t\cdot \theta^{(k)}_{1,2} + (n-t) \cdot \theta^{(k)}_{2,2}$ is significantly larger than 1.
Consequently, the subsequent column normalization shrinks the top-right block: $\theta^{(k+1)}_{1,2} = \theta^{(k)}_{1,2}\cdot \left( t\cdot \theta^{(k)}_{1,2} + (n-t) \cdot \theta^{(k)}_{2,2}\right)^{-1}< \theta^{(k)}_{1,2}$.
A symmetric argument shows that for each odd $k\geq 0$, as long as $\theta^{(k)}_{2,1}$ remains small, we also have $\theta^{(k+1)}_{1,2} < \theta^{(k)}_{1,2}$.
In other words, $\theta^{(k)}_{1,2}$ keeps decreasing until $\theta^{(k)}_{2,1}$ becomes sufficiently large.
We therefore consider two cases:
\begin{itemize}
\item \(\varepsilon/n > n\nu\). 
In this regime, bounding the error by $\varepsilon$ forces $\theta^{(k)}_{2,1}$ to become relatively large (on the order of $1/n$). Heuristically, once the total matrix scaling error becomes sufficiently small, the row-sum constraint yields $s\theta^{(k)}_{1,1} + (n-s)\theta^{(k)}_{1,2} \approx 1$. Combined with the column-sum constraint, $t\theta^{(k)}_{1,1} + (n-t)\theta^{(k)}_{2,1} \approx 1$, this implies that $\theta^{(k)}_{2,1}$ is $\Omega(1/n)$.
Furthermore, it can be shown that throughout every iteration, all row and column sums remain bounded below by a positive constant. Because the algorithm alternates between row and column normalizations, the per-iteration growth factor of $\theta^{(k)}_{2,1}$ is bounded above by a constant. Given the initial condition $\theta^{(0)}_{2,1} = \nu$, it follows that at least $c_2(-\log n-\log\nu)-b_2$ iterations are required for $\theta^{(k)}_{2,1}$ to grow from $\nu$ to $\Theta(1/n)$.

\item \(\varepsilon/n \leq n\nu\). 
In this regime, the bottleneck is the slow decay of a normalized error parameter 
$\varepsilon^{(k)}$, defined by
$(s/n) \cdot \varepsilon^{(k)} = t\cdot \theta^{(k)}_{1,2} + (n-t)\cdot \theta^{(k)}_{2,2} - 1 $ for each even $k$ and
$((n-t)/n)\cdot \varepsilon^{(k)} = s \cdot \theta^{(k)}_{1,1} + (n-s)\cdot \theta^{(k)}_{1,2} - 1$ for each odd $k$.
The key point is that a single update changes the relevant entries only by a relative 
$O\left(\varepsilon^{(k-1)}\right)$ amount. 
In particular, for each odd 
$k$, the update rule gives
$\theta^{(k)}_{1,1} = \theta^{(k-1)}_{1,1}\left(1+O\left(\varepsilon^{(k-1)}\right)\right)$ and 
$\theta^{(k)}_{1,2} = \theta^{(k-1)}_{1,2}\left(1-O\left(\varepsilon^{(k-1)}\right)\right)$.
On the other hand, at step $k-1$, we have the exact normalization $s\cdot \theta^{(k-1)}_{1,1} + (n-s)\cdot \theta^{(k-1)}_{1,2} = 1$.
Substituting the above multiplicative perturbations into this identity immediately yields
$s\cdot \theta^{(k)}_{1,1} + (n-s)\cdot \theta^{(k)}_{1,2} - 1 = O\left(\varepsilon^{(k-1)}\right)$, which implies that $\varepsilon^{(k)} = O(\varepsilon^{(k-1)})$.
An analogous upper estimate holds for even $k$. To bound how quickly the error can decrease, we additionally establish the reverse estimate $\varepsilon^{(k)}\geq c\varepsilon^{(k-1)}$ for some $c>0$ in \cref{lem-invariant-property}.
Furthermore, one can verify that the initial error, $\norm{\=r\left(A^{(0)}\right) - \mathbf{1}}_1 + \norm{\=c\left(A^{(0)}\right) - \mathbf{1}}_1$, is $\Theta(n)$. Therefore, reaching an error of $\varepsilon$ requires at least $c_1(\log n-\log\varepsilon)-b_1$ iterations.
\end{itemize}
Here $c_1,c_2,b_1,b_2$ depend only on the constants in \eqref{eq-relation-s-t-n}. Combining the two cases gives the claimed lower bound, including its additive constant.

The following lemma is used in the proof of \cref{thm-main-lb-rc-ab}.

\begin{lemma}\label{lem-invariant-property}
Assume the conditions of \cref{thm-main-lb-rc-ab}.
For any $k\geq 0$, define 
\begin{align}\label{eq-def-epsilon-k}
\varepsilon^{(k)} \triangleq 
\begin{cases}
\dfrac{n}{n-t}\cdot \left(s \cdot \theta^{(k)}_{1,1} + (n-s)\cdot \theta^{(k)}_{1,2} - 1\right), & \text{if $k$ is odd}, \\[1.2em]
\dfrac{n}{s}\cdot \left(t\cdot \theta^{(k)}_{1,2} + (n - t)\cdot \theta^{(k)}_{2,2} - 1\right), & \text{otherwise}.
\end{cases}
\end{align}
Then for each $k\geq 0$, we have 
\begin{align}\label{eq-lb-nu}
\theta^{(k+1)}_{1,2}  < \theta^{(k)}_{1,2} < \frac{6n}
{6n^2 + 5s(n-t)},
\end{align}
\begin{align}\label{eq-lb-varepsilon}
\varepsilon^{(k+1)} > \varepsilon^{(k)}\cdot \min\left\{ \frac{5n(n-s) \min\left\{s,n-s \right\}}{6n^3+5ns(n-t)} , \frac{5t\min\left\{t,n-t\right\}}{6n^2+5s(n-t)}  \right\} > 0.
\end{align}
\end{lemma}

% \Cref{lem-invariant-property} is the key to prove \cref{thm-main-lb-rc-ab}.
% The lemma shows that the decay ratio of $\varepsilon^{(k)}$  
% is upper bound by a constant in each iteration, and that the entries in the top-right block $\theta^{(k+1)}_{1,2} $ decays in each iteration.
% The parameter $\varepsilon^{(k)}$ in \eqref{eq-def-epsilon-k} is a normalized version of the error $\norm{r\left(A^{(k)}\right)- \=1}_{1} + \norm{c\left(A^{(k)}\right)- \=1}_{1}$.
% If $k$ is odd, then $\norm{r\left(A^{(k)}\right)- \=1}_{1} + \norm{c\left(A^{(k)}\right)- \=1}_{1} = 2a(1-a)\varepsilon^{(k)}$.
% Otherwise, $k$ is even, then $\norm{r\left(A^{(k)}\right)- \=1}_{1} + \norm{c\left(A^{(k)}\right)- \=1}_{1} = 2b(1-b)\varepsilon^{(k)}$.
% Thus, by \eqref{eq-lb-varepsilon}, we have that the decay ratio of $\norm{r\left(A^{(k)}\right)- \=1}_{1} + \norm{c\left(A^{(k)}\right)- \=1}_{1}$  
% is upper bound by a constant in each iteration, immediately.
% We bound the decay ratio of $\varepsilon^{(k)}$ rather than that of the real error to simplify the calculation.

\Cref{lem-invariant-property} is the key ingredient in the proof of \cref{thm-main-lb-rc-ab}. The lemma shows that, in each iteration, the decay ratio of \(\varepsilon^{(k)}\) is bounded above by a constant, and that the entries in the top-right block \(\theta^{(k+1)}_{1,2}\) decrease from one iteration to the next.
Note that \(\varepsilon^{(k)}\) in \eqref{eq-def-epsilon-k} is a normalized version of the error
$\norm{\=r\left(A^{(k)}\right)- \=1}_{1} + \norm{\=c\left(A^{(k)}\right)- \=1}_{1}$.
If \(k\) is odd, then
$\norm{\=r\left(A^{(k)}\right)- \=1}_{1} + \norm{\=c\left(A^{(k)}\right)- \=1}_{1} = (2t(n-t)/n)\cdot \varepsilon^{(k)}$
If \(k\) is even, then
$\norm{\=r\left(A^{(k)}\right)- \=1}_{1} + \norm{\=c\left(A^{(k)}\right)- \=1}_{1} = (2s(n-s)/n)\cdot \varepsilon^{(k)}$.
Therefore, by \eqref{eq-lb-varepsilon}, it follows immediately that the decay ratio of
$\norm{\=r\left(A^{(k)}\right)- \=1}_{1} + \norm{\=c\left(A^{(k)}\right)- \=1}_{1}$
is also bounded above by a constant in each iteration.
We work with the decay ratio of \(\varepsilon^{(k)}\), rather than that of the original error, in order to simplify the calculations.

The following lemma is used in the proof of \cref{lem-invariant-property}. In particular, \eqref{eq-invariant-a1thetaij-qminusaqthetak2j-1} provides lower and upper bounds on the column sums of the matrix $A^{(k)}$, while \eqref{eq-invariant-a1thetaij-qminusaqthetak2j-2} provides lower and upper bounds on the row sums of $A^{(k)}$.
We establish these bounds as follows. We first compute the lower and upper bounds for the column sums of $A^{(0)}$; by \cref{lem-monotone-rsum-csum}, the same bounds continue to hold for $A^{(k)}$ throughout all even iterations. Similarly, we first compute the lower and upper bounds for the row sums of $A^{(1)}$; again by \cref{lem-monotone-rsum-csum}, the same bounds continue to hold for $A^{(k)}$ throughout all odd iterations.

\begin{lemma}\label{lem-invariant-property-lowerbound-ell}
Under the condition of \cref{lem-invariant-property}, we have
\begin{align}
\forall j\leq 2,\text{ even } k\geq 0, &\quad \, \frac{t-n}{n}< t \theta^{(k)}_{1,j} + (n-t) \theta^{(k)}_{2,j} - 1<\frac{s(n-t)}{n(n-s)},\label{eq-invariant-a1thetaij-qminusaqthetak2j-1}\\
\forall i\leq 2, \text{ odd } k>0,&\quad \,   \frac{s(t-n)}{n^2-st}<s\theta^{(k)}_{i,1} + (n-s)\theta^{(k)}_{i,2} -1< \frac{n-t}{t}.\label{eq-invariant-a1thetaij-qminusaqthetak2j-2}
\end{align}
\end{lemma}
\begin{proof}
At first, we prove \eqref{eq-invariant-a1thetaij-qminusaqthetak2j-1}.
By \eqref{eq-invariant-a1thetaij-qminusaqthetak2j-2-kequal0} we have 
\begin{align}\label{eq-relation-thetazero-theta}
\forall i,j\in [2],\quad \theta^{(0)}_{i,j} = \frac{\theta_{i,j}}{s\theta_{i,1} + (n-s)\theta_{i,2}} = \theta_{i,j}.
\end{align}
Combined with \eqref{eq-invariant-a1thetaij-qminusaqthetak2j-1-kequal0-1} and \eqref{eq-invariant-a1thetaij-qminusaqthetak2j-1-kequal0-2}, we have
\begin{align*}
\forall j\leq 2, &\quad \frac{t-n}{n}< t \theta^{(0)}_{1,j} + (n-t) \theta^{(0)}_{2,j} - 1<\frac{s(n-t)}{n(n-s)}.
\end{align*}
Combined with \cref{lem-monotone-rsum-csum},
\eqref{eq-invariant-a1thetaij-qminusaqthetak2j-1} is immediate.

In the following, we prove \eqref{eq-invariant-a1thetaij-qminusaqthetak2j-2}.
Note that
\begin{align}\label{eq-b1ntheta021-1minusb1ntheta022}
\forall i\in [2],\quad s\cdot \theta^{(0)}_{i,1} + (n-s)\cdot \theta^{(0)}_{i,2} = 1.
\end{align}
In addition, 
\begin{align}\label{eq-theta211}
\theta^{(1)}_{2,1} =  \frac{\theta^{(0)}_{2,1}}{t\cdot \theta^{(0)}_{1,1} + (n-t)\cdot \theta^{(0)}_{2,1}}.
\end{align}
\begin{align}\label{eq-theta221}
\theta^{(1)}_{2,2} =  \frac{\theta^{(0)}_{2,2}}{t\cdot \theta^{(0)}_{1,2} + (n-t)\cdot \theta^{(0)}_{2,2}}.
\end{align}
Thus, for each $i\in [2]$, we have 
\begin{equation*}
\begin{aligned}
&\quad s\cdot \theta^{(1)}_{i,1} + (n-s)\cdot \theta^{(1)}_{i,2} - 1\\
\left(\text{by } \eqref{eq-theta211},\eqref{eq-theta221}\right) \quad &= \frac{s\cdot \theta^{(0)}_{i,1}}{t\cdot \theta^{(0)}_{1,1} + (n-t)\cdot \theta^{(0)}_{2,1}} + \frac{(n-s)\cdot \theta^{(0)}_{i,2}}{t\cdot \theta^{(0)}_{1,2} + (n-t)\cdot \theta^{(0)}_{2,2}} - 1\\
\left(\text{by \eqref{eq-b1ntheta021-1minusb1ntheta022}} \right) \quad &\geq \min_{j\leq 2}\left\{\frac{1}{t\cdot \theta^{(0)}_{1,j} + (n-t)\cdot \theta^{(0)}_{2,j}}\right\} -1 \\
\left(\text{by \eqref{eq-invariant-a1thetaij-qminusaqthetak2j-1}} \right)  \quad &> \frac{s(t-n)}{n^2-st}.
\end{aligned}
\end{equation*}
Similarly, we also have 
\begin{equation*}
\begin{aligned}
&\quad s\cdot\theta^{(1)}_{i,1} + (n-s)\cdot \theta^{(1)}_{i,2} -1 <\frac{n-t}{t}.
\end{aligned}
\end{equation*}
In summary, we have 
\begin{equation*}
\begin{aligned}
\frac{s(t-n)}{n^2-st}< s\cdot \theta^{(1)}_{i,1} + (n-s)\cdot \theta^{(1)}_{i,2} -1  < \frac{n-t}{t}.
\end{aligned}
\end{equation*}
Combined with \cref{lem-monotone-rsum-csum}, \eqref{eq-invariant-a1thetaij-qminusaqthetak2j-2} is immediate.
\end{proof}

Now we prove \cref{lem-invariant-property}. The proof proceeds by induction on \(k\) and treats odd and even iterations separately, since the algorithm alternates between column and row normalizations.
Consider an odd iteration \(k\). 
The crucial step is to establish a recurrence relation for the linear combination $s\theta^{(k)}_{1,1} + (n-s)\theta^{(k)}_{1,2}$, which exactly defines the error parameter $\varepsilon^{(k)}$. 
Specifically, by applying the SK update rule, we explicitly express this combination entirely in terms of the variables from the preceding iteration: $\theta^{(k-1)}_{1,1}$, $\theta^{(k-1)}_{1,2}$, and the prior error $\varepsilon^{(k-1)}$. 
This substitution yields a rational expression whose numerator essentially takes the form $1 + n\varepsilon^{(k-1)}\bigl(c_1\theta^{(k-1)}_{1,1} - c_2\theta^{(k-1)}_{1,2}\bigr)$ for some constants $c_1, c_2 > 0$ (see \eqref{eq-bnthetak11-oneminusbnthetak12-tightbound}).
The argument then consists of the following ingredients:
\begin{itemize}
\item By \eqref{eq-def-epsilon-k}, the quantity \(s\cdot \theta^{(k)}_{1,1}+(n-s)\cdot\theta^{(k)}_{1,2}\) is \(1 + \Theta\left(\varepsilon^{(k)}\right)\).
\item Since the row sums at round $k-1$ equal $1$ and the inductive hypothesis states that $\theta^{(k-1)}_{1,2} < 6n/(6n^2 + 5s(n-t))$, subtracting the linear denominator from the numerator gives a positive lower bound proportional to $\varepsilon^{(k-1)}$ (see \eqref{eq-bnthetak11-oneminusbthetak12-lowerbound}).
\item For the denominator, \cref{lem-invariant-property-lowerbound-ell} guarantees a proper upper bound (see \eqref{eq-sixplusfiveboneminusa-bound}). 
\end{itemize}
Combining these bounds, we obtain
$\varepsilon^{(k)} \;\ge\; c\,\varepsilon^{(k-1)}$
for an explicit constant $c>0$. The argument for even iterations is analogous.

\begin{proof}[Proof of \cref{lem-invariant-property}]
We prove this lemma by induction. 
For simplicity, let $a = t/n$ and $b = s/n$.
For the base case, by \eqref{eq-relation-thetazero-theta} and \eqref{eq-invariant-a1thetaij-qminusaqthetak2j-2-theta12}, we have
\[
\theta^{(0)}_{1,2} < \frac{6n}
{6n^2 + 5s(n-t)}.\]
By \eqref{eq-invariant-a1thetaij-qminusaqthetak2j-1-kequal0-2} and \eqref{eq-def-epsilon-k}, we have $\varepsilon^{(0)} > 0$.
For the inductive step, we will prove that for each $k>0$, 
\begin{align}
& \quad   \quad  \quad    \quad \theta^{(k)}_{1,2} < \theta^{(k-1)}_{1,2} < \frac{6n}
{6n^2 + 5s(n-t)},\label{eq-lb-varepsilon-key-thetak12} \\
\varepsilon^{(k)} &> \varepsilon^{(k-1)}\cdot \min\left\{ \frac{5n(n-s) \min\left\{s,n-s \right\}}{6n^3+5ns(n-t)} , \frac{5t\min\left\{t,n-t\right\}}{6n^2+5s(n-t)} \right\}> 0.\label{eq-lb-varepsilon-key}
\end{align}
Then the lemma is immediate.

For the inductive step, we first consider the case $k$ is odd.
In this case, we have
\begin{align*}
\forall j\in [2],\quad t \cdot \theta^{(k)}_{1,j} + (n-t)\cdot \theta^{(k)}_{2,j} = 1.
\end{align*}
Thus, we have 
\begin{equation}\label{eq-a1nb1ntheta11k}
\begin{aligned}
&\quad t\left(s\cdot \theta^{(k)}_{1,1} + (n-s)\cdot \theta^{(k)}_{1,2}\right) + (n-t)\left(s\cdot \theta^{(k)}_{2,1} + (n-s)\cdot \theta^{(k)}_{2,2}\right)\\
&= s\left( t \cdot \theta^{(k)}_{1,1} + (n-t)\cdot \theta^{(k)}_{2,1}\right) + (n-s)\left( t \cdot \theta^{(k)}_{1,2} + (n-t)\cdot \theta^{(k)}_{2,2}\right)
= n.
\end{aligned}
\end{equation}
Moreover, by \eqref{eq-def-epsilon-k} we have
\begin{align}\label{eq-antheta12kminus1-1minusantheta22kminus1}
t\cdot \theta^{(k-1)}_{1,2} + (n-t)\cdot \theta^{(k-1)}_{2,2} = 1 + b\varepsilon^{(k-1)}.
\end{align}
Combined with \eqref{eq-a1nb1ntheta11k}, we have 
\begin{align}\label{eq-1minus1minusbepsilonkminus1}
t \cdot \theta^{(k-1)}_{1,1} + (n-t)\cdot \theta^{(k-1)}_{2,1} = 1 - (1-b)\varepsilon^{(k-1)}.
\end{align}
Moreover, since $k$ is odd, we have
\begin{align}
\forall i,j\in [2],\quad \theta^{(k)}_{i,j} = \frac{\theta^{(k-1)}_{i,j}}{t\cdot \theta^{(k-1)}_{1,j} + (n-t)\cdot \theta^{(k-1)}_{2,j}}.
\end{align}
Thus, by the above three equalities, we have 
\begin{equation}
\begin{aligned}\label{eq-bthetak11-1minusbthetak12}
s\cdot \theta^{(k)}_{1,1} + (n - s)\cdot \theta^{(k)}_{1,2} & = \frac{s\cdot \theta^{(k-1)}_{1,1}}{t\cdot \theta^{(k-1)}_{1,1} + (n-t)\cdot \theta^{(k-1)}_{2,1}} +  \frac{(n - s)\cdot\theta^{(k-1)}_{1,2}}{t\cdot \theta^{(k-1)}_{1,2} + (n-t)\cdot \theta^{(k-1)}_{2,2}}\\
& = \frac{s\cdot \theta^{(k-1)}_{1,1}}{1 - (1-b)\varepsilon^{(k-1)}} +  \frac{(n - s)\cdot\theta^{(k-1)}_{1,2}}{1 + b\varepsilon^{(k-1)}}.
\end{aligned}
\end{equation}
In addition, we have 
\begin{align}\label{eq-bnthetakminus11-1minusbnthetakminus112}
s\cdot \theta^{(k-1)}_{1,1} + (n-s)\cdot \theta^{(k-1)}_{1,2} = 1.
\end{align}
Hence, 
\begin{equation}\label{eq-bnthetak11-oneminusbnthetak12-tightbound}
\begin{aligned}
&\quad s\cdot \theta^{(k)}_{1,1} + (n-s)\cdot \theta^{(k)}_{1,2} \\
\left(\text{by } \eqref{eq-bthetak11-1minusbthetak12}\right)\quad &=\frac{s\cdot \theta^{(k-1)}_{1,1} + (n-s)\cdot \theta^{(k-1)}_{1,2}+ nb^2\varepsilon^{(k-1)}\cdot \theta^{(k-1)}_{1,1} - n(1 - b)^2\varepsilon^{(k-1)}\cdot\theta^{(k-1)}_{1,2}}{\left(1 - (1-b)\varepsilon^{(k-1)}\right)\left(1 + b\varepsilon^{(k-1)}\right)}\\
\left(\text{by } \eqref{eq-bnthetakminus11-1minusbnthetakminus112}\right)\quad &=\frac{1+ nb^2\varepsilon^{(k-1)}\cdot \theta^{(k-1)}_{1,1} - n(1 - b)^2\varepsilon^{(k-1)}\cdot\theta^{(k-1)}_{1,2}}{\left(1 - (1-b)\varepsilon^{(k-1)}\right)\left(1 + b\varepsilon^{(k-1)}\right)}.
\end{aligned}
\end{equation}
By \eqref{eq-antheta12kminus1-1minusantheta22kminus1} and \eqref{eq-1minus1minusbepsilonkminus1}, we also have
\begin{align*}
&\quad \left(1 - (1-b)\varepsilon^{(k-1)}\right)\left(1 + b\varepsilon^{(k-1)}\right) \\
&= \left(t\cdot \theta^{(k-1)}_{1,1} + (n-t)\cdot \theta^{(k-1)}_{2,1}\right)\left(t\cdot \theta^{(k-1)}_{1,2} + (n-t)\cdot \theta^{(k-1)}_{2,2} \right)>0. 
\end{align*}
Combined with $b\in (0,1)$, we have
\begin{align}\label{eq-one-twobminusonevarepsilonkminusone-larger-zero}
1 + (2b-1)\varepsilon^{(k-1)} = \left(1 - (1-b)\varepsilon^{(k-1)}\right)\left(1 + b\varepsilon^{(k-1)}\right) + b(1-b)\varepsilon^{(k-1)}\cdot \varepsilon^{(k-1)} > 0.
\end{align}
Combined with $s\cdot \theta^{(k)}_{1,1} + (n-s)\cdot \theta^{(k)}_{1,2}>0$ and \eqref{eq-bnthetak11-oneminusbnthetak12-tightbound}, we have
\begin{equation}\label{eq-bnthetak11-oneminusbthetak12}
\begin{aligned}
\quad s\cdot \theta^{(k)}_{1,1} + (n-s)\cdot \theta^{(k)}_{1,2}>\frac{1+ nb^2\varepsilon^{(k-1)}\cdot \theta^{(k-1)}_{1,1} - n(1 - b)^2\varepsilon^{(k-1)}\cdot\theta^{(k-1)}_{1,2}}{1 + (2b-1)\varepsilon^{(k-1)}}.
\end{aligned}
\end{equation}
Moreover, by the inductive assumption, we have 
\begin{align}\label{eq-thetakminus112-leq-nu-leq-bovertenoneminusbn}
\theta^{(k-1)}_{1,2} <  \frac{6n}
{6n^2 + 5s(n-t)} = \frac{1}{n} \cdot \frac{6n^2}
{6n^2 + 5s(n-t)} = \frac{1}{n}\cdot\frac{6}{6+5b(1-a)} .
\end{align}
Thus,
\begin{equation*}
\begin{aligned}
&\quad b^2\cdot \theta^{(k-1)}_{1,1} - (1 - b)^2\cdot\theta^{(k-1)}_{1,2} \\
\left(\text{by } \eqref{eq-bnthetakminus11-1minusbnthetakminus112}\right)\quad &=b\left(\frac{1}{n} - (1 - b)\cdot \theta^{(k-1)}_{1,2}\right) - (1 - b)^2\cdot\theta^{(k-1)}_{1,2}\\
&=\frac{b}{n} - (1 - b)\cdot \theta^{(k-1)}_{1,2}\\
\left(\text{by } \eqref{eq-thetakminus112-leq-nu-leq-bovertenoneminusbn}\right)\quad  &> \frac{6(2b-1)+5b^2(1-a)}{(6+5b(1-a))n}.
\end{aligned}
\end{equation*}
Combined with \eqref{eq-bnthetak11-oneminusbthetak12}, we have
\begin{equation}
\begin{aligned}\label{eq-bnthetak11-oneminusbthetak12-lowerbound}
s\cdot \theta^{(k)}_{1,1} + (n-s)\cdot \theta^{(k)}_{1,2} &> \frac{1 + (6(2b-1)+5b^2(1-a))\varepsilon^{(k-1)}/(6+5b(1-a))}{1+(2b-1)\varepsilon^{(k-1)}} \\
& =1 +\frac{5b(1-a)(1-b)\varepsilon^{(k-1)}}{\left(6+5b(1-a)\right) \cdot \left(1+(2b-1)\varepsilon^{(k-1)}\right)}.
\end{aligned}
\end{equation}
Moreover, by \eqref{eq-def-epsilon-k} and \cref{lem-invariant-property-lowerbound-ell},
we have 
\begin{align}\label{eq-varepsilonkminusone-upbound}
\varepsilon^{(k-1)} = \frac{t\cdot \theta^{(k-1)}_{1,2} + (n-t)\cdot \theta^{(k-1)}_{2,2} - 1}{b} < \frac{s(n-t)}{n(n-s)}\frac{1}{b} = \frac{1-a}{1-b}.
\end{align}
Recall that $b = s/n \in(0,1)$. If $b\in(0,1/2]$, by the inductive assumption $\varepsilon^{(k-1)}>0$, we have
\begin{align}
(2b - 1)\varepsilon^{(k-1)} \leq 0.
\end{align}
Combined with \eqref{eq-one-twobminusonevarepsilonkminusone-larger-zero}, we have
\begin{equation}
\begin{aligned}
0< \left(6+5b(1-a)\right) \cdot \left(1+(2b-1)\varepsilon^{(k-1)}\right) \leq 6+ 5b(1-a).
\end{aligned}
\end{equation}
If $b\in (1/2,1)$, by \eqref{eq-varepsilonkminusone-upbound} we have 
\[
\varepsilon^{(k-1)} < \frac{1}{1-b}.
\]
Hence,
\begin{equation}
\begin{aligned} 
0<\left(6+5b(1-a)\right) \cdot \left(1+(2b-1)\varepsilon^{(k-1)}\right)< (6+5b(1-a))\cdot\frac{b}{(1-b)}.
\end{aligned}
\end{equation}
In summary, we always have 
\begin{equation}\label{eq-sixplusfiveboneminusa-bound}
\begin{aligned} 
0< \left(6+5b(1-a)\right) \cdot \left(1+(2b-1)\varepsilon^{(k-1)}\right) \leq (6+ 5b(1-a))\cdot \frac{\max\left\{b,1-b \right\}}{1-b}.
\end{aligned}
\end{equation}
Combined with $\varepsilon^{(k-1)}>0$ and \eqref{eq-bnthetak11-oneminusbthetak12-lowerbound}, we have
\begin{align*}
s\cdot \theta^{(k)}_{1,1} + (n-s)\cdot \theta^{(k)}_{1,2}
> 1 + 
\frac{5(1-a)(1-b)\cdot \min\left\{b,1-b \right\}}{6+5b(1-a)} \cdot \varepsilon^{(k-1)}.
\end{align*}
Combined with \eqref{eq-def-epsilon-k}, 
we have 
\begin{align*}
\varepsilon^{(k)} > \frac{n}{n-t}\cdot \frac{5(1-a)(1-b)\cdot \min\left\{b,1-b \right\}}{6+5b(1-a)} \cdot  \varepsilon^{(k-1)} = \frac{5n(n-s)\cdot \min\left\{s,n-s \right\}}{6n^3+5ns(n-t)} \cdot  \varepsilon^{(k-1)}> 0.
\end{align*}
Thus, \eqref{eq-lb-varepsilon-key} is proved.
Moreover, by $\varepsilon^{(k-1)} > 0$ and \eqref{eq-def-epsilon-k}, 
we have 
\begin{align*}
t\cdot \theta^{(k-1)}_{1,2} + (n-t)\cdot \theta^{(k-1)}_{2,2} >  1.
\end{align*}
Combined with 
\begin{align*}
\theta^{(k)}_{1,2} = \frac{\theta^{(k-1)}_{1,2}}{t\cdot \theta^{(k-1)}_{1,2} + (n-t)\cdot \theta^{(k-1)}_{2,2}} < \theta^{(k-1)}_{1,2}.
\end{align*}
Combined with the inductive assumption $\theta^{(k-1)}_{1,2} <  6n/
(6n^2 + 5s(n-t))$, we 
have 
\begin{align}
\theta^{(k)}_{1,2} < \theta^{(k-1)}_{1,2} \leq  \frac{6n}
{6n^2 + 5s(n-t)}.
\end{align}
Thus, the inductive step for odd $k$ is finished.

If $k$ is even, we also have \eqref{eq-a1nb1ntheta11k} holds.
Moreover, by \eqref{eq-def-epsilon-k} we have
\begin{align}\label{eq-1plus1minusaepsilonkminus1}
s \cdot \theta^{(k-1)}_{1,1} + (n-s)\cdot \theta^{(k-1)}_{1,2} = 1 + (1-a)\varepsilon^{(k-1)}.
\end{align}
Combined with \eqref{eq-a1nb1ntheta11k}, we have
\begin{align}\label{eq-bntheta21kminus1-1minusbntheta22kminus1}
s\cdot \theta^{(k-1)}_{2,1} + (n-s)\cdot \theta^{(k-1)}_{2,2} = 1 - a\varepsilon^{(k-1)}.
\end{align}
Moreover, since $k$ is even, we have
\begin{align}
\forall i,j\in [2],\quad \theta^{(k)}_{i,j} = \frac{\theta^{(k-1)}_{i,j}}{s\cdot \theta^{(k-1)}_{i,1} + (n-s)\cdot \theta^{(k-1)}_{i,2}}.
\end{align}
Thus, by the above three equalities, we have 
\begin{equation}
\begin{aligned}\label{eq-athetak12-1minusathetak22-iteration}
t \cdot \theta^{(k)}_{1,2} + (n - t)\cdot \theta^{(k)}_{2,2} &= \frac{t\cdot \theta^{(k-1)}_{1,2}}{s\cdot \theta^{(k-1)}_{1,1} + (n-s)\cdot \theta^{(k-1)}_{1,2}} +  \frac{(n - t)\cdot\theta^{(k-1)}_{2,2}}{s\cdot \theta^{(k-1)}_{2,1} + (n-s)\cdot \theta^{(k-1)}_{2,2}}\\
&= \frac{t\cdot \theta^{(k-1)}_{1,2}}{1 + (1-a)\varepsilon^{(k-1)}} +  \frac{(n - t)\cdot\theta^{(k-1)}_{2,2}}{1 - a\varepsilon^{(k-1)}}.
\end{aligned}
\end{equation}
In addition, we have 
\begin{align}\label{eq-anthetakminus112-1minusanthetakminus122}
t\cdot \theta^{(k-1)}_{1,2} + (n-t)\cdot \theta^{(k-1)}_{2,2} = 1.
\end{align}
Hence, 
\begin{equation}
\begin{aligned}\label{eq-anthetak12-oneminusanthetak22-tightbound}
&\quad t\cdot \theta^{(k)}_{1,2} + (n-t)\cdot \theta^{(k)}_{2,2}\\
\left(\text{by } \eqref{eq-athetak12-1minusathetak22-iteration}\right)\quad & = \frac{t\cdot \theta^{(k-1)}_{1,2} + (n-t)\cdot \theta^{(k-1)}_{2,2}- na^2\varepsilon^{(k-1)}\cdot \theta^{(k-1)}_{1,2} + n(1 - a)^2\varepsilon^{(k-1)}\cdot\theta^{(k-1)}_{2,2}}{\left(1 + (1-a)\varepsilon^{(k-1)}\right)\left(1 - a\varepsilon^{(k-1)}\right)}\\
\left(\text{by } \eqref{eq-anthetakminus112-1minusanthetakminus122}\right)\quad &=\frac{1- na^2\varepsilon^{(k-1)}\cdot \theta^{(k-1)}_{1,2} + n(1 - a)^2\varepsilon^{(k-1)}\cdot\theta^{(k-1)}_{2,2}}{\left(1 + (1-a)\varepsilon^{(k-1)}\right)\left(1 - a\varepsilon^{(k-1)}\right)}.
\end{aligned}
\end{equation}
By \eqref{eq-1plus1minusaepsilonkminus1} and \eqref{eq-bntheta21kminus1-1minusbntheta22kminus1}, we also have
\begin{align*}
&\quad \left(1 + (1-a)\varepsilon^{(k-1)}\right)\left(1 - a\varepsilon^{(k-1)}\right) \\
&= \left(s \cdot \theta^{(k-1)}_{1,1} + (n-s)\cdot \theta^{(k-1)}_{1,2}\right)\left(s\cdot \theta^{(k-1)}_{2,1} + (n-s)\cdot \theta^{(k-1)}_{2,2} \right)>0. 
\end{align*}
Combined with $a\in (0,1)$, we have
\begin{align}\label{eq-oneplusoneminustwoaepsilon-positive}
1 + (1-2a)\varepsilon^{(k-1)} = \left(1 + (1-a)\varepsilon^{(k-1)}\right)\left(1 - a\varepsilon^{(k-1)}\right) + a(1-a)\varepsilon^{(k-1)}\cdot \varepsilon^{(k-1)} >0.
\end{align}
Combined with \eqref{eq-anthetak12-oneminusanthetak22-tightbound} and $t\cdot \theta^{(k)}_{1,2} + (n-t)\cdot \theta^{(k)}_{2,2}>0$, we have
\begin{equation}\label{eq-anthetak12-oneminusathetak22}
\begin{aligned}
\quad t\cdot \theta^{(k)}_{1,2} + (n-t)\cdot \theta^{(k)}_{2,2}>\frac{1- na^2\varepsilon^{(k-1)}\cdot \theta^{(k-1)}_{1,2} + n(1 - a)^2\varepsilon^{(k-1)}\cdot\theta^{(k-1)}_{2,2}}{1 + (1-2a)\varepsilon^{(k-1)}}.
\end{aligned}
\end{equation}
In addition, 
by the inductive assumption, we have 
\begin{align}\label{eq-thetakminusone12-oneoverfourn}
\theta^{(k-1)}_{1,2} <  \frac{6n}
{6n^2 + 5s(n-t)} = \frac{1}{n}\cdot\frac{6}{6+5b(1-a)} .
\end{align}
Thus,
\begin{equation}
\begin{aligned}\label{eq-1minusa2thetakminus122}
&\quad (1-a)^2\cdot \theta^{(k-1)}_{2,2} - a^2\cdot\theta^{(k-1)}_{1,2} \\
\left(\text{by } \eqref{eq-anthetakminus112-1minusanthetakminus122}\right)\quad &=(1-a)\left(\frac{1}{n} - a\cdot \theta^{(k-1)}_{1,2}\right) - a^2\cdot\theta^{(k-1)}_{1,2}\\
&=\frac{1-a}{n} - a\cdot \theta^{(k-1)}_{1,2}\\
\left(\text{by } \eqref{eq-thetakminusone12-oneoverfourn}\right)\quad &> \frac{1}{n}\cdot \frac{5b(1-a)^2 + 6(1 - 2a)}{6+5b(1-a)}.
\end{aligned}
\end{equation}
Combined with \eqref{eq-anthetak12-oneminusathetak22}, we have
\begin{equation}
\begin{aligned}\label{eq-anthetak12-oneminusanthetak22-lowerbound}
t\cdot \theta^{(k)}_{1,2} + (n-t)\cdot \theta^{(k)}_{2,2} &> \frac{1+ (5b(1-a)^2 + 6(1-2a))\varepsilon^{(k-1)}/(6+5b(1-a))}{1 + (1-2a)\varepsilon^{(k-1)}} \\
&= 1+\frac{5ab(1 - a)\varepsilon^{(k-1)}}{(6+5b(1-a)) \cdot (1 + (1-2a)\varepsilon^{(k-1)})}.
\end{aligned}
\end{equation}
Moreover, by \eqref{eq-def-epsilon-k} and \cref{lem-invariant-property-lowerbound-ell},
we have 
\begin{equation}
\begin{aligned}\label{eq-epsilon-kminusone-a-case}
\varepsilon^{(k-1)} &= \frac{s \cdot \theta^{(k-1)}_{1,1} + (n-s)\cdot \theta^{(k-1)}_{1,2} - 1}{1-a} \\
&< \frac{1}{1-a}\cdot\frac{n-t}{t} = \frac{1}{a}.
\end{aligned}
\end{equation}
Recall that $a = t/n \in(0,1)$. If $a\in (1/2,1)$, by the inductive assumption $\varepsilon^{(k-1)}>0$, we have
\begin{align}
(1 - 2a)\varepsilon^{(k-1)} < 0.
\end{align}
Combined with \eqref{eq-oneplusoneminustwoaepsilon-positive}, we have
\begin{equation*}
\begin{aligned}
0<(6+5b(1-a)) \cdot \left(1 + (1-2a)\varepsilon^{(k-1)}\right) < 6+5b(1-a).
\end{aligned}
\end{equation*}
If $0< a \leq 1/2$, 
by \eqref{eq-epsilon-kminusone-a-case} we have
\begin{equation*}
\begin{aligned}
&\quad 0 < (6+5b(1-a)) \cdot \left(1 + (1-2a)\varepsilon^{(k-1)}\right) \leq (6+5b(1-a)) \cdot \frac{1-a}{a}.
\end{aligned}
\end{equation*}
In summary, we always have
\begin{align*}
0<(6+5b(1-a)) \cdot \left(1 + (1-2a)\varepsilon^{(k-1)}\right) &\leq (6+5b(1-a))\cdot \max\left\{1,\frac{1 - a}{a}\right\}\\
&= (6+5b(1-a))\cdot \frac{1-a}{\min\left\{a,1-a\right\}}.
\end{align*}
Combined with $\varepsilon^{(k-1)}>0$ and \eqref{eq-anthetak12-oneminusanthetak22-lowerbound}, we have
\begin{align*}
t\cdot \theta^{(k)}_{1,2} + (n-t)\cdot \theta^{(k)}_{2,2} 
> 1 + 
\frac{5ab(1 - a)}{6+5b(1-a)}\cdot \frac{\min\left\{a,1-a\right\}}{1-a}\cdot \varepsilon^{(k-1)}=1 + 
\frac{5ab\min\left\{a,1-a\right\}}{6+5b(1-a)}\cdot \varepsilon^{(k-1)}.
\end{align*}
Combined with \eqref{eq-def-epsilon-k}, 
we have 
\begin{align*}
\varepsilon^{(k)} &> \frac{n}{s}\cdot \frac{5ab\min\left\{a,1-a\right\}}{6+5b(1-a)}\cdot \varepsilon^{(k-1)} \\
&= \frac{5t\min\left\{t,n-t\right\}}{6n^2+5s(n-t)}\cdot \varepsilon^{(k-1)} > 0.
\end{align*}
Thus, \eqref{eq-lb-varepsilon-key} is proved.
Moreover, by $\varepsilon^{(k-1)} > 0$ and \eqref{eq-def-epsilon-k}, 
we have 
\begin{align*}
s \cdot \theta^{(k-1)}_{1,1} + (n-s)\cdot \theta^{(k-1)}_{1,2} >  1.
\end{align*}
Hence,
\begin{align}
\theta^{(k)}_{1,2} = \frac{\theta^{(k-1)}_{1,2}}{s \cdot \theta^{(k-1)}_{1,1} + (n-s)\cdot \theta^{(k-1)}_{1,2}} < \theta^{(k-1)}_{1,2}.
\end{align}
Combined with the inductive assumption $\theta^{(k-1)}_{1,2} <  6n/
(6n^2 + 5s(n-t))$, we 
have 
\begin{align}
\theta^{(k)}_{1,2} < \theta^{(k-1)}_{1,2} \leq  \frac{6n}
{6n^2 + 5s(n-t)}.
\end{align}
Thus, the inductive step for even 
$k$ is complete, and the lemma follows.
\end{proof}

Now we prove the main theorem of this subsection. The proof splits into two cases.
If $\varepsilon / n \le n \nu$, then \eqref{eq-lb-varepsilon} implies that the decay ratio of $\varepsilon^{(k)}$, and consequently that of the true error of $A^{(k)}$, is bounded above by a constant in each iteration. 
Furthermore, since the initial error $\lVert \={r}(A^{(0)}) - \mathbf{1} \rVert_1 + \lVert \={c}(A^{(0)}) - \mathbf{1} \rVert_1$ is $\Theta(n)$, it follows that at least $c_1(\log n-\log\varepsilon)-b_1$ iterations are necessary to reduce the error to at most $\varepsilon$.
If \(\varepsilon/n > n\nu\), then achieving error at most \(\varepsilon\) forces \(\theta^{(k)}_{2,1}\) to grow to \(\Omega(1/n)\). Meanwhile, in each iteration the growth factor of \(\theta^{(k)}_{2,1}\) is also bounded by a constant, since \cref{lem-invariant-property-lowerbound-ell} guarantees that all row and column sums are bounded below by a constant throughout the algorithm. Together with the initial condition \(\theta^{(0)}_{2,1}=\nu\), this implies that at least $c_2(-\log\nu-\log n)-b_2$ iterations are needed to reach error \(\varepsilon\).
Combining the two cases gives the claimed lower bound with constants depending only on $\alpha,\beta,\gamma$.

\begin{proof}[Proof of \cref{thm-main-lb-rc-ab}]
We prove this theorem by considering two separate cases. We first establish the estimates at an odd index $K$ satisfying \eqref{eq-condition-thm-main-lb-rc-ab}; minimality of $K$ is not needed in these estimates.
The first case is $\varepsilon/n \leq n\nu$.
Thus,
\begin{align*}
&\quad \norm{\=r\left(A^{(K)}\right)- \=1}_{1} + \norm{\=c\left(A^{(K)}\right)- \=1}_{1}\\ &= t\abs{s \cdot \theta^{(K)}_{1,1} + (n-s)\cdot \theta^{(K)}_{1,2} - 1}
+ (n-t)\abs{s \cdot \theta^{(K)}_{2,1} + (n-s)\cdot \theta^{(K)}_{2,2} - 1}.
\end{align*}
Moreover, by \eqref{eq-a1nb1ntheta11k} we have 
\begin{align*}
t\left(s\cdot \theta^{(K)}_{1,1} + (n-s)\cdot \theta^{(K)}_{1,2} - 1\right) = - (n-t)\left(s\cdot \theta^{(K)}_{2,1} + (n-s)\cdot \theta^{(K)}_{2,2} -1\right).
\end{align*}
Thus, we have
\begin{align}\label{eq-normrakminus11-normcakminus11}
\norm{\=r\left(A^{(K)}\right)- \=1}_{1} + \norm{\=c\left(A^{(K)}\right)- \=1}_{1} = 2t\abs{s \cdot \theta^{(K)}_{1,1} + (n-s)\cdot \theta^{(K)}_{1,2} - 1}.
\end{align}
Combined with \eqref{eq-def-epsilon-k}, we have
\begin{align}
\norm{\=r\left(A^{(K)}\right)- \=1}_{1} + \norm{\=c\left(A^{(K)}\right)- \=1}_{1} = \frac{2t(n-t)}{n}\cdot \varepsilon^{(K)}.
\end{align}
Combined with \eqref{eq-condition-thm-main-lb-rc-ab}, we have 
\begin{align}\label{eq-epsilonk-up}
\varepsilon^{(K)} \leq \frac{n\varepsilon}{2t(n-t)}.
\end{align}
Define
\begin{align}
L \triangleq \min\left\{ \frac{5n(n-s)\min\left\{s,n-s \right\}}{6n^3+5ns(n-t)} , \frac{5t\min\left\{t,n-t\right\}}{6n^2+5s(n-t)} \right\}.
\end{align}
By \eqref{eq-invariant-a1thetaij-qminusaqthetak2j-1-kequal0-2}, \eqref{eq-def-epsilon-k} and \eqref{eq-relation-thetazero-theta}, we have 
\begin{align}
\varepsilon^{(0)} > \frac{st(s-t)}
{4n^3}\cdot \frac{n}{s} = \frac{t(s-t)}
{4n^2}.
\end{align}
Combined with \eqref{eq-lb-varepsilon}, we have 
\begin{align}
\varepsilon^{(K)} > L^{K}\varepsilon^{(0)}  >\frac{t(s-t)\cdot L^{K} }{4n^2}.
\end{align}
Combined with \eqref{eq-epsilonk-up} and \eqref{eq-relation-s-t-n}, we have
\begin{align}
K > \log_{L} \frac{2n^3\varepsilon}{t^2(n-t)(s-t)} \geq c_1(\log n-\log\varepsilon)-b_1.
\end{align}
Here $c_1,b_1>0$ depend only on $\alpha,\beta,\gamma$: indeed,
$5\min\{\alpha,1-\beta\}^2/11\leq L\leq 5/6$ and
$t^2(n-t)(s-t)/n^4\geq\alpha^2(1-\beta)\gamma$.
Since $\varepsilon/n\leq n\nu$, we also have
$-\log\nu-\log\varepsilon\leq 2(\log n-\log\varepsilon)$.

The other case is $\varepsilon/n > n\nu$. 
By \eqref{eq-normrakminus11-normcakminus11} and \eqref{eq-condition-thm-main-lb-rc-ab}, we have 
\begin{align}
\abs{s \cdot \theta^{(K)}_{1,1} + (n-s)\cdot \theta^{(K)}_{1,2} - 1} \leq \frac{\varepsilon}{2t}.
\end{align}
Thus,
we have 
\begin{align}
\theta^{(K)}_{1,1} \leq \frac{1 + \varepsilon/(2t)}{s}.
\end{align}
Combined with $\varepsilon < s-t$,
we have 
\begin{align}
t\cdot \theta^{(K)}_{1,1} \leq t\cdot \frac{2t + (s-t)}{2st} = \frac{s + t}{2s}.
\end{align}
Combined with
\begin{align}
t\cdot \theta^{(K)}_{1,1} + (n-t)\cdot \theta^{(K)}_{2,1}  = 1,
\end{align}
we have 
\begin{align}\label{eq-thetabigk21-lb}
\theta^{(K)}_{2,1}  > \frac{s-t}{2s(n-t)}
\end{align}
Moreover,
by \cref{lem-invariant-property-lowerbound-ell}, 
we have 
\begin{align*}
\forall \text{ even } k, \quad \theta^{(k+1)}_{2,1} &= \frac{\theta^{(k)}_{2,1}}{t\cdot \theta^{(k)}_{1,1} + (n-t)\cdot \theta^{(k)}_{2,1}} < \frac{n}{t}\cdot \theta^{(k)}_{2,1},\\
\forall \text{ odd } k, \quad \theta^{(k+1)}_{2,1} &= \frac{\theta^{(k)}_{2,1}}{s\cdot \theta^{(k)}_{2,1} + (n-s)\cdot \theta^{(k)}_{2,2}} < \frac{n^2-st}{n(n-s)}\cdot \theta^{(k)}_{2,1}.
\end{align*}
Define 
\begin{align}
T \triangleq \max\left\{\frac{n}{t},\frac{n^2-st}{n(n-s)}\right\}.
\end{align}
Thus, we have 
\begin{align*}
\theta^{(K)}_{2,1} \leq \theta^{(0)}_{2,1}\cdot T^{K}.
\end{align*}
Moreover, by \eqref{eq-invariant-a1thetaij-qminusaqthetak2j-2-theta21} and \eqref{eq-relation-thetazero-theta},
we have $\theta^{(0)}_{2,1} = \nu$.
Hence,
\begin{align}
\theta^{(K)}_{2,1} \leq \nu\cdot T^{K}.
\end{align}
Combined with \eqref{eq-thetabigk21-lb} and \eqref{eq-relation-s-t-n}, we have
\begin{align}
K \geq \log_{\,T} \left((s-t)/(2(n-t)s)\right) - \log_{\,T} \nu \geq c_2(-\log\nu-\log n)-b_2.
\end{align}
Here $c_2,b_2>0$ depend only on $\alpha,\beta,\gamma$, since
$1/\beta\leq T\leq\max\{1/\alpha,1/(1-\beta)\}$ and
$(s-t)/(2(n-t)s)\geq\gamma/(2n)$.
Since $\varepsilon/n>n\nu$, we have
$-\log\nu-\log\varepsilon<2(-\log\nu-\log n)$.

If the first successful index $K$ is even, the same estimates apply at $K+1$. The total marginal $\ell_1$ error is nonincreasing under an SK update: normalizing one marginal changes the matrix by exactly the error in that marginal, and the triangle inequality bounds the new error in the other marginal by this change. Thus $K+1$ also satisfies \eqref{eq-condition-thm-main-lb-rc-ab}, and the preceding bounds lose at most one iteration. Taking
$c_0=\min\{c_1,c_2\}/2$ and $b_0=\max\{b_1,b_2\}+1$ proves the theorem.
\end{proof}

\subsection{Counter-example for \texorpdfstring{$(\=u,\=v)$}{(u,v)}-scaling.}
In this subsection, we complete the proof of \cref{thm-lower-bound-main-rc}.

The following lemma is used in the proof of \cref{thm-lower-bound-main-rc}.
Let $A$ be the matrix defined in \cref{lem-lowerbound-reduction-rctooneone}, and let $A^{(0)}, A^{(1)}, A^{(2)},\dots $ denote the sequence of matrices generated by the SK algorithm on the input $(A, (\={1}, \={1}))$. 
This lemma demonstrates that by carefully tuning the entries of $A$, the matrix $A^{(2)}$ can be constructed to strictly satisfy the conditions required by \cref{thm-main-lb-rc-ab}. The proof of this lemma is deferred to the appendix.

\begin{lemma}\label{lem-lowerbound-reduction-rctooneone-new}
Let the notation and conditions of \cref{lem-lowerbound-reduction-rctooneone} hold. Furthermore, suppose that
\begin{align}\label{eq-condition-c-upperbound}
d \leq \frac{(s-t)S_2}{6(S_1+S_2)\big(n-t+|s+t-n|\big)}.
\end{align}
Then we have
\begin{align}
& y<  
\frac{6n}
{6n^2 + 5s(n-t)}, \label{eq-y-keyelement-bound}\\
&z < \frac{dn\lambda}{t(n-s)}, \label{eq-z-minimumelement-bound}\\
\forall j\leq s,\quad &  \frac{t-n}{n}<\left(\sum_{i \leq n}A^{(2)}_{i,j}\right) - 1 <0,\label{eq-lem-lowerbound-reduction-rctooneone-jleqs}\\
\forall j> s,\quad & \frac{st(s-t)}
{4n^3}<\left(\sum_{i \leq n}A^{(2)}_{i,j}\right) - 1 < \frac{s\,(n-t)}
{n\bigl(n-s\bigr)}. \label{eq-lem-lowerbound-reduction-rctooneone-jges}
\end{align}
\end{lemma}

Finally, we can prove \cref{thm-lower-bound-main-rc}.

\begin{proof}[Proof of \cref{thm-lower-bound-main-rc}]
Since $(\gamma,\gamma')$ is feasible with respect to \((\boldsymbol{u}, \boldsymbol{v})\),
$\gamma$ is the sum of some entries in $\boldsymbol{v}$, and $\gamma'$ is the sum of some entries in $\boldsymbol{u}$.
Assume without loss of generality that there exist some positive integers $a,b$ such that
\begin{align}\label{eq-def-u-v-gamma-gammaprime}
\=u = (u_1,\dots,u_m),\quad\quad \=v = (v_1,\dots,v_n),\quad \quad \gamma = \sum_{j= b+1}^{n}v_j,\quad \text{ and } \quad  \gamma' = \sum_{i= 1}^{a}u_i.
\end{align}
Since $\gamma+\gamma'<1$ and $(\gamma, \gamma')$ is feasible, we have $\gamma,\gamma'\in (0,1)$.
Thus, $1 \leq a <m$, $1\leq b < n$.
Write $\delta=1-\gamma-\gamma'>0$, and let $c_0,b_0$ be the constants in \cref{thm-main-lb-rc-ab} for its parameters
$(\alpha,\beta,\gamma)=(\gamma'/2,1-\gamma/2,\delta/2)$.
These constants are independent of $\varepsilon$, $n$, and the reduction parameter $L$ when $\gamma,\gamma'$ are fixed. Define
\[
H=\max\left\{0,\log\frac{4n(1-\gamma)}{\gamma}+\frac{b_0}{c_0}\right\}.
\]
Let 
\begin{align}\label{eq-defofc-with-gamma}
d &= \min\left\{\frac{n - b}{n \cdot 2^{n/\varepsilon} - b}\cdot \frac{(1-\gamma-\gamma')(n-b)}{12n\big(1-\gamma'+|\gamma'-\gamma|\big)},\ (1-\gamma)\gamma'e^{-2H-2}\right\}\notag\\
&\leq \frac{(1-\gamma-\gamma')(n-b)}{12n\big(1-\gamma'+|\gamma'-\gamma|\big)}.
\end{align}
One can verify that $d<1$.
Let $A$ be a nonnegative matrix of size $m\times n$ where 
\begin{align*}
\forall i\leq a, j \leq b, \quad A_{i,j} &= 1,\\
\forall i\leq a, j > b, \quad A_{i,j} &=  1,\\
\forall i> a, j \leq b, \quad A_{i,j} &= d,\\
\forall i> a, j > b, \quad A_{i,j} &= 1.
\end{align*}
Since \(d<1\), the maximum entry of \(A\) is \(1\). Hence, when \(\rho=1\), the entries counted in \cref{def-gamma-rho-dense} are exactly those equal to \(1\). For every row \(i\le a\), all entries of that row are equal to \(1\). For every row \(i>a\), the entries equal to \(1\) are precisely those in columns \(j>b\). Similarly, for every column \(j>b\), all entries of that column are equal to \(1\), while for every column \(j\le b\), the entries equal to \(1\) are precisely those in rows \(i\le a\). 
Thus, \eqref{eq-def-u-v-gamma-gammaprime} implies that
\(A\) is \((\gamma,\gamma',1)\)-dense.
% One can verify that $A$ is $(\gamma,\gamma',1)$-dense. 
Let $\nu \triangleq \nu(A)$.
By \eqref{eq-defofc-with-gamma} we have
\begin{align}\label{eq-nu-dnplusnminusbovern}
\nu = \frac{\min_{A_{i,j} > 0} \left( A_{i,j} / r_i(A) \right)}{\max_{i,j} \left( A_{i,j} / r_i(A) \right)} = \frac{d/(db + n-b)}{1/(db + n-b)} = d \leq 2^{-n/\varepsilon}.
\end{align}
Thus, we have $-\log \nu(A) = \Omega(n/\varepsilon)$.
Define
\begin{align}\label{def-k-alpha-nu-gamma}
D \triangleq \log\frac{(1-\gamma)\gamma'}{\nu\varepsilon},\qquad
K \triangleq \frac{c_0D}{4} = \Omega(\log (1-\gamma) + \log \gamma'-\log \nu - \log \varepsilon).
\end{align}
By \eqref{eq-defofc-with-gamma} and $\varepsilon<1/2$, we have $D>2H+2>0$. The constant hidden in $\Omega$ depends only on $\gamma,\gamma'$. 
Let $A^{(0)}, A^{(1)},\dots$ denote the sequence of matrices generated by SK with input $(A,(\boldsymbol{u}, \boldsymbol{v}))$. 
Furthermore, given any positive integer $L$, assume
\begin{align*}
&\=u' = (u'_1,\dots,u'_m) \triangleq f_1(\boldsymbol{u}, \boldsymbol{v}, L) ,\quad  \=v' = (v'_1,\dots,v'_n) \triangleq f_2(\boldsymbol{u}, \boldsymbol{v}, L),\\
&N \triangleq \norm{\=v'}_1,  \quad \; \; R \triangleq R(\boldsymbol{u}, \boldsymbol{v}, L), \quad \; \; t \triangleq \sum_{i\leq a} u'_i,\quad \; \; s \triangleq \sum_{j \leq b}v'_j,\\
&\+D(\=u') = \diag(U_1,\dots,U_N) , \quad \quad \quad   \quad \; \; \+D(\=v') = \diag(V_1,\dots,V_N).
\end{align*}
Therefore, combining Definitions \ref{def-integer-vector} and \ref{def-splitted-matrix} with \eqref{eq-def-u-v-gamma-gammaprime}, it follows that
\begin{align*}
\lim_{L\rightarrow \infty}\frac{N}{L} = 1,\quad\quad\quad \lim_{L\rightarrow \infty}\frac{s}{L} = 1- \gamma, \quad\quad\quad \lim_{L\rightarrow \infty}\frac{t}{L} = \gamma'.
\end{align*}
Hence, we have
\begin{align*}
\lim_{L\rightarrow \infty}\frac{s-t}{L} = 1 - \gamma -\gamma',\quad \lim_{L\rightarrow \infty}\frac{N L}{t(N-s)} = \frac{1}{\gamma'\gamma},\quad \lim_{L\rightarrow \infty}\frac{s-t}{N - t+\abs{s+t-N}} = \frac{(1-\gamma-\gamma')}{1-\gamma'+|\gamma'-\gamma|}.
\end{align*}
Thus, by $2\varepsilon< 1- \gamma-\gamma'$, one can choose a sufficiently large $L$ such that 
\begin{equation}\label{eq-condition-large-L}
\begin{aligned}
\quad \quad \quad  R& > 2,  \quad \quad \ \quad  \quad \quad  \quad \quad  \quad \quad \quad\quad\quad \quad 2\varepsilon < \frac{s - t}{L}, \\
\frac{N L}{t(N-s)} &\leq \frac{2}{\gamma'\gamma},\quad\quad \quad \frac{(1-\gamma-\gamma')}{2\big(1-\gamma'+|\gamma'-\gamma|\big)}  \leq \frac{s-t}{N - t+\abs{s+t-N}},\\
\frac{t}{N}&\geq\frac{\gamma'}{2},\qquad \frac{s}{N}\leq 1-\frac{\gamma}{2},\qquad \frac{s-t}{N}\geq\frac{\delta}{2}.
\end{aligned}
\end{equation}
Let $C = G(A, f_1(\boldsymbol{u}, \boldsymbol{v}, L), f_2(\boldsymbol{u}, \boldsymbol{v}, L))$.
Also let $C^{(0)}, C^{(1)},\dots$ denote the sequence of matrices generated by SK with input $(C,(\boldsymbol{1}, \boldsymbol{1}))$.
By \cref{thm-reduction-precision}, there exists a sufficiently large $\ell$ such that for every $L > \ell$ satisfying \eqref{eq-condition-large-L} and any $k \leq K$,
\begin{align}\label{eq-reduction-deviation-error}
\abs{\norm{\=r\left(A^{(k)}\right)- \boldsymbol{u}}_{1} + \norm{\=c\left(A^{(k)}\right)- \boldsymbol{v}}_{1} - \frac{\norm{\=r\left(C^{(k)}\right)- \boldsymbol{1}}_{1} + \norm{\=c\left(C^{(k)}\right)- \boldsymbol{1}}_{1}}{L}}\leq \varepsilon.
\end{align}
Furthermore, we claim for each $L$ satisfying \eqref{eq-condition-large-L} and each $k\leq K$,
\begin{align}\label{eq-claim-error-bound-l}
\norm{\=r\left(C^{(k)}\right)- \boldsymbol{1}}_{1} + \norm{\=c\left(C^{(k)}\right)- \boldsymbol{1}}_{1} > 2 L\varepsilon.
\end{align}
Combined with \eqref{eq-reduction-deviation-error}, we have
\begin{align}
\norm{\=r\left(A^{(k)}\right)- \boldsymbol{u}}_{1} + \norm{\=c\left(A^{(k)}\right)- \boldsymbol{v}}_{1} > \varepsilon.
\end{align}
for each $k\leq K$.
Combined with \eqref{def-k-alpha-nu-gamma}, the theorem is immediate.
In the following, we establish the claim to complete the proof.

By \cref{def-splitted-matrix} and $N=\norm{\=v'}_1$, we have $C$ is a matrix of size $N\times N$.
Let
\begin{align*}
x \triangleq C^{(2)}_{1,1},\quad   y \triangleq C^{(2)}_{1,s+1},   \quad \, z \triangleq C^{(2)}_{t+1,1},\quad  q \triangleq C^{(2)}_{t+1,s+1}.
\end{align*}
Also by \cref{def-splitted-matrix}, one can verify that $C$ satisfies
\begin{align*}
\forall i\leq t, j \leq s, \quad C_{i,j} &= \frac{1}{U_i\cdot V_j},\\
\forall i\leq t, j>s, \quad C_{i,j} &=  \frac{1}{U_i\cdot V_j},\\
                                    \forall i>t, j \leq s, \quad C_{i,j} &= \frac{d}{U_i\cdot V_j},\\
\forall i>t, j>s, \quad C_{i,j} &= \frac{1}{U_i\cdot V_j}.
\end{align*}
Combined with \cref{lem-lowerbound-reduction-rctooneone}, we have
\begin{align}\label{eq-d2ij-block}
\setlength{\jot}{4pt} 
C^{(2)}_{i,j}=
\begin{cases}
x, & i\le t,\ j\le s,\\[4pt]
y, & i\le t,\ j> s,\\[4pt]
z, & i> t,\ j\le s,\\[4pt]
q, & i> t,\ j> s.
\end{cases}
\end{align}
Let 
\begin{align}
S_1 \triangleq \sum_{j\leq s}\frac{1}{V_j},\quad S_2 \triangleq \sum_{s< j\leq N}\frac{1}{V_j},\quad \lambda \triangleq  \frac{S_1+S_2}{d S_1+S_2}. 
\end{align}
Combined with $s = \sum_{j \leq b}v'_j$, $\+D(\=v') = \diag(V_1,\dots,V_N)$ and \eqref{eq-def-diag},
we have 
\begin{align}\label{eq-s1eqy-s2equalnminusy}
S_1 = \sum_{j\leq s}\frac{1}{V_j} = \sum_{j\leq b}\frac{v'_j}{v'_j} = b, \quad S_2 = \sum_{s< j\leq N}\frac{1}{V_j} = \sum_{b< j\leq n} \frac{v'_j}{v'_j} = n - b, \quad \lambda = \frac{n}{d b+n - b}.
\end{align}
Combined with \eqref{eq-nu-dnplusnminusbovern} and $1\leq b \leq n$,
we have 
\begin{align}\label{eq-lambda-nu-relation}
d\lambda \leq d n = n\nu.
\end{align}
In addition, by \eqref{eq-defofc-with-gamma}, \eqref{eq-condition-large-L} and \eqref{eq-s1eqy-s2equalnminusy}, we have 
\begin{align*}
d \leq \frac{(1-\gamma-\gamma')(n-b)}{12n\big(1-\gamma'+|\gamma'-\gamma|\big)} \leq \frac{(s-t)(n-b)}{6n\big(N - t+\abs{s+t-N}\big)} = \frac{(s-t)S_2}{6(S_1+S_2)\big(N-t+|s+t-N|\big)}.
\end{align*}
Hence, by \cref{lem-lowerbound-reduction-rctooneone-new} we have
\begin{equation}\label{eq-yz-sum-bound}
\begin{aligned}
& y<  
\frac{6N}
{6N^2 + 5s(N-t)}, \\
&z < \frac{dN\lambda}{t(N-s)}, \\
\forall j\leq s,\quad &  \frac{t-N}{N}<\left(\sum_{i \leq N}C^{(2)}_{i,j}\right) - 1 <0,\\
\forall j> s,\quad & \frac{st(s-t)}
{4N^3}<\left(\sum_{i \leq N}C^{(2)}_{i,j}\right) - 1 <  \frac{s(N-t)}{N(N-s)}. 
\end{aligned}
\end{equation}
Combined with \eqref{eq-lambda-nu-relation} and \eqref{eq-condition-large-L}, we have
\begin{align}\label{eq-z-nu-bound}
z < \frac{d\lambda N}{t(N-s)} \leq \frac{2 n \nu}{\gamma'\gamma L}.
\end{align}
In addition, by \eqref{eq-condition-large-L} we have
\begin{align}\label{eq-epsilon-bound}
2L\varepsilon < s - t.
\end{align}
Combining \eqref{eq-d2ij-block}, \eqref{eq-yz-sum-bound}, \eqref{eq-z-nu-bound}, \eqref{eq-epsilon-bound}, and the block-proportion bounds in \eqref{eq-condition-large-L} with \cref{thm-main-lb-rc-ab}, let $J$ be the first index attaining error at most $2L\varepsilon$ with $\left(C^{(2)},(\=1,\=1)\right)$ as input. Then
\begin{align*}
J&\geq c_0\left(-\log z-\log(2L\varepsilon)\right)-b_0\\
&\geq c_0\left(D-\log\frac{4n(1-\gamma)}{\gamma}\right)-b_0\\
&\geq c_0(D-H)>\frac{c_0D}{2}=2K.
\end{align*}
In particular, $C^{(2)}$ has error greater than $2L\varepsilon$. Since the total marginal $\ell_1$ error is nonincreasing, the same holds for $C^{(0)}$ and $C^{(1)}$. For every integer $2\leq k\leq K$, the matrix $C^{(k)}$ is the $(k-2)$-th iterate starting from $C^{(2)}$, and $k-2<J$.
This establishes \eqref{eq-claim-error-bound-l} and concludes the proof of the theorem.
\end{proof}

%% file: Tightness.tex
\section{On the tightness of the results.}\label{sec-tightness}
In this section, we prove Theorems \ref{thm-tightness-dense-complexity}, \ref{thm-lower-bound-main-rc-tight}, and \ref{thm-lower-bound-main-rc-threshould-tight}.

\subsection{Tight iteration complexity for dense matrices.}
In this subsection, we prove \cref{thm-tightness-dense-complexity}.

The following theorem constructs a pair of vectors $(\=u,\=v)$ and a
$\left(\frac{3}{5}, \frac{1}{2}\right)$-dense, \((\boldsymbol{u}, \boldsymbol{v})\)-scalable matrix $A$, such that with $(A,\=u,\=v)$ as input, 
SK takes $\Omega(\log n - \log \varepsilon)$ iterations to output a matrix with error less than $\varepsilon$.
Hence, \cref{thm-tightness-dense-complexity} is immediate.

The construction also yields a matching lower bound for the finite, strictly positive cost matrices in \cref{thm-upper-bound-general-ot-uv}, after the following perturbation.
For $0<\delta<1$, obtain $A_\delta$ by replacing each $1$ in $A$ by $e^{-1}$ and each $0$ by $e^{-1}\delta$.
Take $\eta=1$ and $C_\delta=-\log A_\delta$, entrywise. The entries of $C_\delta$ are $1$ and $1-\log\delta$, so they are finite and strictly positive. The minimum row and column bulk capacities at $\tau=1$ are $3/5$ and $1/2$, respectively; hence $C_\delta$ is $(1,1/10)$-well-bounded, with both parameters independent of $n$ and $\varepsilon$.
For $0<\varepsilon<1/20$, apply \cref{thm-tightness-lognterm} with accuracy $2\varepsilon$. It gives a prefix of $\Omega(\log n-\log\varepsilon)$ iterations whose errors exceed $2\varepsilon$.
Every finite sequence of SK iterates depends continuously on the input near $e^{-1}A$, because all row and column sums along this sequence are positive. Moreover, multiplying the input by $e^{-1}$ does not change the iterates after the initial row normalization.
Consequently, for each $n,\varepsilon$, choosing $\delta>0$ sufficiently small preserves error greater than $\varepsilon$ throughout this prefix. Thus the iteration complexity in \cref{thm-upper-bound-general-ot-uv} is tight even for finite, strictly positive costs.

\begin{theorem}\label{thm-tightness-lognterm}
Let $n$ be a positive integer multiple of $10$. 
Assume $\varepsilon \in (0,1/10)$.
Define
\begin{align}
\boldsymbol{u} \triangleq \Bigl(
\underbrace{\frac{1}{n},\ldots,\frac{1}{n}}_{n\text{ entries}}
\Bigr),\quad 
\boldsymbol{v} \triangleq \Bigl(
\underbrace{\frac{1}{n},\ldots,\frac{1}{n}}_{2n/5\text{ entries}},
\frac{1}{5}, \frac{2}{5}
\Bigr).
\end{align}
Thus, we have
$\norm{\boldsymbol{u}}_1 = \norm{\boldsymbol{v}}_1 = 1$.
Let $A$ be a nonnegative matrix of size $n\times (2+2n/5)$ where $A_{i,j} = 0$ if 
\begin{align}
i \leq \frac{n}{2}, j\leq \frac{2n}{5} \quad \text{ or } \quad i > \frac{n}{2}, j= 2 + \frac{2n}{5}.
\end{align}
Otherwise, $A_{i,j} = 1$.
With $(A,(\=u,\=v))$ as input,
SK takes $\Omega(\log n - \log \varepsilon)$ iterations to output a matrix $B$ satisfying 
\[\norm{\=r\left(B\right)- \boldsymbol{u}}_{1} + \norm{\=c\left(B\right)- \boldsymbol{v}}_{1} \leq \varepsilon.\]
\end{theorem}

\begin{proof}
The matrix $A$ is $(3/5,1/2,1)$-dense with respect to $(\=u,\=v)$.
It is also $(\=u,\=v)$-scalable: multiplying its first $2n/5$ columns by $2/n^2$, its next column by $1/(5n)$, and its last column by $4/(5n)$ gives the prescribed row and column sums.
Let $A^{(0)}, A^{(1)},\dots$ denote the sequence of matrices generated by SK with $(A,(\=u,\=v))$ as input.
Define 
\begin{align*}
\forall k \geq 0,\quad \varepsilon^{(k)} &\triangleq \norm{\=r\left(A^{(k)}\right)- \boldsymbol{u}}_{1} + \norm{\=c\left(A^{(k)}\right)- \boldsymbol{v}}_{1}.
\end{align*}
One can verify that $A^{(k)}$ always has the form
\begin{align*}
A^{(k)}_{i,j} = 0 \quad\quad &\text{ if } i\leq \frac{n}{2}, j\leq \frac{2n}{5};\quad\quad\quad\quad\;\, A^{(k)}_{i,j} = c^{(k)} \quad \text{ if } i> \frac{n}{2}, j\leq \frac{2n}{5};\\
A^{(k)}_{i,j} = a^{(k)} \quad &\text{ if } i\leq \frac{n}{2}, j= 1+\frac{2n}{5}; \quad\quad\quad 
A^{(k)}_{i,j} = d^{(k)} \quad \text{ if } i> \frac{n}{2}, j= 1+\frac{2n}{5};\\
A^{(k)}_{i,j} = b^{(k)} \quad &\text{ if } i\leq \frac{n}{2}, j= 2+\frac{2n}{5};\quad \quad\quad
A^{(k)}_{i,j} = 0 \quad \quad \text{ if } i> \frac{n}{2}, j= 2+ \frac{2n}{5}.
\end{align*}
Moreover, for each odd $k>0$, we have 
\begin{align*}
\forall j\leq \frac{2n}{5},\quad &\frac{1}{n}= \sum_{i\in [n]}A^{(k)}_{i,j} = \frac{n c^{(k)}}{2};\\
&\frac{1}{5}= \sum_{i\in [n]}A^{(k)}_{i,1+2n/5} = \frac{n \left(a^{(k)}+d^{(k)}\right)}{2};\\
&\frac{2}{5}= \sum_{i\in [n]}A^{(k)}_{i,2+2n/5} = \frac{n b^{(k)}}{2}.
\end{align*}
Thus, we have
\begin{align}
a^{(k)}+d^{(k)} = \frac{2}{5n}, \quad\quad b^{(k)} = \frac{4}{5n}, \quad\quad c^{(k)} = \frac{2}{n^2}.
\end{align}
For odd $k\geq 1$, define
\begin{align}\label{eq-def-thetak-relation-ak}
    \theta^{(k)} \triangleq \frac{5na^{(k)}}{2}.
\end{align}
We have 
\[
d^{(k)} = \frac{2}{5n}\left(1-\theta^{(k)}\right).
\]
Thus,
\begin{align*}
\forall i\leq \frac{n}{2},\quad &\left(\sum_{j\leq 2+2n/5}A^{(k)}_{i,j}\right) -\frac{1}{n} = a^{(k)}+b^{(k)} - \frac{1}{n} = \frac{2\theta^{(k)}}{5n} - \frac{1}{5n};\\
\forall i> \frac{n}{2},\quad &\left(\sum_{j\leq 2+2n/5}A^{(k)}_{i,j}\right) -\frac{1}{n} =\frac{2nc^{(k)}}{5}+d^{(k)} - \frac{1}{n} =  \frac{1}{5n} - \frac{2\theta^{(k)}}{5n}.
\end{align*}
Hence,
\[\varepsilon^{(k)} = n\cdot \abs{\frac{2\theta^{(k)}}{5n} - \frac{1}{5n}} =\frac{1}{5} \abs{2\theta^{(k)} - 1}.\]
Moreover, we claim that
\begin{align}\label{eq-thetakplusone-thetak-reduction}
\theta^{(k+2)} = \frac{\theta^{(k)}\left(3-\theta^{(k)}\right)}{2\left(1+\theta^{(k)} - \left(\theta^{(k)}\right)^2\right)}.
\end{align}
Define 
\[\omega^{(k)} = \frac{2\theta^{(k)} - 1}{1 - \theta^{(k)}}.\]
We have 
\begin{align}\label{eq-thetak-omegak-varepsilonk-omegak}
\theta^{(k)} = \frac{\omega^{(k)}+1}{\omega^{(k)}+2}, \quad \quad \varepsilon^{(k)} =\frac{1}{5} \abs{2\theta^{(k)} - 1} = \frac{1}{5} \abs{\frac{\omega^{(k)}}{\omega^{(k)}+2}}.
\end{align}
Combined with \eqref{eq-thetakplusone-thetak-reduction}, 
we have
\begin{align}
\theta^{(k+2)} = \frac{\left(2\left(\omega^{(k)}\right)^2 + 7\omega^{(k)} + 5\right)\left(2+\omega^{(k)}\right)^{-2}}{2\left(\left(\omega^{(k)}\right)^2 + 5\omega^{(k)} + 5\right)\left(2+\omega^{(k)}\right)^{-2}} =\frac{2\left(\omega^{(k)}\right)^2 + 7\omega^{(k)} + 5}{2\left(\left(\omega^{(k)}\right)^2 + 5\omega^{(k)} + 5\right)} .
\end{align}
Thus,
\begin{align*}
\omega^{(k+2)} = \frac{2\theta^{(k+2)} -1}{1 - \theta^{(k+2)}} =\frac{2\omega^{(k)}\left(2 + \omega^{(k)}\right)}{5 + 3\omega^{(k)}} > \frac{2\omega^{(k)}}{3}.
\end{align*}
If $\omega^{(k)}>0$,
we have 
\begin{align}\label{eq-induction-omegarelation}
\omega^{(k+2)} > \frac{2\omega^{(k)}}{3}.
\end{align}
Moreover, one can verify that 
\begin{align*}
a^{(1)} = \frac{2(2n+5)}{5n(2n+15)}. 
\end{align*}
Hence, we have
\begin{align*}
\theta^{(1)} = \frac{2n+5}{2n+15},\quad \omega^{(1)} = \frac{2n-5}{10}. 
\end{align*}
Combined with \eqref{eq-induction-omegarelation},
we have 
\begin{align}
\forall \text{ odd } k\leq 2\log_{3/2} \frac{2n-5}{200\varepsilon}, \quad \omega^{(k)} > 20\varepsilon.
\end{align}
Combined with \eqref{eq-thetak-omegak-varepsilonk-omegak} and $\varepsilon < 1/10$, we have
\begin{align}
\forall \text{ odd } k\leq 2\log_{3/2} \frac{2n-5}{200\varepsilon}, \quad \varepsilon^{(k)} =\frac{1}{5} \abs{\frac{\omega^{(k)}}{\omega^{(k)}+2}} > \varepsilon.
\end{align}
For even $k$, column normalization and the triangle inequality give
\[
\varepsilon^{(k+1)}
=\norm{\=r(A^{(k+1)})-\=r(A^{(k)})}_1
\leq \sum_{i,j}\abs{A^{(k+1)}_{i,j}-A^{(k)}_{i,j}}
=\norm{\=c(A^{(k)})-\=v}_1=\varepsilon^{(k)}.
\]
Hence the same error lower bound holds for even $k\leq 2\log_{3/2}((2n-5)/(200\varepsilon))-1$.
In summary, we have SK takes $k=\Omega(\log n - \log \varepsilon)$ iterations to output a matrix $A^{(k)}$ with
$\varepsilon^{(k)}\leq \varepsilon$.

In the following, we prove \eqref{eq-thetakplusone-thetak-reduction} for odd $k\geq 1$. Then the theorem is immediate.
We have
\begin{align*}
a^{(k+1)} &= \frac{a^{(k)}}{n\cdot r_1\left(A^{(k)}\right)} =\frac{a^{(k)}}{n(a^{(k)}+b^{(k)})} =  \frac{a^{(k)}}{n}\cdot \left(\frac{2\theta^{(k)}}{5n} + \frac{4}{5n}\right)^{-1} = \frac{5a^{(k)}}{2(\theta^{(k)}+2)}.\\
d^{(k+1)} &=\frac{d^{(k)}}{n\cdot r_n\left(A^{(k)}\right)} = \frac{d^{(k)}}{n\left(2nc^{(k)}/5+d^{(k)} \right)} = \frac{d^{(k)}}{n}\cdot \left(\frac{6}{5n} - \frac{2\theta^{(k)}}{5n}\right)^{-1} = \frac{5d^{(k)}}{6 - 2\theta^{(k)}}.
\end{align*}
Hence, we have
\begin{align*}
\frac{n}{2}\cdot \left(a^{(k+1)} + d^{(k+1)} \right)&= \frac{5na^{(k)}}{4(\theta^{(k)}+2)} + \frac{5nd^{(k)}}{2\left(6 - 2\theta^{(k)}\right)}= \frac{\theta^{(k)}}{2(\theta^{(k)}+2)} + \frac{1-\theta^{(k)}}{6 - 2\theta^{(k)}}.
\end{align*}
Therefore,
\begin{align*}
a^{(k+2)} &= \frac{a^{(k+1)}}{5\cdot c_{1+2n/5}\left(A^{(k+1)}\right)} = \frac{a^{(k+1)}}{5}\cdot \left(\frac{na^{(k+1)}}{2} + \frac{nd^{(k+1)}}{2}\right)^{-1} = \frac{a^{(k+1)}}{5}\cdot \left(\frac{\theta^{(k)}}{2(\theta^{(k)}+2)} + \frac{1-\theta^{(k)}}{6 - 2\theta^{(k)}}\right)^{-1} \\
&= \frac{a^{(k)}}{2(\theta^{(k)}+2)}\cdot \left(\frac{\theta^{(k)}}{2(\theta^{(k)}+2)} + \frac{1-\theta^{(k)}}{6 - 2\theta^{(k)}}\right)^{-1}. 
\end{align*}
Combined with \eqref{eq-def-thetak-relation-ak}, we have
\begin{align}
\theta^{(k+2)} = \frac{\theta^{(k)}}{2(\theta^{(k)}+2)}\cdot \left(\frac{\theta^{(k)}}{2(\theta^{(k)}+2)} + \frac{1-\theta^{(k)}}{6 - 2\theta^{(k)}}\right)^{-1} = \frac{\theta^{(k)}\left(3-\theta^{(k)}\right)}{2\left(1+\theta^{(k)} - \left(\theta^{(k)}\right)^2\right)}. 
\end{align}
Thus, \eqref{eq-thetakplusone-thetak-reduction} holds, which concludes the proof.
\end{proof}

\subsection{Tight iteration complexity for sparse matrices.}
In this subsection, we prove \cref{thm-lower-bound-main-rc-tight}.

Given $\=u= \left(\frac{5}{6},\frac{1}{6}\right),\=v= \left(\frac{7}{8},\frac{1}{8}\right)$, one can verify the existence of feasible $(\gamma, \gamma')$ with respect to $(\={u}, \={v})$ with $\gamma + \gamma' < 1$. Consequently, \cref{thm-lower-bound-main-rc-tight} follows as an immediate corollary of the result below.

\begin{theorem}\label{thm-tight-rc-scaling}
Without loss of generality, assume $\varepsilon \in (0, 1/100)$.
Let 
\begin{equation*}
\begin{aligned}
\=u= (u_1,u_2)\triangleq\left(\frac{5}{6},\frac{1}{6}\right),\quad \quad 
\=v=(v_1,v_2)\triangleq \left(\frac{7}{8},\frac{1}{8}\right).
\end{aligned}
\end{equation*}
Given any nonnegative matrix $A$ of size $2\times 2$ which is $(\=u,\=v)$-scalable, 
let $A^{(0)}, A^{(1)},\dots$ denote the sequence of matrices generated by SK with $(A,(\=u,\=v))$ as input.
Let $K$ be the minimum integer $k$ such that 
\begin{align*}
\norm{\=r\left(A^{(k)}\right)- \=u}_{1} + \norm{\=c\left(A^{(k)}\right)- \=v}_{1} \leq \varepsilon.
\end{align*}
Then we have 
$K =O\left(-\log \nu (A) - \log \varepsilon\right)$.
\end{theorem}

Our analysis relies on a two-phase framework that leverages a key structural observation: 
by exploiting the engineered asymmetry of the target marginal distributions, the intermediate scaling matrices are guaranteed to enter a dense regime. This denseness, in turn, triggers local exponential convergence. This property allows us to decouple the algorithm's trajectory into two distinct phases, bounded as follows:

\begin{enumerate}[(1)]
\item \emph{\underline{Structural Property and Local Denseness.}}
First, we construct specific asymmetric target probability vectors $\boldsymbol{u} = (u, 1 - u)$ and $\boldsymbol{v} = (v, 1-v)$ satisfying $v > u > 1/2$ (as instantiated in \cref{thm-tight-rc-scaling} with $u=5/6$ and $v=7/8$). The purpose of this construction is to guarantee a crucial structural property: any non-negative matrix whose row and column sums are sufficiently close to $\boldsymbol{u}$ and $\boldsymbol{v}$ must be strictly dense.
Specifically, assume the marginal errors of the current matrix are extremely small. Since the sum of the second column is approximately $1-v$, it necessarily follows that $A_{1,2} \lesssim 1-v$. To satisfy the condition that the first row sum approaches $u$, $A_{1,1}$ must have a strictly positive lower bound: $A_{1,1} \gtrsim u - (1-v) = u + v - 1 > 0$. Through similar deductions, we obtain $A_{2,1} \gtrsim v - u > 0$ and $A_{1,2} + A_{2,2} \approx 1-v$. Consequently, once the marginal errors fall below a specific constant threshold, the matrix is guaranteed to be dense. The strictly positive bounds on $A_{1,1}$ and $A_{2,1}$, combined with the strictly positive sum $A_{1,2} + A_{2,2}$, preclude the matrix from degrading into a sparse configuration.

\item \emph{\underline{Phase I: Global Convergence via the Permanent.}}
During the initial iterations of the SK algorithm, the marginal errors need not be small. To analyze the algorithm's progress during this phase, we reduce the $(\boldsymbol{u}, \boldsymbol{v})$-scaling problem to a standard $(\boldsymbol{1}, \boldsymbol{1})$-scaling problem. 
Through this reduction, we demonstrate that the permanent of the reduced matrix grows by a constant factor in each iteration whose error exceeds the fixed threshold. 
Because the initial permanent is lower bounded by $\poly(\nu(A))$ and the permanent of each row- or column-stochastic iterate is at most $1$, this large-error phase is guaranteed to terminate in at most $O(1-\log \nu(A))$ iterations.

\item \emph{\underline{Phase II: Local Exponential Decay.}}
Once the algorithm advances past the initial $O(1-\log \nu(A))$ iterations, the error drops below the aforementioned constant threshold. 
At this stage, the inherent imbalance of the marginal distributions dictates the matrix dynamics, ensuring the intermediate matrices $A^{(k)}$ remain in a dense regime.
Because the SK algorithm exhibits rapid linear convergence on dense matrices, this structural property immediately yields a local exponential decay of the error. 
Therefore, achieving the final $\varepsilon$ accuracy from this point requires only a logarithmic number of additional steps, bounded by $O(-\log \varepsilon)$.
\end{enumerate}

Combining the iteration bounds of these two phases yields the convergence time of $O(-\log \nu(A) - \log \varepsilon)$, since $\varepsilon<1/100$.

The following lemma is used in the proof of \cref{thm-tight-rc-scaling}.

\begin{lemma}\label{lem-l-lognu}
Assume the conditions of \cref{thm-tight-rc-scaling}.
Let $L$ be the minimum integer $k$ such that 
\begin{align*}
\norm{\=r\left(A^{(k)}\right)- \=u}_{1} + \norm{\=c\left(A^{(k)}\right)- \=v}_{1} \leq \frac{1}{1000}.
\end{align*}
Then we have $L = O(1-\log \nu (A))$.
\end{lemma}

\begin{proof}
Let $t = 1/1000$, $\nu \triangleq \nu(A)$, and define the constant 
$$\alpha = \max\left\{\left(1-\frac{4t}{7}\right)^{42}(1+4t)^6, \left(1-\frac{3t}{5}\right)^{40}(1+3t)^8\right\}.$$ 
One can verify $0<\alpha<1$. 
We then define $S$ as the set of iteration indices satisfying
\begin{align}\label{eq-def-k-counterexample}
S = \left\{k\geq 0 \;\middle|\; \norm{\=r\left(A^{(k)}\right)- \=u}_{1} + \norm{\=c\left(A^{(k)}\right)- \=v}_{1} \geq t\right\}.
\end{align}
Since \(A\) is \((\boldsymbol u,\boldsymbol v)\)-scalable, SK converges and \(S\) is finite.
Combined with $t = 1/1000$,
we have $L \leq \abs{S}+1$.
Thus, to prove \cref{lem-l-lognu}, it is sufficient to prove 
\[\abs{S} \leq 
\frac{50\cdot(\log 336 - \log \nu)}{-\log \alpha} = O(1-\log \nu).\]
Assume for contradiction that 
\begin{align}\label{eq-lb-K}
\abs{S}> 
\frac{50\cdot(\log 336 - \log \nu)}{-\log \alpha} = -\log_{\alpha} \left(\frac{336}{\nu}\right)^{50}.
\end{align}

Define 
\[
\=a = (a_1,a_2) \triangleq 48\cdot \=u = (40,8),\quad \=b = (b_1,b_2) \triangleq 48\cdot \=v = (42,6).
\]
Let $B$ be a matrix of size $48 \times 48$ where 
\begin{align*}
\forall i\leq 40,j\leq 42,\quad \quad B_{i,j} \triangleq \frac{A_{1,1}}{a_1b_1}; \quad \quad &\quad \quad \forall i\leq 40, 42<j\leq 48,\quad \quad B_{i,j} \triangleq \frac{A_{1,2}}{a_1b_2};\\
\forall 40<i\leq 48, j\leq 42,\quad \quad B_{i,j} \triangleq \frac{A_{2,1}}{a_2b_1}; \quad \quad &
\forall 40 <i\leq 48, 42<j\leq 48,\quad \quad B_{i,j} \triangleq \frac{A_{2,2}}{a_2b_2}.
\end{align*}
Hence,
\begin{align*}
\forall i\leq 40, \quad &r_i\left(B\right) = \frac{r_1(A)}{a_1},\\
\forall  40<i\leq 48,\quad  &r_i\left(B\right) = \frac{r_2(A)}{a_2}.
\end{align*}
We have 
\begin{equation}\label{eq-minaijlargerthanzero-maxinaijoverria}
\begin{aligned}
\min_{A_{i,j} > 0} \left( A_{i,j} / r_i(A) \right) &\leq \min_{B_{i,j} > 0} \left( B_{i,j} / r_i(B) \right)\cdot \max\{b_1,b_2\} = 42\min_{B_{i,j} > 0} \left( B_{i,j} / r_i(B) \right).\\
\max_{i,j} \left( A_{i,j} / r_i(A) \right)&\geq \max_{i,j} \left( B_{i,j} / r_i(B) \right)\cdot \min\{b_1,b_2\} = 6 \max_{i,j} \left( B_{i,j} / r_i(B) \right).
\end{aligned}    
\end{equation}
Let $B^{(0)}, B^{(1)},\dots$ denote the sequence of matrices generated by SK with $(B,(\=1,\=1))$ as input.
Define $\tau_{-}$ and $\tau_{+}$ as the minimum nonzero elements and the maximum elements in $B^{(0)}$, respectively.
By \eqref{eq-minaijlargerthanzero-maxinaijoverria}, we have
\begin{align}\label{eq-relation-nu-tau}
\nu = \frac{\min_{A_{i,j} > 0} \left( A_{i,j} / r_i(A) \right)}{\max_{i,j} \left( A_{i,j} / r_i(A) \right)} \leq \frac{42\min_{B_{i,j} > 0} \left( B_{i,j} / r_i(B) \right)}{6 \max_{i,j} \left( B_{i,j} / r_i(B) \right)} =  \frac{7\tau_{-}}{\tau_{+}}. 
\end{align}
Moreover, 
by \cref{thm-reduction-precision} we have
\begin{align}\label{eq-a-b-identicalerror}
\forall k\geq 0,\quad \norm{\=r\left(A^{(k)}\right)- \boldsymbol{u}}_{1} + \norm{\=c\left(A^{(k)}\right)- \boldsymbol{v}}_{1} = \frac{1}{48}\left(\norm{\=r\left(B^{(k)}\right)- \boldsymbol{1}}_{1} + \norm{\=c\left(B^{(k)}\right)- \boldsymbol{1}}_{1}\right).
\end{align}
Therefore, by \eqref{eq-def-k-counterexample} we have 
\begin{align}\label{eq-condition-thm-main-lb-rc-ab-b}
S = \left\{k\geq 0 \mid \norm{\=r\left(B^{(k)}\right)- \=1}_{1} + \norm{\=c\left(B^{(k)}\right)- \=1}_{1} \geq 48t\right\}.
\end{align}
Furthermore, one can verify that each matrix $B^{(k)}$ where $k\geq 0$ can be partitioned into $2\times 2$ blocks, with all entries within each block being identical. 
For each $i,j\leq 2$,
the $(i,j)$-block of $B^{(k)}$ has size \(a_i\times b_j\), and all its entries are denoted by \(\theta^{(k)}_{i,j}\).
Thus, we have 
\begin{align}
\forall i,j\leq 40 \text{ or } 40<i,j\leq 48, \quad &r_i\left(B^{(k)}\right) = r_{j}\left(B^{(k)}\right),\label{eq-rowsum-eq}\\
\forall i,j\leq 42 \text{ or } 42<i,j\leq 48, \quad &c_i\left(B^{(k)}\right) = c_{j}\left(B^{(k)}\right).\label{eq-columnsum-eq}
\end{align}

We claim that for each even $k\in S$, 
\begin{align}\label{eq-up-prodcolumnsum-counterexample}
\prod_{j\leq 48}c_j\left(B^{(k)}\right)\leq \left(1-\frac{4t}{7}\right)^{42}(1+4t)^6.
\end{align}
Similarly, for each odd $k\in S$, 
\begin{align}\label{eq-up-prodrowsum-counterexample}
\prod_{i\leq 48}r_i\left(B^{(k)}\right)\leq \left(1-\frac{3t}{5}\right)^{40}(1+3t)^8.
\end{align}

Let $\ell$ be the maximum number in $S$. 
By \eqref{eq-increase-perm-ri} and \eqref{eq-increase-perm-ci}, we have 
\begin{align*}
\perman\left(B^{(\ell+1)}\right) &= \perman\left(B^{(0)}\right)\left(\prod_{k = 0 }^{\lfloor\ell /2\rfloor}\prod_{i\leq 48} c_i^{-1}\left(B^{(2k)}\right)\right)\left(\prod_{k = 0}^{\lfloor(\ell -1) /2\rfloor}\prod_{i\leq 48} r_i^{-1}\left(B^{(2k+1)}\right)\right).
\end{align*}
Combined with \eqref{eq-ub-ri} and \eqref{eq-ub-ci},
we have 
\begin{align*}
\perman\left(B^{(\ell+1)}\right) &\geq \perman\left(B^{(0)}\right)\left(\prod_{k:2k\in S}\prod_{i\leq 48} c_i^{-1}\left(B^{(2k)}\right)\right)\left(\prod_{k:2k+1\in S}\prod_{i\leq 48} r_i^{-1}\left(B^{(2k+1)}\right)\right).
\end{align*}
Combined with \eqref{eq-up-prodcolumnsum-counterexample} and \eqref{eq-up-prodrowsum-counterexample},
we have 
\begin{align}\label{eq-perm-alplusone}
\perman\left(B^{(\ell+1)}\right) &\geq \perman\left(B^{(0)}\right)\cdot \alpha^{-
\abs{S}}.
\end{align}

Moreover, by $A$ is $(\=u,\=v)$-scalable and \eqref{eq-a-b-identicalerror}, we have 
$B$ is $(\=1,\=1)$-scalable.
Thus, we have $\perman(B)>0$.
Otherwise, $\perman(B)= 0$.
By the definition of $(\=1,\=1)$-scalable, we have 
there are positive diagonal \(X,Y\) such that \(XBY\) is doubly stochastic.
Combined with $\perman(B)= 0$, we have
$\perman(XBY) = 0$,
which is contradictory with \cref{lem-perm-double-stochastic} and that \(XBY\) is doubly stochastic.
Since $\perman(B)>0$, 
it immediately follows that there exist $48$ non-zero entries in $B$, no two of which share a row or a column.
Since $B^{(0)}$ is obtained from $B$ by row normalization, there also exist $48$ nonzero entries in $B^{(0)}$, no two of which share a row or a column. Recall that $\tau_{-}$ and $\tau_{+}$ denote the minimum nonzero entry and the maximum entry of $B^{(0)}$, respectively. Hence, $B^{(0)}$ contains $48$ entries, each at least $\tau_{-}$, with no two sharing a row or a column. Moreover, since every row sum of $B^{(0)}$ is equal to $1$ and $\tau_{+}$ is the maximum entry of $B^{(0)}$, we have $\tau_{+}\ge 1/48$.
Thus,
\[\perman\left(B^{(0)}\right) \geq \tau_{-}^{48} \geq \left(\frac{\tau_{-}}{48\tau_{+}}\right)^{48} \geq  \left(\frac{\nu}{336}\right)^{48} ,\]
where the last inequality is by \eqref{eq-relation-nu-tau}.
Combined with \eqref{eq-lb-K} and \eqref{eq-perm-alplusone},
we have 
\begin{align*}
\perman\left(B^{(\ell+1)}\right) >\left(\frac{\nu}{336}\right)^{48} \cdot \left(\frac{336}{\nu}\right)^{50} > 1.
\end{align*}
However, by \cref{fact-c-a-a0-a1} we have either $c_i\left(B^{(\ell+1)}\right) = 1$ for each $i\in [48]$, or $r_i\left(B^{(\ell+1)}\right) = 1$ for each $i\in [48]$.
Hence, we have $\perman\left(B^{(\ell+1)}\right) \leq 1$, a contradiction.
Thus, we have 
\begin{align*}
\abs{S}\leq 
\frac{50\cdot(\log 336 - \log \nu)}{-\log \alpha} = O(1-\log \nu) .
\end{align*}

Finally, we establish \eqref{eq-up-prodcolumnsum-counterexample} and \eqref{eq-up-prodrowsum-counterexample}, which completes the proof of the lemma. We prove only \eqref{eq-up-prodcolumnsum-counterexample} here; the proof of \eqref{eq-up-prodrowsum-counterexample} is analogous.
Assume without loss of generality that $k$ is even.
By \cref{fact-c-a-a0-a1}, we have $r_i\left(B^{(k)}\right) = 1$ for each $i\leq 48$.
Combined with $k\in S$ and \eqref{eq-condition-thm-main-lb-rc-ab-b}, we have 
\begin{align*}
\norm{\=c\left(B^{(k)}\right)- \=1}_{1} = \sum_{j\leq 48}\abs{ c_j\left(B^{(k)}\right)- 1} \geq 48t.
\end{align*}
Combined with \eqref{eq-columnsum-eq},
we have 
\begin{align}\label{eq-sum-absx}
42\abs{ c_1\left(B^{(k)}\right)- 1} + 6\abs{ c_{43}\left(B^{(k)}\right)- 1} \geq 48t.
\end{align}
Let 
$x = c_1\left(B^{(k)}\right)- 1$.
In addition, by
\[ 42c_1\left(B^{(k)}\right) + 6c_{43}\left(B^{(k)}\right) = \sum_{j\leq 48}c_j\left(B^{(k)}\right) = \sum_{i \leq 48}r_i\left(B^{(k)}\right) = 48,\]
we have 
$c_{43}\left(B^{(k)}\right)- 1 = -7x$.
Combined with \eqref{eq-sum-absx},
we have $\abs{x}\geq 4t/7$.
In addition, by \eqref{eq-columnsum-eq} we have 
\begin{align*}
\prod_{j\leq 48}c_j\left(B^{(k)}\right) = \left(c_1\left(B^{(k)}\right)\right)^{42}\cdot \left(c_{43}\left(B^{(k)}\right)\right)^{6} = (1+x)^{42}\cdot (1 - 7x)^6.
\end{align*}
Let 
\begin{align}\label{eq-definition-fx-prod-columnsum}
f(x) = \ln \left(\prod_{j\leq 48}c_j\left(B^{(k)}\right) \right) = \ln \left((1+x)^{42}\cdot (1 - 7x)^6\right) = 42\ln (1+x)+ 6\ln(1-7x).
\end{align}
We have its derivative is 
\begin{align*}
f'(x) = \frac{42}{1+x} - \frac{42}{1-7x} = \frac{-336x}{(1+x)(1-7x)}.
\end{align*}
Assume $x\in (-1,1/7)$. 
We have $f'(x)>0$ if and only if $x<0$.
Thus, $f(x)$ increases on $(-1,0)$ and decreases on $(0,1/7)$.
Combined with $\abs{x}\geq 4t/7$ and $t = 1/1000$,
we have $f(x)\leq \max\{f(4t/7),f(-4t/7)\}$.
Combined with \eqref{eq-definition-fx-prod-columnsum},
we have 
\begin{align*}
\forall x\in \left(-1,\frac{1}{7}\right)\text{ with }\abs{x}\geq 4t/7,\quad \prod_{j\leq 48}c_j\left(B^{(k)}\right) &\leq \max
\left
\{\left(1+\frac{4t}{7}\right)^{42}(1-4t)^6, \left(1-\frac{4t}{7}\right)^{42}\left(1+4t\right)^6\right\}\\
& =\left(1-\frac{4t}{7}\right)^{42}(1+4t)^6.
\end{align*}
In addition, by
\begin{align*}
x = c_1\left(B^{(k)}\right)- 1,\quad c_{43}\left(B^{(k)}\right)- 1 = -7x, \quad c_1\left(B^{(k)}\right)>0, \quad c_{43}\left(B^{(k)}\right)>0,
\end{align*}
we have 
$x\in (-1,1/7)$.
Thus, we have 
\begin{align*}
\prod_{j\leq 48}c_j\left(B^{(k)}\right) \leq \left(1-\frac{4t}{7}\right)^{42}(1+4t)^6,
\end{align*}
which finishes the proof of \eqref{eq-up-prodcolumnsum-counterexample}.
The lemma is immediate.

\end{proof}

Now we can prove \cref{thm-tight-rc-scaling}.
\begin{proof}[Proof of \cref{thm-tight-rc-scaling}]
Let $t = 1/1000$.
Let $L$ be the minimum integer $k$ such that 
\begin{align}\label{eq-lem-tight-rc-scaling-error-ub-t}
\norm{\=r\left(A^{(k)}\right)- \=u}_{1} + \norm{\=c\left(A^{(k)}\right)- \=v}_{1} \leq t.
\end{align}
Thus, we have 
\[
\max_{i\leq 2 ,j\leq 2}A^{(L)}_{i,j} \leq 1+ 2t.
\]
Otherwise,
\[
\norm{\=r\left(A^{(L)}\right)- \=u}_{1} + \norm{\=c\left(A^{(L)}\right)- \=v}_{1} \geq \norm{\=r\left(A^{(L)}\right)- \=u}_{1} \geq \max_{i\leq 2 ,j\leq 2}A^{(L)}_{i,j} - 1 >2t>t,
\]
which is contradictory with \eqref{eq-lem-tight-rc-scaling-error-ub-t}.
% Thus, we have 
% \[ \max_{i\leq 2 ,j\leq 2}\frac{A_{i,j}}{u_i\cdot v_j}\leq 6\times 8 \times (1+t) = 48(1+t).\]
We further claim that
\begin{align}\label{eq-claim-elementrange}
A^{(L)}_{1,1} \geq \frac{1}{5},\quad A^{(L)}_{2,1} \geq \frac{1}{50},\quad A^{(L)}_{1,2} + A^{(L)}_{2,2} \geq 1/20.
\end{align}
By $A^{(L)}_{1,2} + A^{(L)}_{2,2} \geq 1/20$,
we have either $A^{(L)}_{1,2} \geq 1/40$ or $A^{(L)}_{2,2} \geq 1/40$.
\begin{itemize}
\item If $A^{(L)}_{1,2} \geq 1/40$, 
by $A^{(L)}_{1,1} \geq 1/5$, $A^{(L)}_{2,1} \geq 1/50$,
we have for each $(s,t) \in \{(1,1),(2,1),(1,2)\}$,
\[
A^{(L)}_{s,t} \geq \frac{1}{50}= \frac{2}{100}>\frac{1}{100}\cdot \max_{i\leq 2 ,j\leq 2}A^{(L)}_{i,j}.
\]
Combined with \cref{def-gamma-rho-dense},
we have $A^{(L)}$ is at least $\left(\frac{7}{8},\frac{5}{6},\frac{1}{100}\right)$-dense.
\item If $A^{(L)}_{2,2} \geq 1/40$, 
by $A^{(L)}_{1,1} \geq 1/5$, $A^{(L)}_{2,1} \geq 1/50$,
we have for each $(s,t) \in \{(1,1),(2,1),(2,2)\}$,
\[
A^{(L)}_{s,t} \geq \frac{1}{50}= \frac{2}{100}> \frac{1}{100}\cdot \max_{i\leq 2 ,j\leq 2}A^{(L)}_{i,j}.
\]
Combined with \cref{def-gamma-rho-dense},
we have $A^{(L)}$ is at least $\left(\frac{7}{8},\frac{1}{6},\frac{1}{100}\right)$-dense.
\end{itemize}
In both cases, by \cref{thm-rc-scaling-log} we have 
\[\norm{\=r\left(A^{(L+T)}\right)- \boldsymbol{u}}_{1} + \norm{\=c\left(A^{(L+T)}\right)- \boldsymbol{v}}_{1} \leq \varepsilon\]
for some $T = O(-\log \varepsilon)$.
Furthermore, by \cref{lem-l-lognu} we have $L = O(1-\log \nu(A))$.
By the definition of $K$ in the theorem and $\varepsilon<1/100$,
we have 
\[K \leq L + T= O(-\log \nu(A) -\log \varepsilon).\]

In the following, we prove \eqref{eq-claim-elementrange}. Then the theorem is immediate.
To prove \eqref{eq-claim-elementrange}, it is sufficient to prove $A^{(L)}_{1,1} \geq 1/5$ and $A^{(L)}_{1,2} + A^{(L)}_{2,2} \geq 1/20$.
Indeed, $A^{(L)}_{2,1}\geq c_1(A^{(L)})-r_1(A^{(L)})\geq 7/8-5/6-2t>1/50$.
At first, we prove $A^{(L)}_{1,1} \geq 1/5$.
Assume for contradiction that $A^{(L)}_{1,1} < 1/5$.
By \eqref{eq-lem-tight-rc-scaling-error-ub-t} we have
\begin{align*}
t &\geq \norm{\=r\left(A^{(L)}\right)- \=u}_{1} + \norm{\=c\left(A^{(L)}\right)- \=v}_{1} \geq r_2\left(A^{(L)}\right)- \frac{1}{6}.
\end{align*}
Thus, 
\begin{align*}
r_2\left(A^{(L)}\right)\leq \frac{1}{6}+t.
\end{align*}
Similarly, we also have 
\begin{align*}
c_2\left(A^{(L)}\right)\leq \frac{1}{8}+t.
\end{align*}
Thus, we have
\begin{align*}
\sum_{i\leq 2,j\leq 2}A^{(L)}_{i,j} &\leq A^{(L)}_{1,1}+ r_2\left(A^{(L)}\right)  + c_2\left(A^{(L)}\right) < \frac{1}{5} + \frac{1}{6}+t + \frac{1}{8}+t.
\end{align*}
Therefore,
\begin{align*}
&\quad\norm{\=r\left(A^{(L)}\right)- \=u}_{1} + \norm{\=c\left(A^{(L)}\right)- \=v}_{1} \geq \norm{\=r\left(A^{(L)}\right)- \=u}_{1} \geq u_1+u_2 - \sum_{i\leq 2}r_i\left(A^{(L)}\right)\\
&= 1 - \sum_{i\leq 2,j\leq 2}A^{(L)}_{i,j} \geq 1 - \left(\frac{1}{5} + \frac{1}{6}+t + \frac{1}{8}+t\right) >t,
\end{align*}
which is contradictory with \eqref{eq-lem-tight-rc-scaling-error-ub-t}.
Thus, we have $A^{(L)}_{1,1} \geq 1/5$.
In the following, we prove $A^{(L)}_{1,2} + A^{(L)}_{2,2} \geq 1/20$.
Assume for contradiction that 
$A^{(L)}_{1,2} + A^{(L)}_{2,2} < 1/20$.
By \eqref{eq-lem-tight-rc-scaling-error-ub-t} we have
\begin{align*}
t &\geq \norm{\=r\left(A^{(L)}\right)- \=u}_{1} + \norm{\=c\left(A^{(L)}\right)- \=v}_{1} \geq c_1\left(A^{(L)}\right) - \frac{7}{8}.
\end{align*}
Thus, 
\begin{align*}
c_1\left(A^{(L)}\right) \leq \frac{7}{8} + t.
\end{align*}
Hence,
\begin{align*}
&\quad \sum_{i\leq 2,j\leq 2}A^{(L)}_{i,j} = A^{(L)}_{1,2} + A^{(L)}_{2,2} + c_1\left(A^{(L)}\right) \leq \frac{1}{20} + \frac{7}{8} + t.
\end{align*}
Therefore,
\begin{align*}
&\quad\norm{\=r\left(A^{(L)}\right)- \=u}_{1} + \norm{\=c\left(A^{(L)}\right)- \=v}_{1} \geq \norm{\=r\left(A^{(L)}\right)- \=u}_{1} \geq u_1+u_2 - \sum_{i\leq 2}r_i\left(A^{(L)}\right)\\
&\geq 1 - \left(\frac{1}{20} + \frac{7}{8} + t\right) >t,
\end{align*}
which is contradictory with \eqref{eq-lem-tight-rc-scaling-error-ub-t}.
Thus, we have $A^{(L)}_{1,2} + A^{(L)}_{2,2} > \frac{1}{20}$.
In summary, we have proved \eqref{eq-claim-elementrange}. This completes the proof of the theorem.
\end{proof}

\subsection{On the threshold of phase transition.}
In this subsection, we prove \cref{thm-lower-bound-main-rc-threshould-tight}.

Observe that for $\={u} = \={v} = (1/2, 1/2)$, the feasible set with respect to $(\={u}, \={v})$ is non-empty, as it clearly contains the pair $(1/2, 1/2)$. Consequently, \cref{thm-lower-bound-main-rc-threshould-tight} follows as a direct corollary of the following result.

\begin{theorem}
Fix $\varepsilon > 0$ and let $\={u} = \={v} = (1/2, 1/2)$. Let $A$ be an arbitrary $(\={u}, \={v})$-scalable matrix, and let $A^{(0)}, A^{(1)}, \dots$ denote the sequence of matrices generated by the SK algorithm on input $(A, (\={u}, \={v}))$.
Then there exists some $K = O(1/\varepsilon)$ such that
\begin{align*}
\norm{\=r\left(A^{(K)}\right)- \boldsymbol{u}}_{1} + \norm{\=c\left(A^{(K)}\right)- \boldsymbol{v}}_{1} \leq \varepsilon.
\end{align*}
\end{theorem}
\begin{proof}
For simplicity, define 
\[
\forall k\geq 0,\quad \Delta^{(k)} \triangleq\norm{\=r\left(A^{(k)}\right)- \boldsymbol{u}}_1 + \norm{\=c\left(A^{(k)}\right)- \boldsymbol{v}}_1.
\]
If $\Delta^{(0)}\leq\varepsilon$, take $K=0$. Otherwise, $0<\varepsilon<\Delta^{(0)}<1$, since $A$ has no zero columns.
For any even $k$, assume that for some $p,q\in [0,1/2]$,
\begin{align}
A^{(k)}_{1,1} = p, \quad A^{(k)}_{1,2} = \frac{1}{2}- p,\quad A^{(k)}_{2,1} = q, \quad A^{(k)}_{2,2} = \frac{1}{2}- q.
\end{align}
Also assume without loss of generality that $p+q \geq 1/2$.
Then we have
\begin{align}
\Delta^{(k)} = 2(p+ q) - 1.
\end{align}
In addition, one can verify that
\begin{align*}
A^{(k+1)}_{1,1} = \frac{p}{2(p+q)}, \quad A^{(k+1)}_{1,2} = \frac{1- 2p}{4(1 -(p+q))},\quad A^{(k+1)}_{2,1} = \frac{q}{2(p+q)}, \quad A^{(k+1)}_{2,2} = \frac{1- 2q}{4(1 -(p+q))}.
\end{align*}
Thus,
\begin{align}
r_1\left(A^{(k+1)}\right) = \frac{p}{1 + \Delta^{(k)} } +  \frac{1- 2p}{2\left(1 - \Delta^{(k)}\right) }, \quad \quad r_2\left(A^{(k+1)}\right) = \frac{q}{1 + \Delta^{(k)} } +  \frac{1- 2q}{2\left(1 - \Delta^{(k)}\right)}.
\end{align}
Hence,
\begin{align*}
A^{(k+2)}_{1,1} = \frac{p\left(1 - \Delta^{(k)}\right)}{1+(1-4p)\Delta^{(k)}}, \quad \quad A^{(k+2)}_{2,1} = \frac{q\left(1 - \Delta^{(k)}\right)}{1+(1-4q)\Delta^{(k)}}.
\end{align*}
Therefore,
\begin{align*}
A^{(k+2)}_{1,1} + A^{(k+2)}_{2,1} - \frac{1}{2} &=  \frac{p\left(1 - \Delta^{(k)}\right)}{1+(1-4p)\Delta^{(k)}} + \frac{q\left(1 - \Delta^{(k)}\right)}{1+(1-4q)\Delta^{(k)}} - \frac{1}{2}.
\end{align*}
Let $s \triangleq p - q$.
We have $1 - 4p = -\Delta^{(k)} - 2s$, $1 - 4q = -\Delta^{(k)} + 2s$. 
We have 
\begin{align}\label{eq-upperbound-positive-denominator}
\left(1 - \left(\Delta^{(k)}\right)^2\right)^2 - 4s^2\left(\Delta^{(k)}\right)^2 = \left(1+(1-4p)\Delta^{(k)}\right)\left(1+(1-4q)\Delta^{(k)}\right)>0.
\end{align}
Similarly, 
\begin{align}\label{eq-akplustwo-11-akplustwo-21}
A^{(k+2)}_{1,1} + A^{(k+2)}_{2,1} - \frac{1}{2} &= \frac{2s^2\Delta^{(k)}}{\left(1 - \left(\Delta^{(k)}\right)^2\right)^2 - 4s^2\left(\Delta^{(k)}\right)^2}.
\end{align}
Therefore, by 
\[A^{(k+2)}_{1,1} + A^{(k+2)}_{1,2} + A^{(k+2)}_{2,1} + A^{(k+2)}_{2,2} = \frac{1}{2} + \frac{1}{2} = 1.\]
We have 
\[\abs{A^{(k+2)}_{1,1} + A^{(k+2)}_{2,1} - \frac{1}{2}} = \abs{A^{(k+2)}_{1,2} + A^{(k+2)}_{2,2} - \frac{1}{2}}.\]
Hence,
\begin{align*}
\Delta^{(k+2)} &=  \abs{A^{(k+2)}_{1,1} + A^{(k+2)}_{2,1} - \frac{1}{2}} + \abs{A^{(k+2)}_{1,2} + A^{(k+2)}_{2,2} - \frac{1}{2}} = \frac{4s^2\Delta^{(k)}}{\abs{\left(1 - \left(\Delta^{(k)}\right)^2\right)^2 - 4s^2\left(\Delta^{(k)}\right)^2}} \\
&= \frac{4s^2\Delta^{(k)}}{\left(1 - \left(\Delta^{(k)}\right)^2\right)^2 - 4s^2\left(\Delta^{(k)}\right)^2},
\end{align*}
where the last equality is by \eqref{eq-upperbound-positive-denominator}.
In addition, by $p,q\leq 1/2$ we have
\[
\abs{s} = \abs{p - q} \leq \abs{ \frac{1}{2} - \left(p+q - \frac{1}{2}\right)} = 
\abs{ \frac{1}{2} - \left(\frac{1 + \Delta^{(k)}}{2} - \frac{1}{2}\right)} =  \frac{1}{2} \abs{ 1- \Delta^{(k)}}.
\]
Thus, we have
\begin{align*}
\Delta^{(k+2)} &\leq \frac{\left(1 - \Delta^{(k)}\right)^2\Delta^{(k)}}{\left(1 - \left(\Delta^{(k)}\right)^2\right)^2 - \left(1 - \Delta^{(k)}\right)^2\left(\Delta^{(k)}\right)^2} = \frac{\left(1 - \Delta^{(k)}\right)^2\Delta^{(k)}}{\left(1 - \Delta^{(k)}\right)^2\left(\left(1 + \Delta^{(k)}\right)^2 - \left(\Delta^{(k)}\right)^2\right)}\\*
&= \frac{\Delta^{(k)}}{1 + 2\Delta^{(k)}}.
\end{align*}
If an even iterate has zero error, every subsequent iterate also has zero error. For any even $k$ with $\Delta^{(k+2)}>0$, both errors are positive, and we have
\begin{align*}
\frac{1}{\Delta^{(k+2)}} \geq \frac{1}{\Delta^{(k)}} + 2.
\end{align*}
Hence, whenever $\Delta^{(k+2)}>0$,
\begin{align*}
\frac{1}{\Delta^{(k+2)}} \geq  \frac{1}{\Delta^{(0)}} + 2+k.
\end{align*}
Therefore,
\begin{align*}
\Delta^{(k+2)} \leq  \frac{1}{ k+2}.
\end{align*}
By setting $K = 2\lceil 1/(2\varepsilon)\rceil$, we obtain $\Delta^{(K)} \leq \varepsilon$.
The theorem is proved.
\end{proof}

%% file: MissingProofs.tex
% \section{Proof of Fact 2.5}\label{appendix-fact}
% \hktodo{whether fact 2.5 is used in the proof}

% \Cref{item-first-fact-a0-a1} of this lemma is immediate.
% \Cref{item-forth-fact-c-a-a0-a1} of this lemma is immediate by Item \ref{item-second-fact-c-a-a0-a1} and \cref{lem-monotone-rsum-csum}.
% Here, we only prove \cref{item-second-fact-c-a-a0-a1}.
% Since $A$ is $(\gamma,\rho)$-dense,
% we have
% \begin{align}\label{eq-lower-upper-bound-r-c}
% \forall i,j\in [n], \quad \rho\lceil\gamma n\rceil \leq  r_i(A),c_i(A) \leq n.
% \end{align}
% Thus 
% \begin{align*}
% \forall i,j\in [n], \quad 0 \leq A^{(0)}_{i,j} = \frac{A_{i,j}}{r_i(A)} \leq \frac{1}{\rho\lceil\gamma n\rceil} \leq \frac{1}{\rho\gamma n}.
% \end{align*}
% In addition, by $r_i(A)\leq n$ for each $i\in [n]$,
% we have $A^{(0)}_{i,j}\geq \rho/n$ if $A_{i,j} \geq \rho$.
% Combined with the fact that $A$ is $(\gamma,\gamma',\rho)$-dense, we have that every row of $A^{(0)}$ contains at least $\gamma n$ entries that are no less than $\rho/n$, and every column of $A^{(0)}$ contains at least $\gamma' n$ entries that are no less than $\rho/n$.

\section{Proof of Lemma \ref{lem-upper-bound-max-element-initial}.}
To prove this lemma, it suffices to establish \eqref{eq-upper-bound-max-element-older}. The proof of \eqref{eq-upper-bound-max-element-older-column} is analogous.
Assume that $r_t(B) = 1$ for each $t\in [n]$.
Define
\begin{align}\label{eq-def-tau-aprime}
\tau = \max_{i,j\in [n]}A_{i,j},\quad A' = \diag\left(\frac{1}{\tau},\dots,\frac{1}{\tau}\right) \cdot A.
\end{align}
Thus, we have $A'_{i,j} \leq 1$ for each $i,j\in [n]$.
Furthermore, for each $i\in [n]$,
denote
\begin{align*}
\alpha_i \triangleq \abs{c_i(B) - 1},\quad \quad x_i \triangleq \tau\cdot X_{i,i}, \quad \quad y_i \triangleq Y_{i,i}.
\end{align*}
Fix one index $i$ of row and another index $j$ of column.
Then $B_{i,j} = x_i A'_{i,j} y_j$.
Furthermore,
we have 
\begin{align}
\sum_{\ell\in [n]}x_i A'_{i,\ell} y_{\ell} = r_i(B) = 1, \quad \sum_{k\in [n]}x_k A'_{k,j}y_j =c_j(B) \leq 1+\alpha_j.
\end{align}
Hence,
\begin{align}\label{eq-upperbound-xkakj-ylail}
\sum_{\ell \neq j}y_{\ell}A'_{i,\ell} = (1- B_{i,j})/x_i, \quad \sum_{k\neq i}x_k A'_{k,j} = \left(c_j(B) - B_{i,j}\right)/y_j \leq (1+\alpha_j - B_{i,j})/y_j.
\end{align}
Define
\begin{align}\label{eq-definition-R-C}
R \triangleq \{\ell \neq j | A'_{i,\ell} \geq \rho\},\quad C \triangleq \{k \neq i | A'_{k,j} \geq \rho\}.
\end{align}
Since $A$ is $(\gamma,\gamma',\rho)$-dense and \eqref{eq-def-tau-aprime} holds,
we have 
\[\abs{R} \geq \lceil \gamma n \rceil - 1,\quad \abs{C} \geq \lceil \gamma' n \rceil - 1, \quad \abs{C}+\abs{R} \geq (\gamma + \gamma') n - 2.\]
By $r_t(B) = 1$ and $\alpha_t = \abs{c_t(B) - 1}$ for each $t\in [n]$,
we have 
\begin{align*}
\sum_{k\in C}\sum_{\ell\in R}B_{k,\ell} + \sum_{k\in C}\sum_{\ell\not\in R}B_{k,\ell} = \abs{C}, \quad
\sum_{k\in C}\sum_{\ell\in R}B_{k,\ell} + \sum_{k\not\in C}\sum_{\ell\in R}B_{k,\ell} \geq \abs{R}-\sum_{\ell\in R}\alpha_\ell.
\end{align*}
Thus, we have 
\begin{equation}
\begin{aligned}\label{eq-n-largerthan-s1tos4}
n &= \sum_{k\in C}\sum_{\ell\in R}B_{k,\ell} + \sum_{k\in C}\sum_{\ell\not\in R}B_{k,\ell} + \sum_{k\not\in C}\sum_{\ell\in R}B_{k,\ell} + \sum_{k\not\in C}\sum_{\ell\not \in R}B_{k,\ell}\\
&\geq \left(\sum_{k\in C}\sum_{\ell\in R}B_{k,\ell} + \sum_{k\in C}\sum_{\ell\not\in R}B_{k,\ell}\right) + \left(\sum_{k\in C}\sum_{\ell\in R}B_{k,\ell}+\sum_{k\not\in C}\sum_{\ell\in R}B_{k,\ell}\right) +  B_{i,j} - \sum_{k\in C}\sum_{\ell\in R}B_{k,\ell} \\
&\geq  B_{i,j}+\abs{C} + \abs{R} -\sum_{\ell\in R}\alpha_\ell - \sum_{k\in C}\sum_{\ell\in R}B_{k,\ell}.
\end{aligned}
\end{equation}
In addition, 
\begin{equation}\label{eq-upperbound-s1b}
\begin{aligned}
&\quad \sum_{k\in C}\sum_{\ell\in R}B_{k,\ell} 
\\\left(\text{ by $B_{k,\ell} = x_kA'_{k,\ell}y_{\ell}$}\right)\quad &= \sum_{k\in C}\sum_{\ell \in R}x_kA'_{k,\ell}y_{\ell}
\\(\text{ by $A'_{k,\ell} \leq 1$})\quad      &\leq \left(\sum_{k\in C}x_k\right) \left(\sum_{\ell \in R}y_{\ell}\right) \\
(\text{ by \eqref{eq-definition-R-C}})\quad &\leq \frac{1}{\rho^2}\cdot \left(\sum_{k\in C}x_kA'_{k,j}\right) \left(\sum_{\ell \in R}y_{\ell}A'_{i,\ell}\right)\\
(\text{ by \eqref{eq-definition-R-C}})\quad       &\leq \frac{1}{\rho^2}\cdot\left(\sum_{k\neq i}x_kA'_{k,j}\right) \left(\sum_{\ell \neq j}y_{\ell}A'_{i,\ell}\right)\\
(\text{ by \eqref{eq-upperbound-xkakj-ylail}})\quad       &\leq \frac{(c_j(B) - B_{i,j})}{\rho y_j}\cdot\frac{(1-B_{i,j})}{\rho x_i}. 
\end{aligned}
\end{equation}
Recall that $\abs{C}+\abs{R} \geq (\gamma + \gamma') n - 2$.
Combined with \eqref{eq-n-largerthan-s1tos4} and \eqref{eq-upperbound-s1b}, we have
\begin{align*}
n 
\geq B_{i,j} + (\gamma + \gamma') n - 2-\sum_{\ell\in R}\alpha_\ell - \frac{(c_j(B) - B_{i,j})}{\rho y_j}\cdot\frac{(1-B_{i,j})}{\rho x_i}.
\end{align*}
Thus, we have
\begin{equation}\label{eq-nbij-geq-tgamman-bij}
\begin{aligned}
&\quad n B_{i,j}\\
&\geq B^2_{i,j} + \left((\gamma + \gamma') n - 2-\sum_{\ell\in R}\alpha_\ell\right)B_{i,j} -  \frac{(c_j(B) - B_{i,j})(1 - B_{i,j}) \cdot B_{i,j}}{\rho^2 x_iy_j}\\
(\text{by $B_{i,j} = x_iA'_{i,j}y_{j}\leq x_iy_j$})\quad &\geq B^2_{i,j} + \left((\gamma + \gamma') n - 2-\sum_{\ell\in R}\alpha_\ell\right)B_{i,j} -\frac{(c_j(B) - B_{i,j})(1 - B_{i,j})}{\rho^2} \\
  &= \left((\gamma + \gamma') n - 2 + \frac{1+c_j(B)}{\rho^2}-\sum_{\ell\in R}\alpha_\ell \right)B_{i,j} - \frac{c_j(B)}{\rho^2} + \frac{(\rho^2-1)B^2_{i,j}}{\rho^2}\\
\left(\text{by $\rho\leq 1$}\right)\quad &\geq \left((\gamma + \gamma') n + (c_j(B) - 1) -\sum_{\ell\in R}\alpha_\ell \right)B_{i,j} - \frac{c_j(B)}{\rho^2}.
\end{aligned}
\end{equation}
In addition, by 
\eqref{eq-def-alpha} we have
\[1 -c_j(B) + \sum_{\ell\in R}\alpha_\ell \leq \alpha_j + \sum_{\ell\in R}\alpha_\ell \leq \sum_{\ell\in [n]}\alpha_\ell \leq n\alpha(B).\]
Combined with $\gamma + \gamma' - 1- \alpha(B)>0$ and $\rho\in (0,1]$, we have 
\[(\gamma + \gamma') n + (c_j(B) - 1) -\sum_{\ell\in R}\alpha_\ell - n \geq (\gamma + \gamma') n -  n\alpha(B) - n \geq (\gamma + \gamma' - 1 - \alpha(B))n>0.\]
Combined with \eqref{eq-nbij-geq-tgamman-bij}, we have
\begin{align*}
B_{i,j} \leq  \frac{c_j(B)}{\rho^2((\gamma + \gamma') n + (c_j(B) - 1) -\sum_{\ell\in R}\alpha_\ell - n ) } \leq \frac{c_j(B)}{\rho^2(\gamma + \gamma' - 1 - \alpha(B))n }.
\end{align*}
This establishes \eqref{eq-upper-bound-max-element-older}, thus completing the proof of the lemma.

\section{Proof of Lemma \ref{lem-lower-bound-elements}.}
Assume for contradiction that there exists some $k,\ell\in [n]$ where $A_{k,\ell}\geq \rho/n$ and $B_{k,\ell} \leq \theta/n$.
Let $X \triangleq \diag(x_1,\dots,x_n)$ and $Y \triangleq \diag(y_1,\dots,y_n)$.
Thus, $B_{k,\ell} = x_kA_{k,\ell}y_{\ell}$.
Combined with $A_{k,\ell}\geq \rho/n$ and $B_{k,\ell} \leq \theta/n$,
we have $x_{k}y_{\ell} \leq \theta/\rho$.
Since the factorization 
$B=XAY$ is invariant under the rescaling 
$(X,Y)\rightarrow (cX,Y/c)$ for any 
$c>0$, there is one degree of freedom in the choice of the diagonal scalings. Hence, without loss of generality we may rescale so that 
$x_{k} = y_{\ell}$, which together with 
$x_{k}y_{\ell} \leq \theta/\rho$ implies
$x_{k} = y_{\ell} \leq \sqrt{\theta/\rho}$.

By $y_{\ell}\leq \sqrt{\theta/\rho}$, $A_{i,\ell} \leq a/n$ for each $i\in [n]$,  and
\[
\sum_{i\in [n]}B_{i,\ell} = \sum_{i\in [n]}x_iA_{i,\ell}y_{\ell}\geq d,\]
we have there exists some $i\in [n]$ where 
\begin{align}\label{eq-app-upper-bound-yj-rhogammathetaa}
x_{i}\geq \frac{d}{a}\sqrt{\frac{\rho}{\theta}}.
\end{align}
Let $C = \{j\in [n]| A_{i,j}\geq \rho/n\}$.
We have
\begin{align}\label{eq-app-upper-bound-bij}
\forall j\in C, \quad  B_{i,j} = x_iA_{i,j}y_j\leq \frac{b}{n}.
\end{align}
Thus, for each $j\in C$, we have 
\begin{equation}
\begin{aligned}\label{eq-app-xi-small}
&\quad y_j\\
(\text{by \eqref{eq-app-upper-bound-yj-rhogammathetaa} and \eqref{eq-app-upper-bound-bij}})\quad &\leq \frac{ab}{d\rho} \cdot \sqrt{\frac{\theta}{\rho}} \\
(\text{by \eqref{eq-def-theta-constant}})\quad & < \sqrt{\frac{((\gamma+\gamma')(1 -\alpha(B)) - 1)}{a}}.
\end{aligned}
\end{equation}
Similarly, 
by $x_{k}\leq \sqrt{\theta/\rho}$, $A_{k,j} \leq a/n$ for each $j\in [n]$,  and
\[
\sum_{j\in [n]}B_{k,j} = \sum_{j\in [n]}x_kA_{k,j}y_{j}\geq d,\]
we have there exists some $j\in [n]$ where 
\begin{align}\label{eq-app-upper-bound-yj-rhogammathetaa-new}
y_{j}\geq \frac{d}{a}\sqrt{\frac{\rho}{\theta}}.
\end{align}
Let $R = \{t\in [n]| A_{t,j}\geq \rho/n\}$.
We have
\begin{align}\label{eq-app-upper-bound-bij-new}
\forall t\in R, \quad  B_{t,j} = x_tA_{t,j}y_j\leq \frac{b}{n}.
\end{align}
Thus, for each $t\in R$, we have 
\begin{equation}
\begin{aligned}\label{eq-app-xi-small-new}
&\quad x_t\\
(\text{by \eqref{eq-app-upper-bound-yj-rhogammathetaa-new} and \eqref{eq-app-upper-bound-bij-new}})\quad &\leq \frac{ab}{d\rho} \cdot \sqrt{\frac{\theta}{\rho}} \\
(\text{by \eqref{eq-def-theta-constant}})\quad & < \sqrt{\frac{((\gamma+\gamma')(1 -\alpha(B)) - 1)}{a}}.
\end{aligned}
\end{equation}
Combining \eqref{eq-app-xi-small}, \eqref{eq-app-xi-small-new} with $A_{t,j}\leq a/n$ for each $t,j\in [n]$,
we have 
\begin{equation}\label{eq-app-upper-bound-rc}
\begin{aligned}
&\quad \sum_{t\in R}\sum_{j\in C}x_{t}A_{t,j}y_{j}\\
&< n^2\cdot \frac{a}{n}\cdot \frac{((\gamma+\gamma')(1 -\alpha(B)) - 1)}{a}\\
\quad &= ((\gamma+\gamma')(1 -\alpha(B)) - 1)n.
\end{aligned}
\end{equation}

Without loss of generality, assume that $c_t(B) = 1$ for each $t\in [n]$.
The case where $r_t(B) = 1$ for each $t\in [n]$ can be proved analogously.
By \eqref{eq-def-alpha},
we have 
\[(\abs{C}+\abs{R})\alpha(B)\geq (\gamma + \gamma')n\alpha(B) \geq n\alpha(B) \geq \sum_{r\in [n]}\abs{\left(\sum_{t\in [n]}B_{r,t}\right) -1} \geq \sum_{r\in R}\abs{\left(\sum_{t\in [n]}B_{r,t}\right) -1} \geq \abs{R} - \sum_{r\in R}\sum_{t\in [n]}B_{r,t}.\]
Thus, 
we have
\[\sum_{r\in R}\sum_{t\in [n]}B_{r,t} \geq \abs{R} -\alpha(B)\cdot(\abs{C}+\abs{R}).\]
In addition, 
\[\sum_{r\in [n]}\sum_{t\in [n]}B_{r,t} = \sum_{t\in [n]}\sum_{r\in [n]}B_{r,t}  = \sum_{t\in [n]}c_t(B) = n .\]
Hence, 
\[\sum_{r\not\in R}\sum_{t\in [n]}B_{r,t} \leq n - \abs{R} +\alpha(B)\cdot(\abs{C}+\abs{R}).\]
Moreover,
\[\sum_{r\in [n]}\sum_{t\not\in C}B_{r,t} =\sum_{t\not\in C}c_t(B) = n - \abs{C}.\]
Thus, we have 
\begin{align*}
\sum_{r\in R}\sum_{t\in C} B_{r,t} \geq n - \sum_{r\not\in R}\sum_{t\in [n]}B_{r,t} - \sum_{r\in [n]}\sum_{t\not\in C}B_{r,t} \geq n - (2n - (1 -\alpha(B))(\abs{R} + \abs{C})).
\end{align*}
Combined with $\abs{R} + \abs{C} \geq (\gamma'+\gamma) n$, we have
\begin{align*}
\sum_{r\in R}\sum_{t\in C} B_{r,t}\geq n - (2n - (1 -\alpha(B))\cdot (\gamma + \gamma')n) \geq ((\gamma+\gamma')(1 -\alpha(B)) - 1)n.
\end{align*}
This is contradictory with \eqref{eq-app-upper-bound-rc}.
The lemma is proved.

\section{Proof of Lemma \ref{lem-lowerbound-reduction-rctooneone-new}.}

By \eqref{eq-reduction-block-partitioned-matrix}, we have 
\begin{align}\label{eq-typlusnminustqminusone}
ty+(n-t)q-1 = \frac{s\,t\,(1-d)\,(n-t)\Bigl(t- d(n-t)\lambda^2-\bigl(t-s+d(s+t-n)\bigr)\lambda\Bigr)}
{\Bigl(nt+(n-t)\lambda\bigl(s+d(n-s)\bigr)\Bigr)\Bigl(t\bigl((n-s)+ds\bigr)+d\lambda\,n(n-t)\Bigr)}.
\end{align}
In addition, by \eqref{eq-def-lambda-s1-s2} we have 
\begin{align*}
\lambda \leq \frac{S_1+S_2}{S_2}.
\end{align*}
Combined with \eqref{eq-condition-c-upperbound}, we have
\begin{align*}
d\lambda (n-t)\leq \frac{(s-t)(n-t)}{6(n-t+\abs{s+t -n})} \leq \frac{s-t}{4}.
\end{align*}
Also by \eqref{eq-condition-c-upperbound}, we have
\begin{align*}
d(s+t - n) < \frac{s-t}{4}.
\end{align*}
Thus,
\[
d(n-t)\lambda^2+\bigl(t-s+d(s+t-n)\bigr)\lambda = \lambda\left(d(n-t)\lambda+\bigl(t-s+d(s+t-n)\bigr)\right)< \frac{(t-s)\lambda}{2} < 0.
\]
Hence, 
\begin{align*}
ty+(n-t)q-1 &> \frac{s\,t\,(1-d)\,(n-t)\,(2t+(s-t)\lambda)}
{2\Bigl(nt+(n-t)\lambda\bigl(s+d(n-s)\bigr)\Bigr)\Bigl(t\bigl((n-s)+ds\bigr)+d\lambda\,n(n-t)\Bigr)}.
\end{align*}
Moreover, by \eqref{eq-condition-c-upperbound} we have 
\begin{align*}
d \leq \frac{s-t}{6(n-t)} <\frac{1}{6}.
\end{align*}
Combined with \eqref{eq-def-lambda-s1-s2}, we have $\lambda > 1$ and $d\lambda < 1$.
Thus,
\begin{equation}\label{eq-tyminusnminustqminusone-positive}
\begin{aligned}
ty+(n-t)q-1 &> \frac{s\,t\,(1-d)\,(n-t)\,(2t+(s-t)\lambda)}
{2\Bigl(nt+(n-t)\lambda\bigl(s+d(n-s)\bigr)\Bigr)\Bigl(t\bigl((n-s)+ds\bigr)+d\lambda\,n(n-t)\Bigr)}\\
\left(\text{by }d<\frac{1}{6}\right)\quad &\geq \frac{s\,t\,\,(n-t)\,(2t+(s-t)\lambda)}
{4\Bigl(nt+(n-t)\lambda\bigl(s+n-s)\bigr)\Bigr)\Bigl(t\bigl(n-s+s\bigr)+\,n(n-t)\Bigr)}\\
& = \frac{st(n-t)(2t+(s-t)\lambda)}
{4n^3(t+(n-t)\lambda)} > \frac{st(n-t)(s-t)\lambda}
{4n^3(n-t)\lambda} = \frac{st(s-t)}
{4n^3} > 0.
\end{aligned}
\end{equation}

In addition, by \eqref{eq-typlusnminustqminusone} and $\lambda>1$, we have
\begin{equation}
\begin{aligned}\label{eq-ty-nminustqminusone-simplifiedupperbound}
ty+(n-t)q-1 
&< \frac{s\,t\,(1-d)\,(n-t)\Bigl(\bigl(s-t-ds\bigr)\lambda+t\Bigr)}
{\Bigl(nt+(n-t)\lambda\bigl(s+d(n-s)\bigr)\Bigr)\Bigl(t\bigl((n-s)+ds\bigr)+d\lambda\,n(n-t)\Bigr)}\\
\left(\text{by }d<\frac{1}{6}\right)\quad &\leq \frac{s\,t\,\,(n-t)\Bigl(\bigl(s-t-ds\bigr)\lambda+t\Bigr)}
{\Bigl(nt+(n-t)\lambda\bigl(s+d(n-s)\bigr)\Bigr)\Bigl(t\bigl((n-s)+ds\bigr)+d\lambda\,n(n-t)\Bigr)}.
\end{aligned}
\end{equation}
Furthermore, we have
\begin{align*}
\frac{\bigl(s-t-ds\bigr)\lambda+t}
{nt+(n-t)\lambda\bigl(s+d(n-s)\bigr)}\leq \frac{\bigl(s-t\bigr)\lambda+t}
{nt+(n-t)\lambda s} \leq \max\left\{\frac{s-t}
{(n-t)s},\frac{1}{n} \right\} = \frac{1}{n}.
\end{align*}
Hence, 
\begin{align*}
\frac{s\,t\,\,(n-t)\Bigl(\bigl(s-t-ds\bigr)\lambda+t\Bigr)}
{\Bigl(nt+(n-t)\lambda\bigl(s+d(n-s)\bigr)\Bigr)\Bigl(t\bigl(n-s+ds\bigr)+d\lambda\,n(n-t)\Bigr)}
&\leq \frac{s\,t\,\,(n-t)}
{t\bigl(n-s\bigr)} \cdot \frac{\bigl(s-t-ds\bigr)\lambda+t}
{\Bigl(nt+(n-t)\lambda\bigl(s+d(n-s)\bigr)\Bigr)}\\
& \leq \frac{s\,(n-t)}
{n\bigl(n-s\bigr)}.
\end{align*}
Combined with \eqref{eq-ty-nminustqminusone-simplifiedupperbound}, we have
\begin{align*}
ty+(n-t)q-1 < \frac{s\,(n-t)}
{n\bigl(n-s\bigr)}.
\end{align*}
Combined with \eqref{eq-tyminusnminustqminusone-positive},
we have
\begin{align}\label{eq-tyminusnminustqminusone-positive-upperbound}
\frac{st(s-t)}
{4n^3}< ty+(n-t)q-1 < \frac{s\,(n-t)}
{n\bigl(n-s\bigr)}.
\end{align}
Moreover, by \eqref{eq-reduction-block-partitioned-matrix} we have
\begin{align*}
\forall j> s, \quad  \sum_{i\leq n}A^{(2)}_{i,j} = ty+(n-t)q.
\end{align*}
Thus, \eqref{eq-lem-lowerbound-reduction-rctooneone-jges} is immediate by \eqref{eq-tyminusnminustqminusone-positive-upperbound}.

In the next, we prove \eqref{eq-lem-lowerbound-reduction-rctooneone-jleqs}.
By \eqref{eq-reduction-block-partitioned-matrix} we have 
\begin{align*}
n =\sum_{i\in [n]}\sum_{j\in [n]}A^{(2)}_{i,j} = \sum_{j>s}\sum_{i\leq n}A^{(2)}_{i,j} + \sum_{j\leq s}\sum_{i\leq n}A^{(2)}_{i,j} =s(tx+(n-t)z) + (ty+(n-t)q)(n-s).
\end{align*}
Hence, 
\begin{align*}
tx+(n-t)z - 1 = -\frac{(ty+(n-t)q - 1)(n-s)}{s} 
\end{align*}
Combined with \eqref{eq-tyminusnminustqminusone-positive-upperbound},
we have 
\begin{align}\label{eq-txnminustzminusone-bound}
\frac{t-n}{n}< tx+(n-t)z-1 < 0.
\end{align}
Moreover, by \eqref{eq-reduction-block-partitioned-matrix} we have
\begin{align*}
\forall j\leq s, \quad  \sum_{i\leq n}A^{(2)}_{i,j} = tx+(n-t)z.
\end{align*}
Thus, \eqref{eq-lem-lowerbound-reduction-rctooneone-jleqs} is immediate by \eqref{eq-txnminustzminusone-bound}.

Then we prove \eqref{eq-z-minimumelement-bound}.
Recall that $\lambda > 1$.
Combined with \eqref{eq-reduction-block-partitioned-matrix}, we have
\begin{align*}
z = \frac{d\bigl(t+\lambda(n-t)\bigr)}
{\,t\bigl(n-s+ds\bigr)+d\lambda n(n-t)\,} < \frac{d\bigl(t\lambda+\lambda(n-t)\bigr)}
{\,t\bigl(n-s\bigr)\,} = \frac{dn\lambda}{t(n-s)}.
\end{align*}
Thus, \eqref{eq-z-minimumelement-bound} is immediate.

At last, we prove \eqref{eq-y-keyelement-bound}.
By \eqref{eq-reduction-block-partitioned-matrix}, we have
\begin{align*}
y = \frac{1}{n}\cdot \frac{nt+dn\lambda(n-t)}
{\,nt+\lambda(n-t)\bigl(s+d(n-s)\bigr)\,} = \frac{1}{n}\cdot 
\frac{nt+dn\lambda(n-t)}
{nt+dn\lambda(n-t) + \lambda(n-t)s(1-d)}.
\end{align*}
In addition, by $d<1/6$ and $\lambda > 1$, we have
\begin{align*}
\lambda(n-t)s(1-d) > \frac{5s(n-t)\lambda n^2}{6n^2}  = \frac{5s(n-t)}{6n^2} \cdot \lambda n^2 \geq \frac{5s(n-t)}{6n^2} \cdot (nt+dn\lambda(n-t)).
\end{align*}
Combining the above two inequalities, we have
\begin{align*}
y <\frac{1}{n}\cdot 
\frac{6n^2}
{6n^2 + 5s(n-t)} = \frac{6n}
{6n^2 + 5s(n-t)}.
\end{align*}
Thus, \eqref{eq-y-keyelement-bound} is proved. 
The lemma is proved.

%% file: references.tex
% Generated with the unmodified siamplain.bst.
% Publication year corrected from the SIAM record: https://doi.org/10.1137/1.9781611978971.157